\documentclass[onecolumn,letterpaper,accepted=2020-08-31]{quantumarticle}
\pdfoutput=1
\usepackage{authblk}

\title{Circuit optimization of Hamiltonian simulation by\newline simultaneous
  diagonalization of Pauli clusters}
\author{Ewout van den Berg and Kristan Temme}
\affil{IBM Quantum, IBM T.J.~Watson Research Center, Yorktown Heights,
  NY, USA}

\usepackage[left=2cm,top=2cm,right=2cm,nohead]{geometry}
\usepackage{amsmath,amssymb,amsthm}
\usepackage{color}
\usepackage{graphicx}
\usepackage{multirow}

\usepackage[numbers,sort&compress]{natbib}
\newtheorem{theorem}{Theorem}[section]

\usepackage{bbold} % \mathbb{1}
\usepackage[normalem]{ulem} % strikethrough (\sout)

\usepackage{colortbl}
\definecolor{darkblue}{rgb}{0.0,0.0,0.3}
\definecolor{highlight}{rgb}{0.90,0.92,1.0}
\usepackage{ulem}
\usepackage{algorithm}
\usepackage{algpseudocode}
\usepackage{qcircuit}

\newcommand{\ket}[1]{\ensuremath{\vert{#1}\rangle}}
\newcommand{\bra}[1]{\ensuremath{\langle{#1}\vert}}

\newcommand{\Pauli}[1]{{\sc{\MakeLowercase{#1}}}}
\newcommand{\cnot}[0]{{\sc{cnot}}}

\newcommand{\cz}[0]{{\sc{cz}}}

\newtheorem{innercustomgeneric}{\customgenericname}
\providecommand{\customgenericname}{}
\newcommand{\newcustomtheorem}[2]{%
  \newenvironment{#1}[1]
  {%
   \renewcommand\customgenericname{#2}%
   \renewcommand\theinnercustomgeneric{##1}%
   \innercustomgeneric
  }
  {\endinnercustomgeneric}
}

\newcustomtheorem{LabelTheorem}{Theorem}

\begin{document}

\maketitle

\begin{abstract}
  Many applications of practical interest rely on time evolution of
  Hamiltonians that are given by a sum of Pauli operators. Quantum
  circuits for exact time evolution of single Pauli operators are well
  known, and can be extended trivially to sums of commuting Paulis by
  concatenating the circuits of individual terms. In this paper we
  reduce the circuit complexity of Hamiltonian simulation by
  partitioning the Pauli operators into mutually commuting clusters
  and exponentiating the elements within each cluster after applying
  simultaneous diagonalization. We provide a practical algorithm for
  partitioning sets of Paulis into commuting subsets, and show that
  the proposed approach can help to significantly reduce both the
  number of \cnot{} operations and circuit depth for Hamiltonians
  arising in quantum chemistry. The algorithms for simultaneous
  diagonalization are also applicable in the context of stabilizer
  states; in particular we provide novel four- and five-stage
  representations, each containing only a single stage of conditional
  gates.
\end{abstract}

\section{Introduction}

Simulation of quantum systems by means of Hamiltonian time evolution
is an important application of quantum
% TODO\cite{}
computers~\cite{FEY1982a,LLO1996a} . The time evolution of a
Hamiltonian $H$ is given by $e^{itH}$, and the main challenge is to
generate an efficient circuit that implements or closely approximates
this time-evolution operator. Given the prominent position of
Hamiltonian time evolution in quantum computing, it should come as no
surprise that this area has been well studied, and that different
approaches have been developed, including those based on, for
instance, product formulas~\cite{TRO1959a,SUZ1991a}, quantum
walks~\cite{BER2012Ca}, linear combinations of
unitaries~\cite{CHI2012Wa}, truncated Taylor
series~\cite{BER2015CCKa}, and quantum signal
processing~\cite{LOW2017Ca} (see~\cite{CHI2018MNRa} for a good
overview).
Product formulas are applicate when, as is often the case, the
Hamiltonian can be decomposed as the sum $H = \sum_j H_j$, such that
the time evolution of each of the terms $H_j$ is readily
evaluated. Through successive application of the terms with
appropriately chosen time steps, it is then possible to simulate the
original Hamiltonian. For instance, using the Lie-Trotter product
formula~\cite{TRO1959a} we have that
\[
e^{itH} = \lim_{k\to\infty}\left(\textstyle\prod_j e^{i(t/k)H_j}\right)^k,
\]
whereas in the non-asymptotic regime, the Trotter scheme provides a
first-order approximation, with the norm of the difference between the
exact and approximate time evolution operators scaling as
$\mathcal{O}(t^2/k)$. More advanced higher-order schemes, such as
those by Suzuki~\cite{SUZ1991a}, are also available, and are analyzed
for example in~\cite{CHI2018MNRa}. The approximation errors arising in
the use of product formulas are ultimately caused by non-commuting
terms in the Hamiltonian.  Indeed, given any set of mutually commuting
operators $P_1$ through $P_m$, the exponent of the sum is equal to
products of the individual exponents, provided that the time slices
for each operator add up to $t$. As a simple example, it holds that
\begin{equation}\label{Eq:ExpP}
e^{it\sum_{j=1}^mP_j} = \prod_{j=1}^m e^{itP_j},
\end{equation}
whenever the operators commute. A natural idea, therefore, is to
partition the operators into mutually commuting subsets. This can be
done by applying graph coloring~\cite{BOL2013a} to a graph whose nodes
correspond to the operators and whose edges connecting nodes for which
the associated operators do not commute. The resulting coloring is
such that all nodes sharing the same color commute. Time evolution for
the sum of nodes within each subset is then trivial, and product
formulas can be applied to the sum of Hamiltonians formed as the sum
of each subset.
This approach is especially applicable in scenarios where the
Hamiltonian is expressed as a sum of Pauli operators, for which the
commutativity relations are easily evaluated. This situation arises by
definition in spin simulation of magnetic systems using the Heisenberg
model. In other applications, such as the quantum simulation of
fermionic systems, the terms in the Hamiltonian can be mapped to Pauli
operators using for example the Jordan-Wigner or Bravyi-Kitaev
transformation~\cite{BRA2002Ka,JOR1928Wa,TRA2018LMVa}.

In this paper we focus on quantum circuits for evaluating the product
of commuting exponentials, appearing on the right-hand side of
equation~\eqref{Eq:ExpP}. We also consider the partitioning of terms,
and application of the proposed methods to quantum chemistry.  Given
the limited qubit connectivity in near-term architectures, we largely
focus on reducing the number of \cnot\ gates, since these may
translate into large numbers of swap gates. For systems that use
error-correction codes, it may be important to reduce other gates,
such as the $T$-gate. These gates only appear in the exponentiation of
the diagonalized operators, and these parts of the circuit can be
independently simplified using techniques such as those described
in~\cite{AMY2014MMa,AMY2019Ma}. We further note that clustering of
Pauli operators and simultaneous diagonalization of commuting
operators also arises in variational quantum
eigensolvers~\cite{KAN2017MTTa,VER2020YIa,JEN2019GMa,YEN2020VIa,GOK2019ADGa,CRA2019SWPa,BON2019BBa}.
In that context, however, the techniques are used for an altogether
different purpose; namely, to reduce the number of measurements to
estimate inner-products of the initial state with different Pauli
operators. The schemes we develop for simultaneous diagonalization and
partitioning are also applicable in the context of variational quantum
eigensolvers.

The paper is organized as follows. In Section~\ref{Sec:Existing} we
review the basic circuit for exponentiation of individual Pauli
operators, and how these can be combined. Section~\ref{Sec:Proposed}
describes the proposed approach based on simultaneous
diagonalization. Synthesis and optimization of circuits for
diagonalization are studied in Section~\ref{Sec:Circuits}. In
Section~\ref{Sec:Experiments} we perform numerical experiments to
obtain the circuit complexity for simulating random Paulis and
Hamiltonians arising in quantum chemistry. Conclusions are given in
Section~\ref{Sec:Conclusions}.

\paragraph{Notation}
We denote the Pauli matrices by $\sigma_x$, $\sigma_y$, and
$\sigma_z$, and write $\sigma_i$ for the two-by-two identity matrix.
The tensor product of $n$ Pauli matrices gives an $n$-Pauli operators,
which we denote by the corresponding string of characters, for example
$\text{\Pauli{ZXI}} = \sigma_z\otimes\sigma_x\otimes\sigma_i$.
We write $[n] = \{1,\ldots,n\}$ and denote the binary group by $\mathbb{F}_2$.

\section{Direct exponentiation of Pauli operators}\label{Sec:Existing}

Given a Hermitian operator $M$ with eigendecomposition
$M = Q\Lambda Q^{\dag} =\sum_k \lambda_k\ket{q_k}\bra{q_k}$, it holds
that exponentiation of the matrix is equivalent to exponentiation of
the individual eigenvalues; that is,
\[
e^{i\theta M} = Qe^{i\theta\Lambda} Q^{\dag} = \sum_k e^{i\theta \lambda_k}\ket{q_k}\bra{q_k}.
\]
Alternatively, we can look at operators $D = Q^{\dag}$ that
diagonalize $M$, that is $DMD^{\dag} = \Lambda$. The identity and
Pauli $\sigma_z$ matrices are already diagonal, and therefore have a
trivial diagonalization with $D=I$. From this it follows directly that
\[
e^{i\theta \sigma_i} = e^{i\theta}I,\quad\mbox{and}\quad
e^{i\theta \sigma_z} = \left[\begin{array}{cc} e^{i\theta} & 0\\0 & e^{-i\theta}\end{array}\right]
=: R_z(\theta)
\]
The remaining two Pauli operators $\sigma_x$ and $\sigma_y$ can be
diagonalized to $\Lambda = \sigma_z$ with operators
\[
D_x = H = \frac{1}{\sqrt{2}}\left[\begin{array}{rr}1&1\\1&-1\end{array}\right],
\quad \mbox{and}\quad
D_y = HSX = \frac{1}{\sqrt{2}}\left(\begin{array}{cc}i & 1\\-i & 1\end{array}\right)
\quad\mbox{where}\quad
S = \left(\begin{array}{cc}1&0\\0&i\end{array}\right).\quad
\]
It then follows that
$
e^{i\theta \sigma_x} = e^{i\theta D_x^{\dag}\sigma_z D_x} =
D_x^{\dag}e^{i\theta \sigma_z}D_x = D_x^{\dag}R_z(\theta)D_x,
$
and likewise for $\sigma_y$. A direct way to exponentiate a Pauli
matrix is to first apply the appropriate diagonalization operator $D$,
followed by the rotation $R_z(\theta)$, and finally the adjoint
diagonalization operator $D^{\dag}$.

In order to exponentiate general $n$-Pauli operators we first
diagonalize the matrix, which is done by applying the tensor product
of the diagonalization operators corresponding to each of the
terms. The resulting diagonal is the tensor product of $\sigma_i$ and
$\sigma_z$ matrices; a $\sigma_i$ for each \Pauli{I} term, and
$\sigma_z$ for each of the \Pauli{X},\Pauli{Y}, or \Pauli{Z} terms.
For a given element in the computational basis we can determine the
sign induced by the $\sigma_z$ diagonal terms and maintain the overall
sign in an ancilla qubit using \cnot{} operators. The rotation
operator $R_z(\theta)$ is then applied to the ancilla to achieve the
exponentiation of the eigenvalue (see also \cite[Chapter
4]{NIE2010Ca}). We then uncompute the ancilla by reapplying the \cnot\
gates, and complete the procedure by applying the adjoint
diagonalization operator. An example for the successive
exponentiation of Pauli operators \Pauli{IXX}, \Pauli{ZYZ},
\Pauli{XXI} with angles $\theta_1$,$\theta_2$, and $\theta_3$, is
shown in Figure~\ref{Fig:IndividualExponentiation}(a). Several remarks
are in place here. First, in the diagonalization of $\sigma_y$ we
include a {\sc{not}} operator (X) to ensure diagonalization to
$\sigma_z$ rather than $-\sigma_z$. In practice this term can be
omitted, and for each occurrence of a $\sigma_y$ term we can simply
multiply the corresponding rotation angle $\theta$ by $-1$.  Second,
it is often the case that time evolution needs to be done as a
conditional circuit. Instead of making each gate conditional it
suffices to merely make the $R_z$ gates conditional. Third, for sets
of commuting Pauli operators it is possible to obtain circuits with
reduced complexity by rearranging the order in which the Pauli
operators are applied in such as way that as many gates as possible
cancel. In the example shown in
Figure~\ref{Fig:IndividualExponentiation}(b) we rearrange the blocks
and apply simple gate cancellation to adjacent pairs of identical
diagonalization operations and \cnot\ gates.

\begin{figure}
\centering
\begin{tabular}{c}
\includegraphics[width=0.99\textwidth]{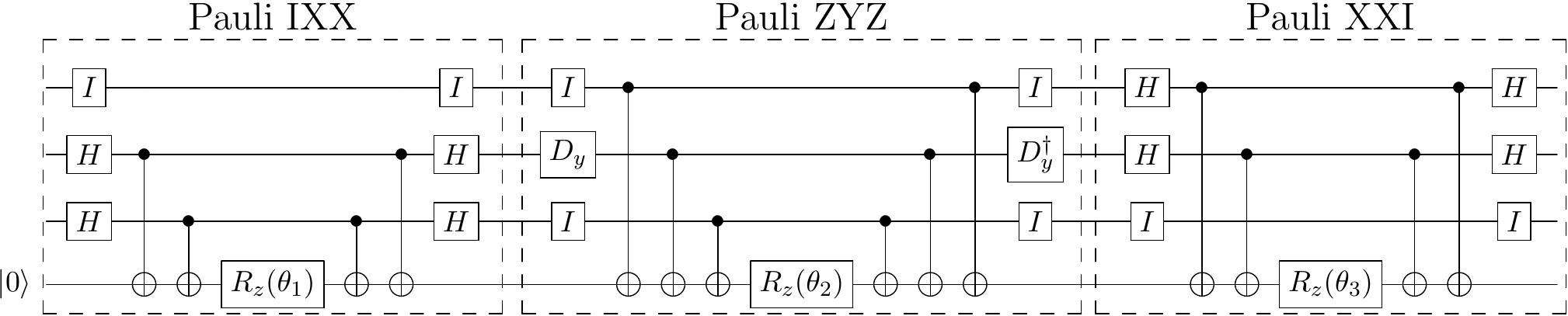}\\
\\[-8pt]
  ({\bf{a}}) Basic circuit\\[12pt]
\includegraphics[width=0.9\textwidth]{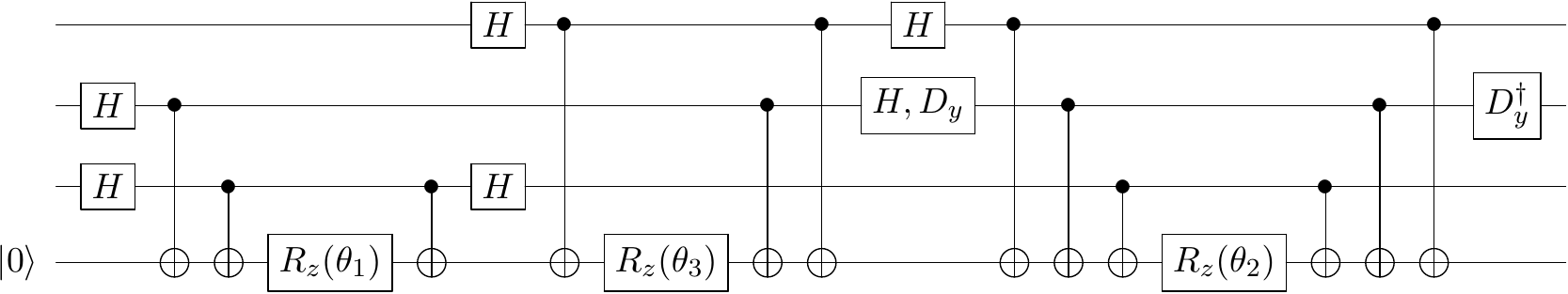}\\
\\[-8pt]
  ({\bf{b}}) Circuit after permuting blocks and applying gate cancellations
\end{tabular}
\caption{Individual exponentiation of the terms in the group of
  the Pauli operators. The top panel shows the basic circuit, the
  bottom panel gives the optimized circuit obtained by reordering the
Paulis and canceling unitaries where possible.}\label{Fig:IndividualExponentiation}
\end{figure}

\section{Proposed approach}\label{Sec:Proposed}

It is well known for any set of mutually commuting operators there
exists a unitary $\mathcal{U}$ that simultaneously diagonalizes each
of the operators in the set~\cite[Thm. 1.3.19]{HOR1985Ja}. Applying
this to a set of commuting $n$-Pauli operations $\{P_j\}_{j=1}^m$, we
know that there exists a unitary
$\mathcal{U}\in\mathbb{C}^{2^n\times 2^n}$, such that
$\mathcal{U}P_j\mathcal{U}^\dag = \Lambda_j$ is diagonal for all
$i\in[m]$. Moreover, not only are the resulting operators diagonal,
they are in fact Pauli operators themselves, consisting only of
$\sigma_i$ and $\sigma_z$ terms along with a sign. As an example we
apply the techniques we develop in Section~\ref{Sec:Circuits} to the
three commuting Paulis used in
Figure~\ref{Fig:IndividualExponentiation}. The resulting circuits that
each diagonalize all three Paulis, along with the resulting diagonals
are shown in Figure~\ref{Fig:CircuitU}.

\begin{figure}[h!]
\centering
\begin{tabular}{ccc}
\begin{minipage}{5cm}
% ========= Generated with generate_figure2a.py =========
\hspace*{-3pt}
\Qcircuit @C=1.0em @R=0.2em @!R {
& \gate{H} & \targ & \qw & \qw & \qw & \qw \\
& \gate{H} & \qw & \targ & \qw & \qw & \qw \\
& \gate{H} & \ctrl{-2} & \ctrl{-1} & \gate{S} & \gate{H} & \qw\\
}
\end{minipage}&
\begin{minipage}{5cm}
% ========= Generated with generate_figure2b.py =========
\hspace*{-3pt}
\Qcircuit @C=1.0em @R=0.2em @!R {
& \gate{S} & \targ & \qw & \gate{S} & \gate{H} & \qw \\
& \gate{S} & \qw & \targ & \gate{S} & \gate{H} & \qw \\
& \gate{H} & \ctrl{-2} & \ctrl{-1} & \gate{S} & \gate{H} & \qw \\
}
\end{minipage}&
\begin{minipage}{5cm}
% ========= Generated with generate_figure2c.py =========
\hspace*{-3pt}
\Qcircuit @C=1.0em @R=0.2em @!R {
& \ctrl{1} &  \qw & \qw & \gate{H} & \qw \\
& \targ   & \ctrl{1} & \gate{S} & \gate{H} & \qw  \\
& \gate{H} & \targ & \gate{S} & \gate{H}  & \qw \\
	 }
\end{minipage}
\\
\\[-3pt]
({\bf{a}}) \Pauli{izi}, \Pauli{izz}, \Pauli{zzi} &    % Algorithm CZ
({\bf{b}}) \Pauli{-izi}, \Pauli{ziz}, \Pauli{zzi} &  % Algorithm CNot
({\bf{c}}) \Pauli{-izz}, \Pauli{-izi}, \Pauli{zii}    % Algorithm Greedy2
\end{tabular}
\caption{Circuits for simultaneous diagonalization of \Pauli{IXX}, \Pauli{ZYZ}, and
\Pauli{XXI}, along with the resulting Paulis.}\label{Fig:CircuitU}
\end{figure}
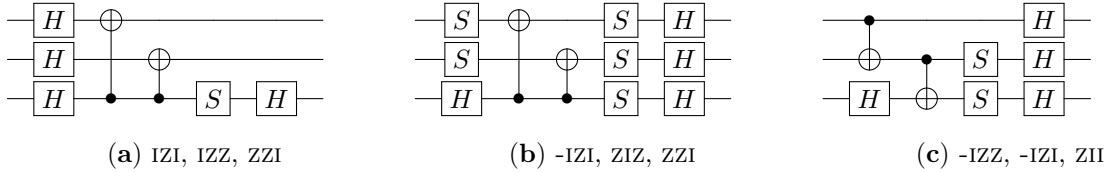

The advantage of simultaneous diagonalization become apparent when
looking at the exponentiation of the sum of commuting Paulis. 
From~\eqref{Eq:ExpP} we know that this is equal to the product of
individual exponents. Additionally using diagonalization then gives
\[
e^{i \sum_{j=1}^m \theta_jP_j} =
\prod_{j=1}^m e^{i\theta_jP_j} =
\prod_{j=1}^m\left( \mathcal{U}^{\dag}e^{i\theta_j\Lambda_j}\mathcal{U}\right) =
{\textstyle\mathcal{U}^{\dag}\left(\prod_{j=1}^m e^{i\theta_j\Lambda_j}\right)\mathcal{U}}.
\]
The last equality follows from the fact that successive
$\mathcal{U}\mathcal{U}^{\dag}$ terms cancel, thereby allowing us to
apply the diagonalization operator and its adjoin only once, instead
of once for each individual term. Since we know how to exponentiate
the diagonal Paulis, we can put everything together to obtain the
circuit shown in Figure~\ref{Fig:Circuit3}(a). If needed, the sign of
the diagonalized terms can be incorporated in the rotation
angle. Similar to the original approach we can often simplify the
circuit by canceling adjacent gates that multiply to identity. A key
advantage of diagonalization then is that, aside from the $R_z$ gates,
each term consists entirely of \cnot\ gates. This provides much more
room for optimization, since instead of having to match four terms
(\Pauli{I}, \Pauli{X}, \Pauli{Y}, and \Pauli{Z}), we only need to
consider two (\Pauli{I} and \Pauli{Z}). This makes is easier to find
orderings of the terms that reduce the number of \cnot\ gates in the
circuit. The circuit after simplification can be seen in
Figure~\ref{Fig:Circuit3}(b).  Compared to
Figure~\ref{Fig:IndividualExponentiation}(b), which requires twelve
\cnot\ gates, this example uses only a total of ten \cnot\ operation:
six for exponentiation and a further four for the diagonalization
circuit and its adjoint.

In practice it is unlikely that all terms in a Hamiltonian
commute. For these cases we first need to partition the terms into
subsets of commuting operators. For each of these subsets we can then
apply simultaneous diagonalization for simulating that part of the
Hamiltonian. When the number of terms in a subset is small we may find
that exponentiation using the original approach gives a better
circuit. We can use this to our advantage by simply choosing the best
method for each subset.

\begin{figure}
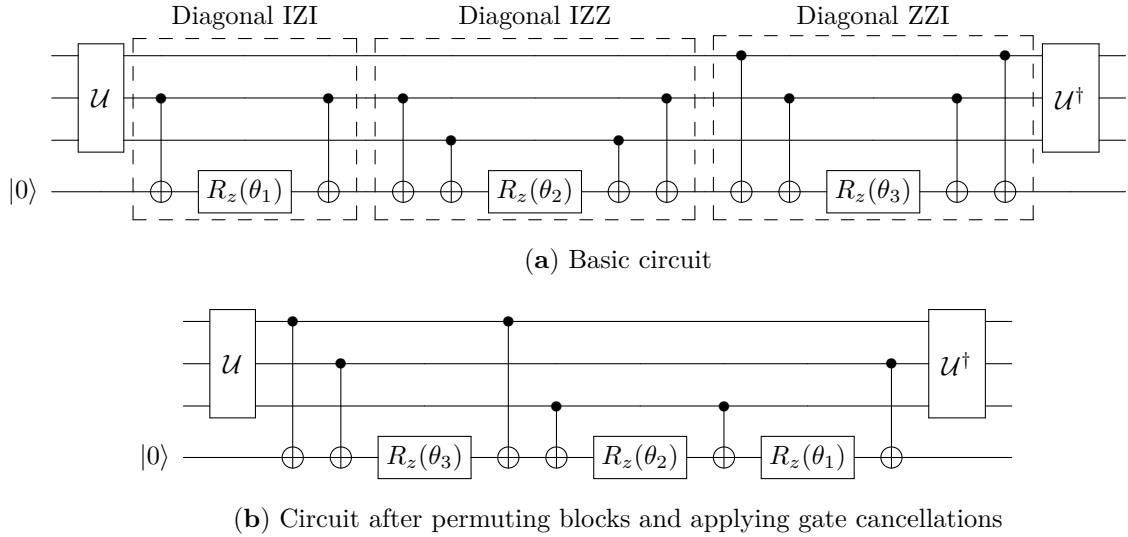

\centering
\begin{tabular}{c}
\input{fig/FigCircuit3a.tex}\\
\\[-2pt]
({\bf{a}}) Basic circuit\\[12pt]
\input{fig/FigCircuit3b.tex}\\
\\[-2pt]
({\bf{b}}) Circuit after permuting blocks and applying gate cancellations
\end{tabular}
\caption{Circuit for exponentiation of Paulis using simultaneous
  diagonalization operator $\mathcal{U}$. The top panel shows the basic
  circuit, the bottom panel gives the optimized circuit obtained by
  reordering the blocks and canceling adjacent \cnot\ gates.}\label{Fig:Circuit3}
\end{figure}

\section{Circuits for simultaneous diagonalization}\label{Sec:Circuits}

In this section we consider the construction of circuits that
simultaneously diagonalize a given set of commuting Pauli
operators. This is conveniently done using the tableau representation
originally used to simulate stabilizer
circuits~\cite{AAR2004Ga,GOT1998a} and reviewed next. The schemes for
simultaneous diagonalization presented
in~\cite{GOK2019ADGa,CRA2019SWPa} use the same techniques, but differ
from ours in the number and type of stages.

\subsection{Tableau representation and operations}

The tableau representation is a binary array in which each row
represents a single $n$-Pauli operator. The columns of the tableau are
partitioned as $[X,Z,s]$,  such that $(X_{i,j},Z_{i,j})$ represents the
$j$th component of the $i$th Pauli operator. The value is $(1,0)$ for
\Pauli{X}, $(0,1)$ for \Pauli{Z}, $(1,1)$ for \Pauli{Y}, and $(0,0)$
for \Pauli{I}. Entries in $s$ are set if the corresponding Pauli operator
has a negative sign. For instance:
\[
\left[\begin{array}{c|c|c}
1001 & 0101 & 0\\
0110 & 1101 & 1
\end{array}\right] =
\left[\begin{array}{r}
\mbox{\Pauli{XZIY}}\\
\mbox{-\Pauli{ZYXZ}}
\end{array}\right]
\]
To keep the exposition clear, we do not show the sign column in the
tableaus for the remainder of the paper. It is of crucial importance,
however, that the appropriate signs are maintained, as they eventually
appear in the exponent of the Paulis. Once the tableau is set up we
can apply different operations. The first two operations, illustrated
in Figures~\ref{Fig:SimpleOperations}(a)
and~\ref{Fig:SimpleOperations}(b), change the order of respectively
the Pauli operators and the qubits. A third operation, shown in
Figures~\ref{Fig:SimpleOperations}(c) sweeps one row with
another. This operation corresponds to multiplication of the
operators, which results in the given entries in the X and Z blocks to
be added modulo two. The sign update is more involved and we refer
to~\cite{AAR2004Ga} for details.  Even though these operations alter
the tableau they do not generate any corresponding gates in the
circuit.

\begin{figure}[h!]
\centering
\begin{tabular}{ccc}
\includegraphics[width=0.3\textwidth]{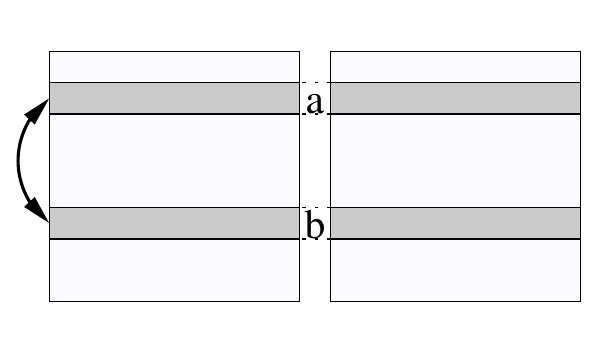} &
\includegraphics[width=0.3\textwidth]{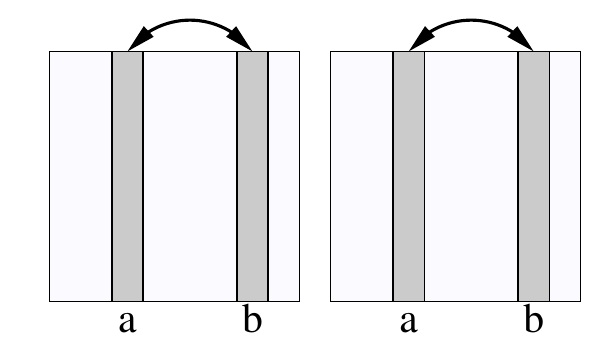} &
\includegraphics[width=0.3\textwidth]{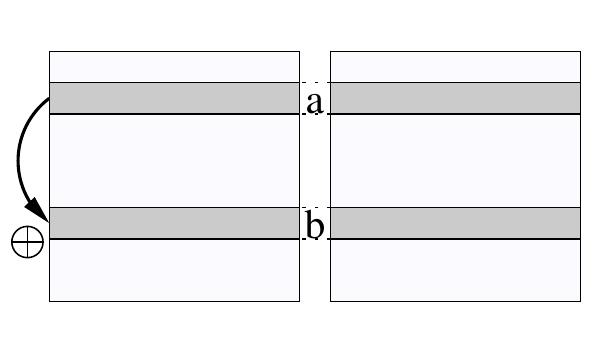}\\
({\bf{a}}) Swap rows $a$ and $b$ &
({\bf{b}}) Swap columns $a$ and $b$ &
({\bf{c}}) Sweep row $b$ with row $a$
\end{tabular}
\caption{Graphical illustration of tableau operations that do not
  generate any gates in the circuit. The column-swap operations
  changes the logical order of the qubits.}\label{Fig:SimpleOperations}
\end{figure}

In addition to these basic operations, we can apply operators from the
Clifford group. Operators $C$ in this group are unitary and have the
property that $\mathcal{C}P\mathcal{C}^{\dag}$ is a Pauli operator for
any Pauli operator $P$. The Clifford group can be generated by three
gates: the Hadamard gate (H), the phase gate (S), and the
conditional-{\sc{not}} (\cnot) gates. The definition of these gates
along with the respective update rules and effect on the tableau are
summarized in Figure~\ref{Fig:GateOperations}. The Hadamard gate
applied to a column (qubit) results in the exchange of the
corresponding columns in the X and Z blocks. The phase gate adds a
given column in the X block to the matching column in the Z block,
along with appropriate updates to the signs as shown in the
figure. The conditional-{\sc{not}} operation $CX(a,b)$, that is,
negation of qubit $b$ conditional on qubit $a$, has the effect of
adding column $a$ to column $b$ in block X, and adding column $b$ to
column $a$ in block Z. From the basic three operations we can form
another convenient gate, the conditional-Z gate (see
also~\cite{GAR2013MCa}). Denoted $CZ(a,b)$, this gate is equivalent to
successively applying $H(b)$, $CX(a,b)$, and $H(b)$, and has the
effect of adding columns $a$ and $b$ of block X to columns $a$ and $b$
of block Z, respectively.

\begin{figure}[t!]
\centering
\begin{tabular}{lcp{5.3cm}p{5.3cm}}
\hline
Operator & matrix & sign update & block update\\[1pt]
\hline
\\[-8pt]
$H(a)$
& {\scriptsize$\frac{1}{\sqrt{2}}\left(\begin{array}{cc}1&1\\1&-1\end{array}\right)$}
&  --
& $\mathrm{swap}(x_a,z_a)$
\\[10pt]
$S(a)$
& {\scriptsize$\left(\begin{array}{cc}1&0\\0&i\end{array}\right)$}
& $s \oplus x_a\otimes z_a$
& $z_a = z_a \oplus x_a$
\\[10pt]
$CX(a,b)$
& {\scriptsize$\left(\begin{array}{cccc}1&0&0&0\\0&1&0&0\\0&0&0&1\\0&0&1&0\end{array}\right)$}
& $s \oplus (x_a\otimes z_b\otimes (x_b\oplus z_a\oplus 1))$
& $z_a = z_a \oplus z_b$, $x_b = x_b \oplus x_a$
\\[20pt]
$CZ(a,b)$
&{\scriptsize$\left(\begin{array}{cccc}1&0&0&0\\0&1&0&0\\0&0&1&0\\0&0&0&-1\end{array}\right)$}
& $s \oplus (x_a\otimes x_b\otimes (z_a\oplus z_b\oplus 1))$
& $z_a = z_a \oplus x_b$, $z_b = z_b \oplus x_a$
\\[18pt]
\hline
\\[-10pt]
\end{tabular}
\begin{tabular}{cccc}
\includegraphics[width=0.22\textwidth]{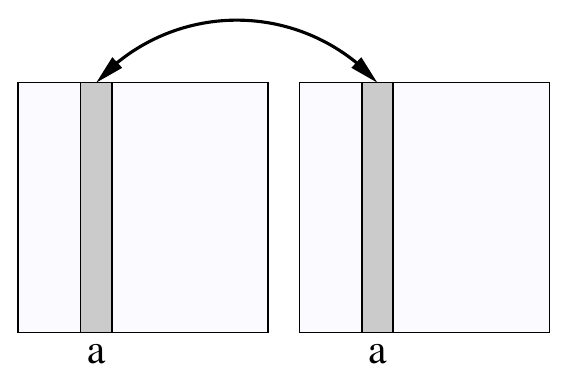} &
\includegraphics[width=0.22\textwidth]{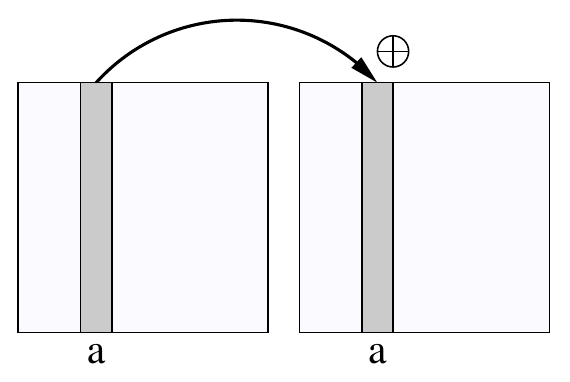} &
\includegraphics[width=0.22\textwidth]{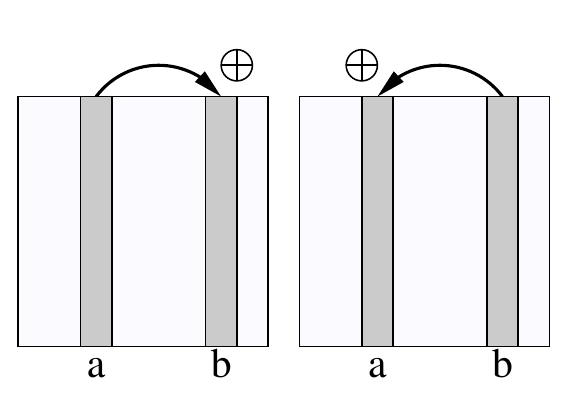}&
\includegraphics[width=0.22\textwidth]{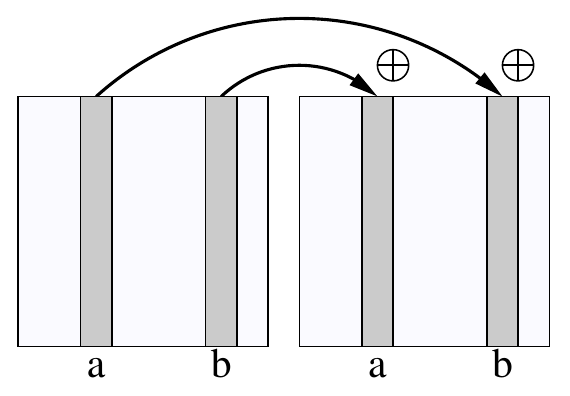}\\
$H(a)$& $S(a)$& $CX(a,b)$ & $CZ(a,b)$
\end{tabular}
\caption{Table and graphical illustration summarizing the effect
  different operations have on the tableau. Columns $a$ and $b$ of $X$
  and $Z$ are indicated by $x_a, x_b, z_a, z_b$, respectively. The
  sign vector $s$ is updated prior to the block updates.}\label{Fig:GateOperations}
\end{figure}

\subsection{Simultaneous diagonalization}

For simultaneous diagonalization we initialize the tableau with our
commuting set of Paulis. We then need to apply the different tableau
operations in such a way that the entries in the X block of the final
tableau are all zero. In our algorithms we use row swaps and row sweep
operations. Even though these operations do not generate any gates for
the circuit, they do alter the tableau, and the underlying Pauli
operations. In order to obtain the appropriate diagonalization of the
original Pauli operations we can do one of two things. First, since
these operations commute with the Clifford operations we can apply the
inverse operations of all row operations at the end. Second, we can
work with a parallel tableau on which only the Clifford operations are
applied. The desired diagonalized Pauli operators are then represented
by the final tableau. We now look at several algorithm that clear the
X block. In this we occasionally need to use the rank of the tableau,
which we define as the rank of the $[X,Z]$ matrix.

\subsection{Diagonalizing the X block}\label{Sec:DiagonalizeX}

For the simultaneous diagonalization process we proceed in phases. In
the first phase we manipulate the tableau such that only the entries
on the diagonal of the X block are nonzero. More precisely, let $r$ be
the rank of the matrix $[X,Z]$, then we would like the first $r$
diagonal elements of the X block to be one, and all remaining elements
of the block to be zero. The algorithm we use for this is given in
Algorithm~\ref{Alg:DiagonalizeX}. At the beginning of the algorithm we
are given a tableau corresponding to commuting Paulis. At this point
there is no clear structure, and the tableau therefore looks something
like Figure~\ref{Fig:DiagonalizeX}(a), where gray indicate possibly
nonzero entries (although we illustrate the procedure on a tableau
with $m > n$, the process applies equally to tableaus with other
shapes). In
steps~\ref{Line:DiagX_Begin_Stage1}--\ref{Line:DiagX_End_Stage1} we
iteratively diagonalize the X block. Starting at $k=1$ we first look
for a nonzero element in rows and columns of the X block with indices
at least $k$. If found, we move the one entry to location $(k,k)$ by
applying appropriate row and column swaps, sweep all other nonzero
entries in the new column, increment $k$, and continue. If no such
item we could be found we are done with the first stage and have a
tableau of the form illustrated in
Figure~\ref{Fig:DiagonalizeX}(b). In
steps~\ref{Line:DiagX_Begin_Stage2}--\ref{Line:DiagX_End_Stage2}, we
then repeat the same process on the Z block, starting off at the
current $k$. The tableau at the end of the second stage would look
like Figure~\ref{Fig:DiagonalizeX}(c). In the third stage, given by
steps~\ref{Line:DiagX_Begin_Stage3}--\ref{Line:DiagX_End_Stage3}, we
apply Hadamard gates to swap the diagonalized columns in the Z block
with the corresponding columns in the X block, resulting the a tableau
as shown in Figure~\ref{Fig:DiagonalizeX}(d).  If the rank $r$ is less
than $n$, there may be spurious nonzero elements to the right of the
diagonal block in X. These are swept using \cnot\ operations in
steps~\ref{Line:DiagX_Begin_Stage4}--\ref{Line:DiagX_End_Stage4}.  The
resulting tableau after the final fourth stage is depicted in
Figure~\ref{Fig:DiagonalizeX}(e).

Recalling that the tableau has rank $r$, it is immediate by
construction that any row in X with index exceeding $r$ will be
zero. It therefore follows immediately that the Paulis associated with
these rows contain only \Pauli{I} and \Pauli{Z} terms. The Pauli
string for rows $i$ with $i \leq k$ consist of all \Pauli{I} and
\Pauli{Z} terms, except for an \Pauli{X} or \Pauli{Y} term at location
$i$. We now show that rows $i$ in Z with $i > r$ are also all
zero. This certainly holds for column indices $j > k$, and we
therefore assume that we have $Z[i,j] =1$ with $i > r$ and $j \leq k$.
The terms in the Pauli operators for rows $i$ and $j$ commute at all
indices except $j$, where row $i$ has \Pauli{Z} and row $j$ has
\Pauli{X} or \Pauli{Y}. The Pauli operations therefore anticommute,
which contradicts our assumption that the Paulis in the tableau
commute, and it therefore follows that rows $i > r$ in Z are all
zero. Now, note that the \cnot\ operations in the fourth stage and the
Hadamard operations in the third stage, did not affect the values in
the bottom-left block of Z. We conclude that these values must
therefore already have been zero at the end of stage two, as shown in
Figure~\ref{Fig:DiagonalizeX}(f). The following result is a direct
consequence of the above discussion (if needed tableaus can always be
augmented to make them full rank, see for example~\cite{CRA2019SWPa}):

\begin{theorem}
  The X block of any tableau corresponding to commuting $n$-Paulis
  with rank $n$ can be diagonalized using only Hadamard gates.
\end{theorem}

The fourth stage of the algorithm for diagonalizing X is applicable
whenever the rank of the tableau is less than $n$. In the
implementation given in Algorithm~\ref{Alg:DiagonalizeX} we clear the
spurious entries using \cnot\ operations. There are several ways in
which this stage could be improved. We could determine, for instance,
if the corresponding column in Z has fewer nonzero entries. If that
were the case, we could swap the column using a Hadamard operation and
sweep the alternative column instead. Likewise, it would be possible
to see if sweeping the Z column with that of X using a phase gate,
followed by a swap would be more efficient. In both these cases the
number of \cnot\ operations would be reduced at the cost of
single-qubit operations. If two columns in the residual column block
are similar, one could be simplified by sweeping with the other using
a \cnot\ operation. Further optimization is possible using a
combination of these techniques.

\begin{figure}
\centering
\setlength{\tabcolsep}{12pt}
\begin{tabular}{ccc}
\includegraphics[width=0.24\textwidth]{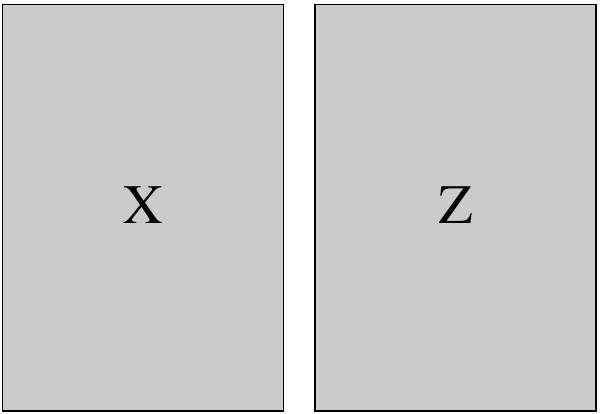} &
\includegraphics[width=0.24\textwidth]{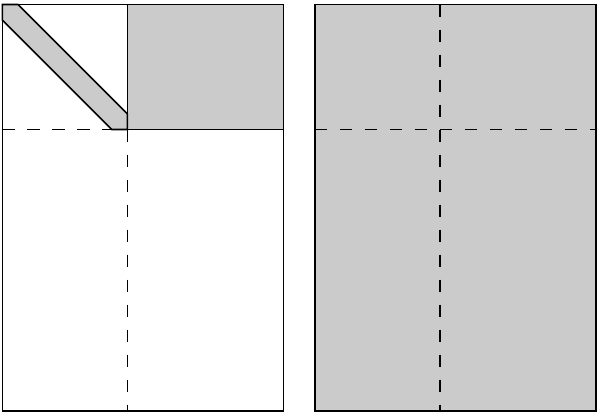} &
\includegraphics[width=0.24\textwidth]{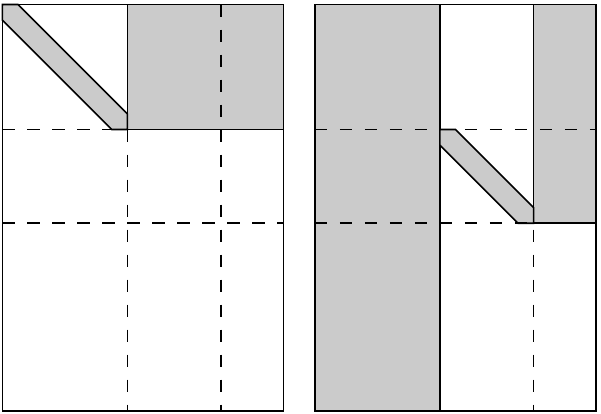} \\
({\bf{a}}) & ({\bf{b}}) & ({\bf{c}})\\[8pt]
\includegraphics[width=0.24\textwidth]{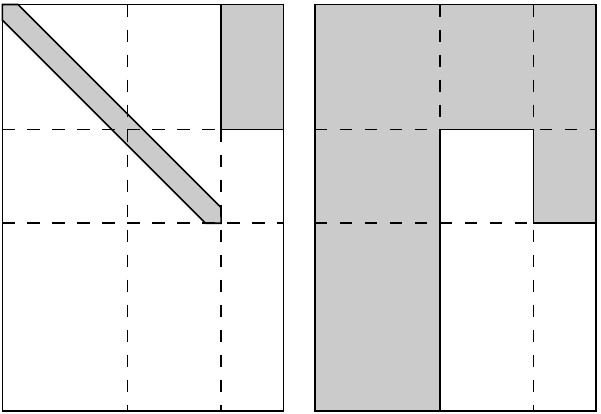} &
\includegraphics[width=0.24\textwidth]{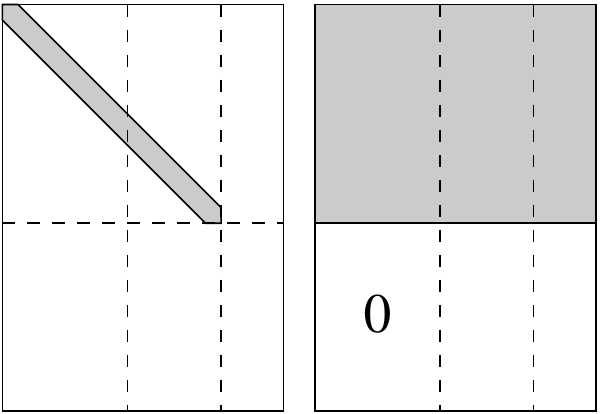} &
\includegraphics[width=0.24\textwidth]{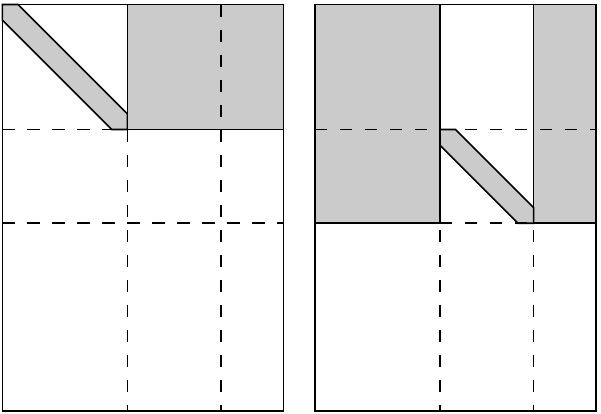} \\
({\bf{d}}) & ({\bf{e}}) & ({\bf{f}})
\end{tabular}
\caption{Diagonalization of the X block with (a) the initial tableau;
  and the situation after (b) partial diagonalization within the X block; 
  (c) continued diagonalization in Z; (d) combination of the diagonal
  parts through application of Hadamard operations; and (e) the final
  result after sweeping the top-right segment of X. Plot (f) shows the
actual zero pattern after diagonalizing part of Z.}\label{Fig:DiagonalizeX}
\end{figure}

\begin{algorithm}[t!]
\caption{Diagonalization of the X block.}\label{Alg:DiagonalizeX}
\hspace*{0pt}{\vspace*{-3pt}}\\
\hspace*{0pt}{{\bf Input}: the input to this function is a tableau
  $T=[X,Z,S]$ of size $m\times 2n+1$, consisting of the X and Z blocks,
  as well as a sign vector $S$. We use the convention that indexing of
  the X or Z blocks corresponds to indexing the tableau at the
  corresponding location. Swapping or sweeping rows applies to the
  entire tableau. Swapping columns $i$ and $j$ means swapping these
  columns in
  both the $X$ and $Z$ blocks.}\\
\hspace*{0pt}{{\bf Output}: updated tableau with off-diagonal entries
  in the X block set to zero.}\\
\hspace*{0pt}{{\bf Complexity}: $\mathcal{O}(n^2\max(m,n))$}\\[-12pt]
\begin{algorithmic}[1]
\State{$k\gets 1$}
\Repeat\label{Line:DiagX_Begin_Stage1}
   \State{Search for index $(i,j)$ with $k \leq i \leq m$ and $k \leq j \leq
   n$ such that $X[i,j] = 1$}
   \If{(index found)}
      \State{Swap rows $i$ and $k$; swap columns $j$ and $k$}
      \For{$i \in [m]$ such that $i\neq k$ and $X[i,k]=1$}
         \State{Sweep row $i$ with row $k$}
      \EndFor
      \State{$k\gets k+1$}
   \EndIf
\Until{(index could not be found)}\label{Line:DiagX_End_Stage1}
\State{$k_x \gets k$}
\Repeat\label{Line:DiagX_Begin_Stage2}
   \State{Search for index $(i,j)$ with $k \leq i \leq m$ and $k \leq j \leq
   n$ such that $Z[i,j] = 1$}
   \If{(index found)}
      \State{Swap rows $i$ and $k$; swap columns $j$ and $k$}
      \For{$i \in [m]$ such that $i\neq k$ and $Z[i,k]=1$}
         \State{Sweep row $i$ with row $k$}
      \EndFor
      \State{$k\gets k+1$}
   \EndIf
\Until{(index could not be found)}\label{Line:DiagX_End_Stage2}
\For{$j \in \{k_x,\ldots,k-1\}$}\label{Line:DiagX_Begin_Stage3}
      \State{Apply gate H$(j)$}
\EndFor\label{Line:DiagX_End_Stage3}
\For{$i \in \{1,\ldots,k\}$ and $j\in\{k,\ldots,n\}$ such that
  $X[i,j]=1$}\label{Line:DiagX_Begin_Stage4}
    \State{Apply gate CNOT$(i,j)$}
\EndFor\label{Line:DiagX_End_Stage4}
\end{algorithmic}
\end{algorithm}

\subsection{Updating Z and clearing X}

After diagonalizing the X block, we need to update the Z block, such
that all nonzero columns in X are matched with a zero or identical
column in Z. Application of combinations of Hadamard and phase gates
then allows us to zero out X and obtain the circuit for simultaneous
diagonalization. In this section we consider three algorithms to
achieve this.

\subsubsection{Pairwise elimination}\label{Sec:PairwiseElimination}

Application of the controlled-Z operation on qubits $a$ and $b$ is
equivalent to successively applying H$(b)$, \cnot$(a,b)$, and
H$(b)$. The overall effect, as illustrated in
Figure~\ref{Fig:GateOperations}, is the sweeping of columns $a$
and $b$ in Z with respectively columns $b$ and $a$ of X. This
operation can therefore simultaneously eliminate $Z[a,b]$ and
$Z[b,a]$ whenever both elements are one. The following result shows
that and off-diagonal one is matched by the reflected element:

\begin{theorem}\label{Thm:SymmetricZ}
Given a tableau $T$ corresponding to a set of commuting Paulis of rank
$k$, and apply the diagonalization procedure. Then the top-left $k$-by-$k$
sub-block of the resulting $Z$ is symmetric.
\end{theorem}
\begin{proof}
  Consider any pair of distinct indices $i,j \in [k]$, and denote the
  string representation of the corresponding Pauli operators of the
  updated tableau $T$ by $P_i$ and $P_j$. The operations performed
  during diagonalization preserve commutativity, and $P_i$ and $P_j$
  therefore commute. For commutativity, we can focus on the symbols at
  locations $i$ and $j$; all others are either $\sigma_i$ or
  $\sigma_z$. It can be verified that symbols $P_i[i]$ and $P_j[i]$ commute iff
  $Z[j,i] = 0$. Likewise, symbols $P_i[j]$ and $P_j[j]$
  commute iff $Z[i,j] = 0$. It follows that in order for the Pauli
  operators to commute, we must have $Z[i,j] = Z[j,i]$. The result
  follows by the fact that indices $i$ and $j$ were arbitrary.
\end{proof}

\noindent
With this, the algorithm for updating the Z block simply reduces to
eliminating the lower-triangular entries in Z (the corresponding
upper-triangular entries will be eliminated simultaneously). This
process is summarized in
lines~\ref{Line:PairwiseZ_Begin_Stage1}--\ref{Line:PairwiseZ_End_Stage1}
of Algorithm~\ref{Alg:PairwiseUpdateZ}. After this first step we are
ready to clear the X block using single-qubit gates, by considering
the values of the diagonal entries in Z. This is done in
lines~\ref{Line:PairwiseZ_Begin_Stage2}--\ref{Line:PairwiseZ_End_Stage2}
of the algorithm. One notable benefit of the algorithm is that the
elimination process only affects the targeted entries, which means
that there is no fill-in. Together with the diagonalization of X in
Section~\ref{Sec:DiagonalizeX}, we obtain a classical complexity of
$\mathcal{O}(n^2\max(m,n))$, along with the following result:

\begin{theorem}
  Given a tableau for commuting $n$-Paulis
  with rank $n$. We can diagonalize the operators using
  H-CZ-S-H stages with $\mathcal{O}(n^2)$ CZ gates.
\end{theorem}

\noindent Since the application of the CZ gates do not affect the
diagonal entries in the Z block, it is possible to apply the phase
gates first and obtain an {\it H-S-CZ-H} scheme. Note that it is
always possible to obtain a full-rank tableau by adding commuting
Paulis that were not in the original span. The resulting
diagonalization then has the stages as given above, and clearly
applies to the original set of Paulis as well. Doing so may however
come at the cost of an increased circuit complexity.

\begin{algorithm}
\caption{Pairwise update of Z, clear X.}\label{Alg:PairwiseUpdateZ}
\hspace*{0pt}{\vspace*{-3pt}}\\
\hspace*{0pt}{{\bf Input}: Tableau $T$ with diagonal $X$ of rank $k$.}\\
\hspace*{0pt}{{\bf Output}: Updated tableau with X block entries set to zero}\\
\hspace*{0pt}{{\bf Complexity}: $\mathcal{O}(k^3)$}
\begin{algorithmic}[1]
\For{$i \in \{2,\ldots,k\}$}\label{Line:PairwiseZ_Begin_Stage1}
   \For{$j \in \{1,\ldots,i-1\}$}
      \State{Apply CZ$(i,j)$ if $Z[i,j] = 1$}
   \EndFor
\EndFor\label{Line:PairwiseZ_End_Stage1}
\For{$i \in\{1,\ldots,k\}$}\label{Line:PairwiseZ_Begin_Stage2}
   \State{Apply S$(i)$ if $Z[i,i] = 1$}
   \State{Apply H$(i)$}
\EndFor\label{Line:PairwiseZ_End_Stage2}
\end{algorithmic}
\end{algorithm}

\subsubsection{Elimination using \cnot\ operations}\label{Sec:CNotDiagonalization}

\begin{figure}[b!]
\hspace*{-11pt}
\begin{tabular}{ccc}
\includegraphics[width=0.32\textwidth]{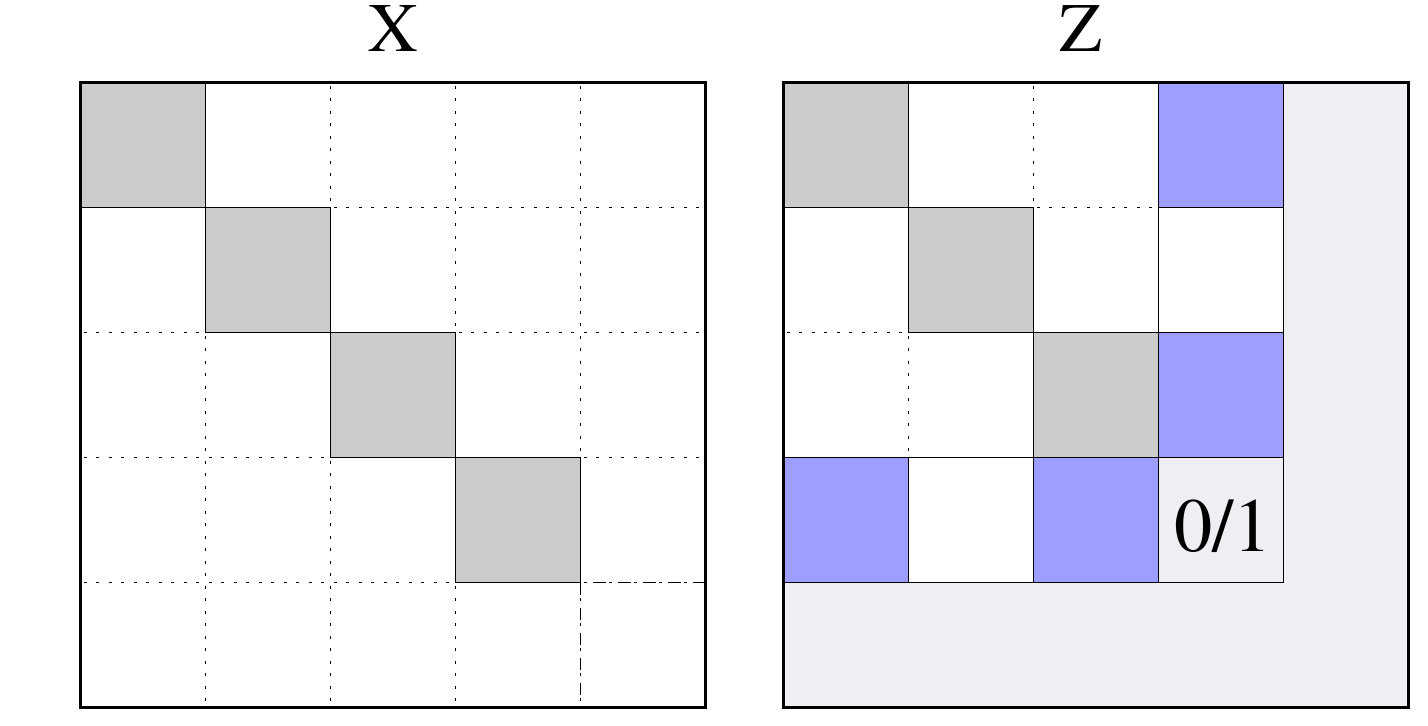}&
\includegraphics[width=0.32\textwidth]{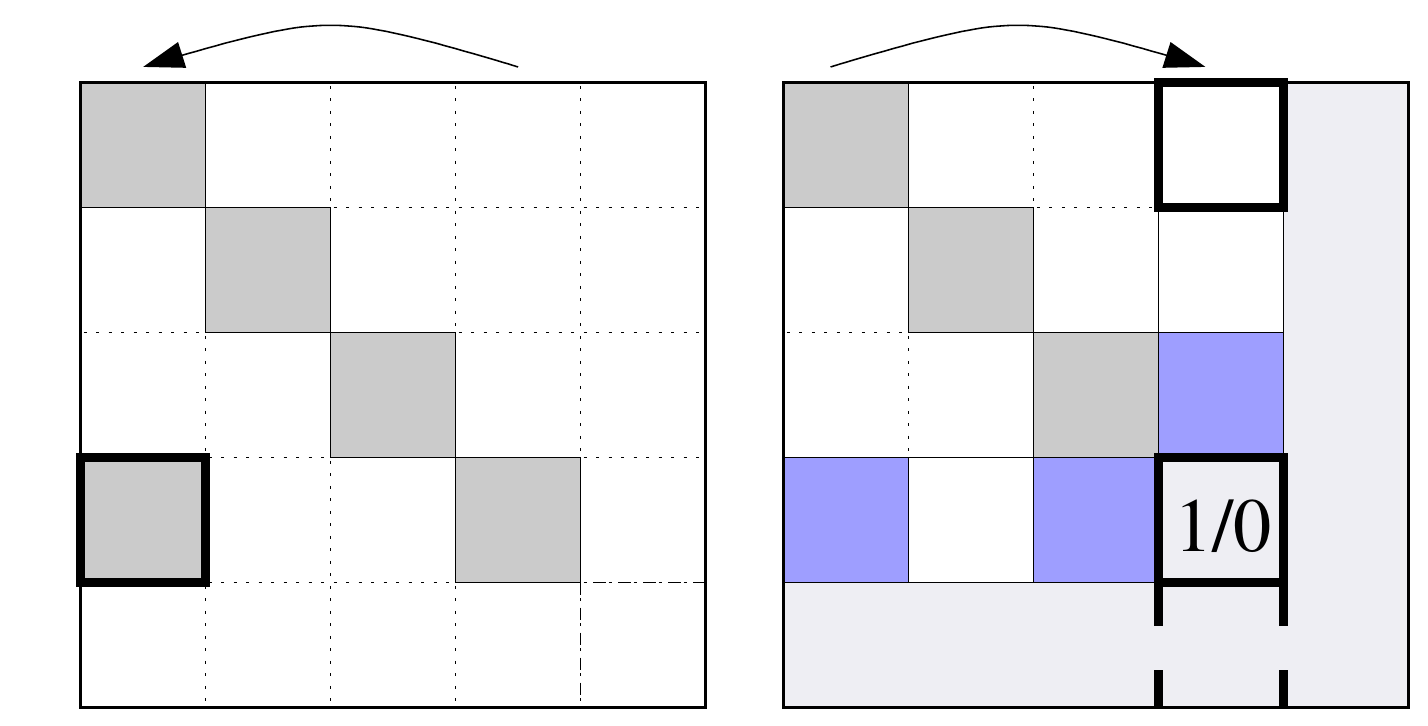}&
\includegraphics[width=0.32\textwidth]{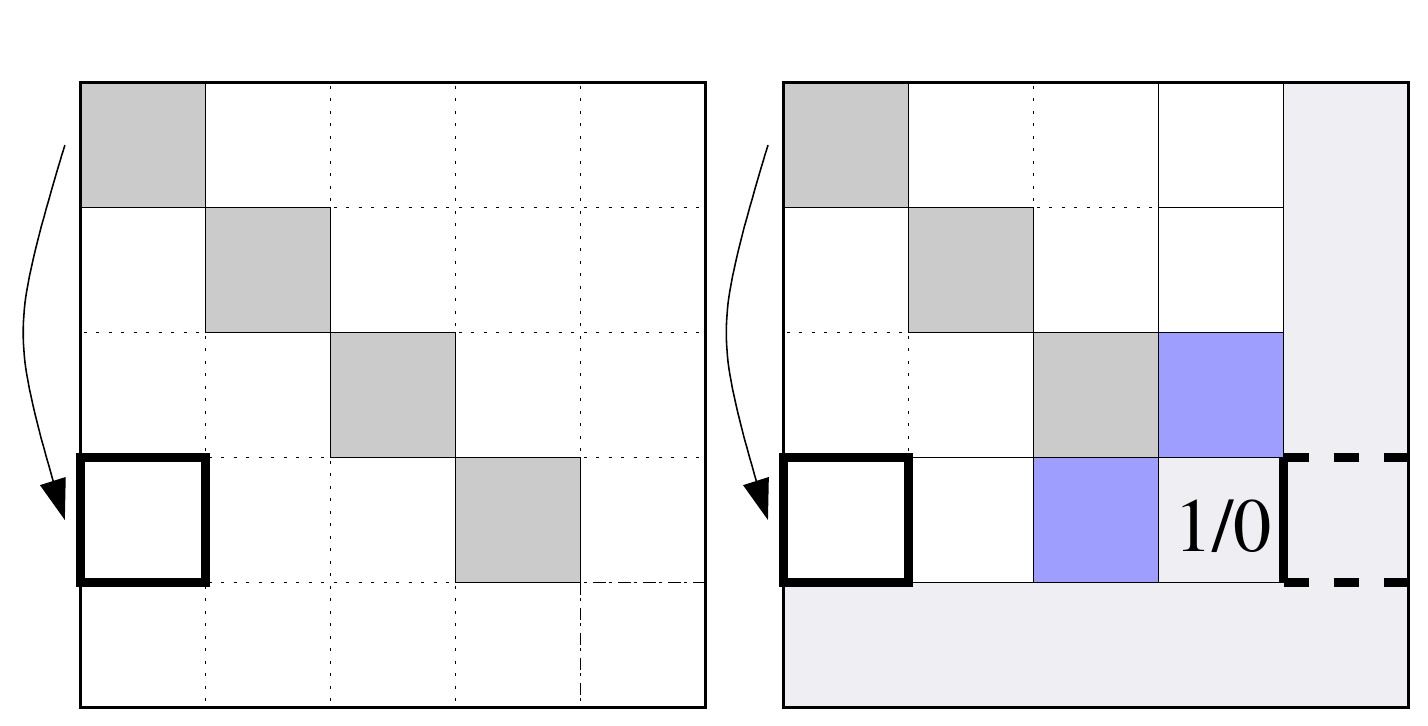}\\
({\bf{a}}) Initial tableau ($i=4$)&
({\bf{b}}) \cnot$(4,1)$ &
({\bf{c}}) sweep$(4,1)$\\[6pt]
\includegraphics[width=0.32\textwidth]{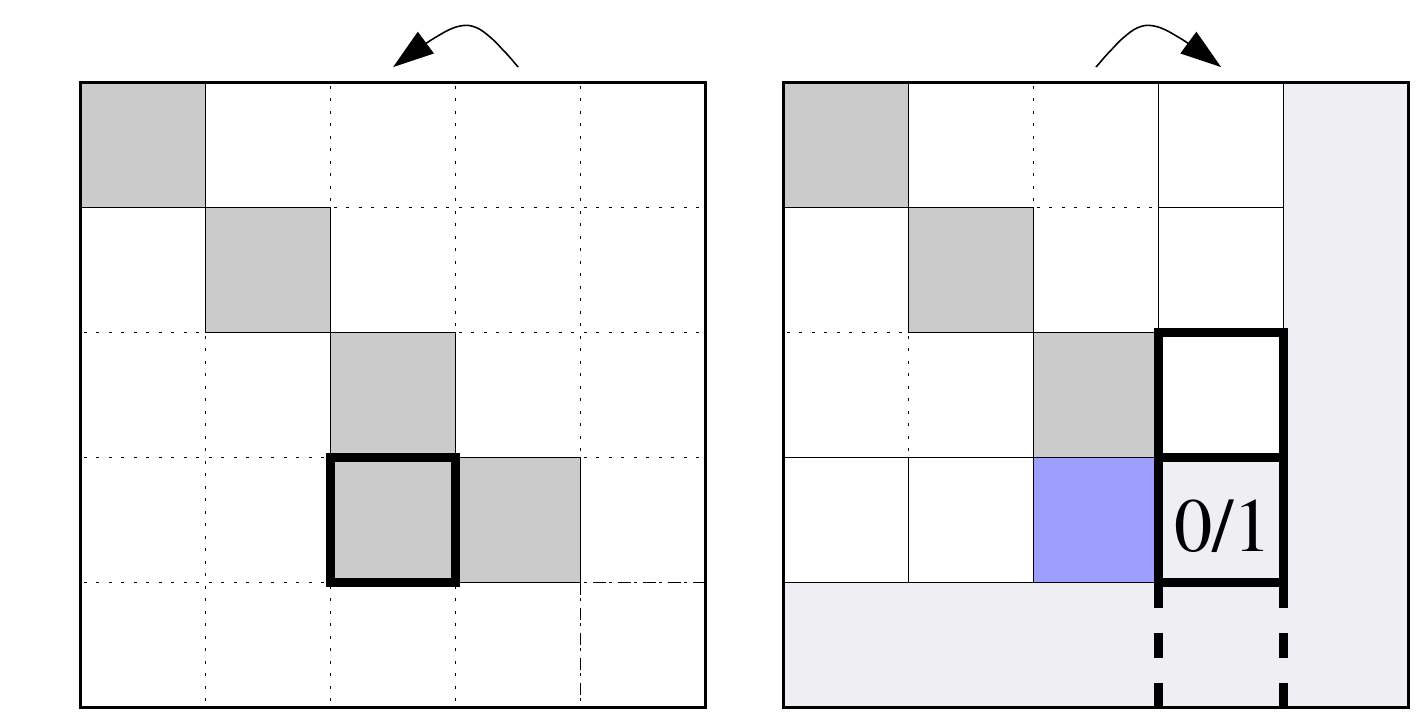}&
\includegraphics[width=0.32\textwidth]{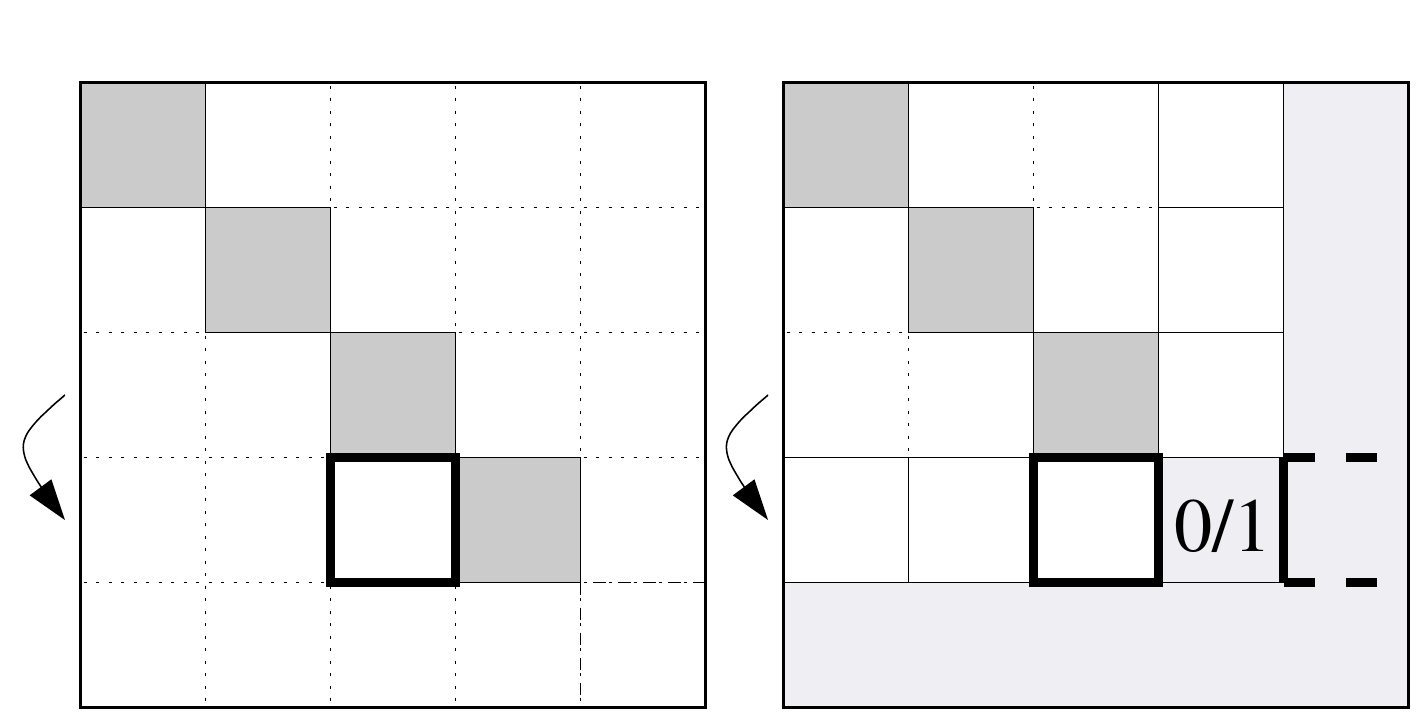}&
\includegraphics[width=0.32\textwidth]{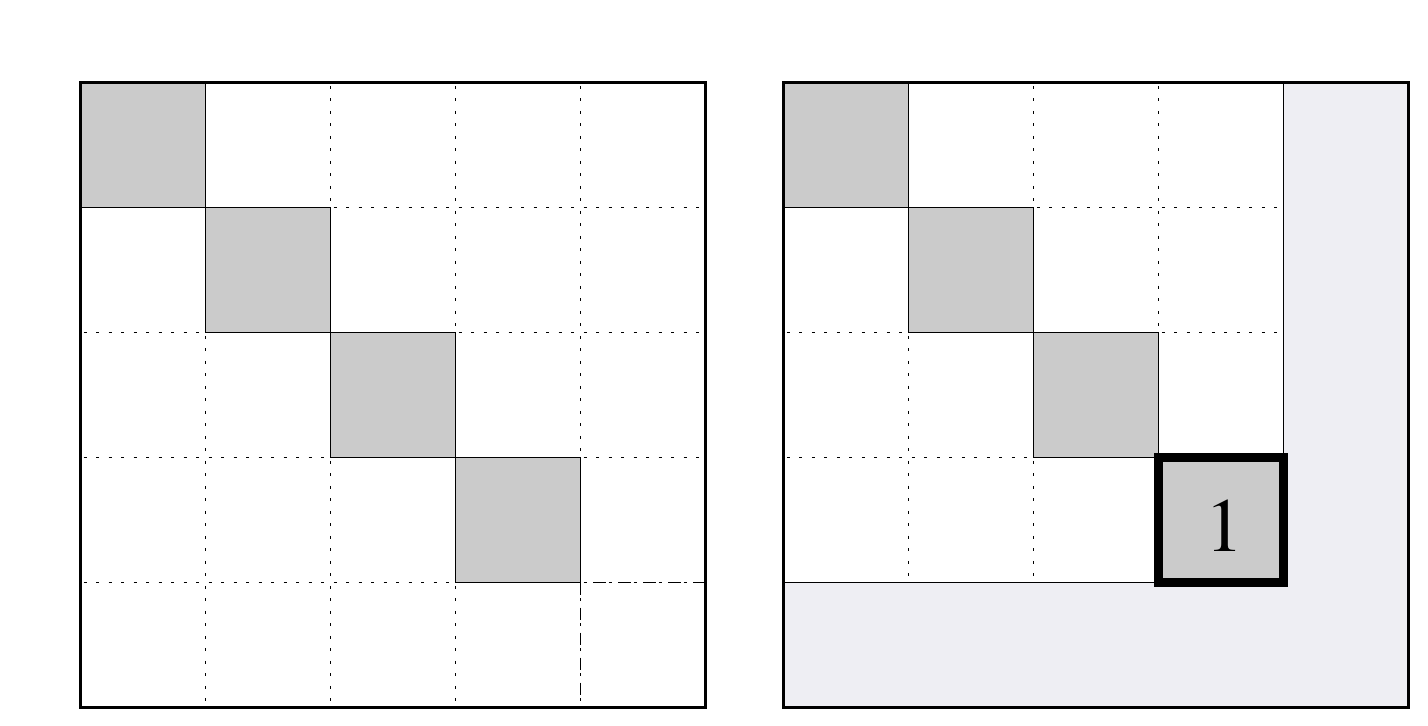}\\
({\bf{d}}) \cnot$(4,3)$ &
({\bf{e}}) sweep$(4,3)$ &
({\bf{f}}) final tableau  ($i=4$)\\
\end{tabular}
\caption{Principle behind the \cnot-based update of the Z block. The
  entries updated by each step are indicated by black boxes.}\label{Fig:CnotUpdateZ}
\end{figure}

Alternative way of updating Z that is based on \cnot\ operations is
given by Algorithm~\ref{Alg:CnotUpdateZ}. The main for-loop in
lines~\ref{Line:CnotZ_Begin_Stage1}--\ref{Line:CnotZ_End_Stage1}
iteratively ensures that the top-left $i\times i$ block of Z has ones
on the diagonal and zeroes elsewhere. The update process for a given
$i$ is illustrated in Figure~\ref{Fig:CnotUpdateZ}. At the beginning
of iteration $i$, the $(i-1)\times(i-1)$ block of Z is diagonal, and
to obtain the desired state at the end of the iteration we therefore
need to eliminate any nonzeros occurring in the first $i-1$ entries in
the $i$-th row and column of Z, and ensure that $Z[i,i]=1$. As an
example, consider the tableau in Figure~\ref{Fig:CnotUpdateZ}(a) at
the beginning of iteration $i$. During the iteration we will need to
eliminate entries $Z[4,1]$, $Z[4,3]$, and their reflections $Z[1,4]$
and $Z[3,4]$. For now we assume that that the entry $Z[i,i]$ is 0 or 1
respectively. To eliminate entry $Z[1,4]$ we first apply a
\cnot$(4,1)$ gate. In addition it also flips the value in $Z[i,i]$ to
1 or 0 respectively, and fills in element $X[4,1]$, as shown in
Figure~\ref{Fig:CnotUpdateZ}(b). Aside from this there are some
further updates to the entries of column $i$ with indices exceeding
$i$; these are irrelevant to the current iteration and will be dealt
with in later iterations. Next, we eliminate the undesirable fill of
element $X[4,1]$ by sweeping row 4 with row 1, which also clears up
element $Z[4,1]$. Note that this is no coincidence: since the X block
is diagonal again, if follows from Theorem~\ref{Thm:SymmetricZ} that
corresponding block in Z must be symmetric. We again ignore the
additional updates beyond the block boundaries. This leaves us at the
state shown in Figure~\ref{Fig:CnotUpdateZ}(c). As the next step we
eliminate entries $Z[3,4]$ and $Z[4,3]$ by applying \cnot$(4,3)$,
followed by a sweep of row 4 with row 4, as shown in
Figures~\ref{Fig:CnotUpdateZ}(d)
and~\ref{Fig:CnotUpdateZ}(e). Applying of the \cnot\ operation again
caused the value of $Z[i,i]$ to flip to 0 or 1 respectively. As a
final step, we now need to ensure that the $Z[i,i]$ entry is one. For
this we could check the latest value, and apply S$(i)$ whenever the
value is zero. Instead, we prefer to set the value appropriately at
the beginning, and ensure that at the end of all value flips it ends
at the one value. For this we can simply consider the value of
$Z[i,i]$ at the beginning and add the number of entries that need to
be eliminated and thus incur a flip. If this result value is even we
need to to change the initial value of $Z[i,i]$ by applying
S$(i)$. This is done in
lines~\ref{Line:CnotZ_Begin_Parity}--\ref{Line:CnotZ_End_Parity} of
Algorithm~\ref{Alg:CnotUpdateZ}. Once completed, the first $k$ columns
in Z exactly match those of X. We can therefore clear the X block by
applying phase and Hadamard operations on the first $k$ qubits, which
is done in
lines~\ref{Line:CnotZ_Begin_Stage2}--\ref{Line:CnotZ_End_Stage2}. Combined
with the diagonalization of X from Section~\ref{Sec:DiagonalizeX}, we
have the following result:

\begin{theorem}
  Given a tableau for commuting $n$-Paulis
  with rank $n$. We can diagonalize the operators using
  H-S-CX-S-H stages with $\mathcal{O}(n^2)$ CX gates.
\end{theorem}

\noindent This result can be further improved using~\cite{PAT2008MHa}, which
shows that \cnot\ circuits consisting of $\mathcal{O}(n^2)$ gates can
be reduced to $\mathcal{O}(n^2/\log(n))$ gates. The overall classical
complexity of this diagonalization procedure is $\mathcal{O}(mn\min(m,n))$.

\begin{algorithm}[t!]
\caption{Update of Z using \cnot\ operations, clear X.}\label{Alg:CnotUpdateZ}
\hspace*{0pt}{\vspace*{-3pt}}\\
\hspace*{0pt}{{\bf Input}: Tableau $T$ with diagonal $X$ of rank $k$.}\\
\hspace*{0pt}{{\bf Output}: Updated tableau with X block entries set to zero}\\
\hspace*{0pt}{{\bf Complexity}: $\mathcal{O}(k^2n)$}
\begin{algorithmic}[1]
\For{$i \in \{1,\ldots,k\}$}\label{Line:CnotZ_Begin_Stage1}
   \If{$\left(\sum_{j=1}^{i}Z[i,j]\ \mbox{is even}\right)$}\label{Line:CnotZ_Begin_Parity}
      \State{Apply S$(i)$}
   \EndIf\label{Line:CnotZ_End_Parity}
   \For{$j \in \{1,\ldots,i-1\}$}
      \If{$Z[i,j] = 1$}
         \State{Apply \cnot$(i,j)$}
         \State{Sweep row $i$ with row $j$}
      \EndIf
   \EndFor
\EndFor\label{Line:CnotZ_End_Stage1}
\For{$i \in \{1,\ldots,k\}$}\label{Line:CnotZ_Begin_Stage2}
   \State{Apply S$(i)$, H$(i)$}
\EndFor\label{Line:CnotZ_End_Stage2}
\end{algorithmic}
\end{algorithm}

\subsubsection{Column-based elimination}\label{Sec:ColumnwiseDiagonalization}

In the two methods described so far, each iteration of the algorithm
for updating the Z block zeroes out exactly two elements. In many
cases we can do much better and clear multiple entries at once. To see
how, consider the situation where the X block is diagonal and the
initial Z block is as shown in Figure~\ref{Fig:ColumnZ}(a). The second
and third column are nearly identical, and sweeping one with the other
using a \cnot\ operation would leave only a single non-zero entry in
the updated column in the location where the two differed. This
suggests the following approach. Given a set of columns that is yet to
be swept, $\mathcal{I}$, we first determine the column
$i\in\mathcal{I}$ that has the minimal number of non-zero off-diagonal
elements; that is, the number of \cnot\ gates needed to clear them. We
then consider the Hamming distance between all pairs of columns
$i,j\in\mathcal{I}$, excluding rows $i$ and $j$. The reason for
excluding these entries is that the X block is diagonal, and we can
therefore easily update the diagonal entries in the Z block to the
desired value using Hadamard or phase gates. The total number of
\cnot\ operations to clear column $i$ with column $j$ is then equal to
their off-diagonal distance plus one for the column sweep itself. That
is, after sweeping the columns we still need to take care of the
remaining entries in the column using elementwise elimination. There
are many possible ways to combine these steps, but one approach is to
greedily determine the lowest number of \cnot\ operations needed to
clear any of the remaining columns in $\mathcal{I}$, an approach we
refer to as greedy-1. Once the column has been cleared aside from the
diagonal entry we can zero out the corresponding column in the X block
and remove the entry from $\mathcal{I}$.

\begin{figure}
\centering
\begin{tabular}{ccccc}
\multicolumn{1}{r}{\includegraphics[height=90pt]{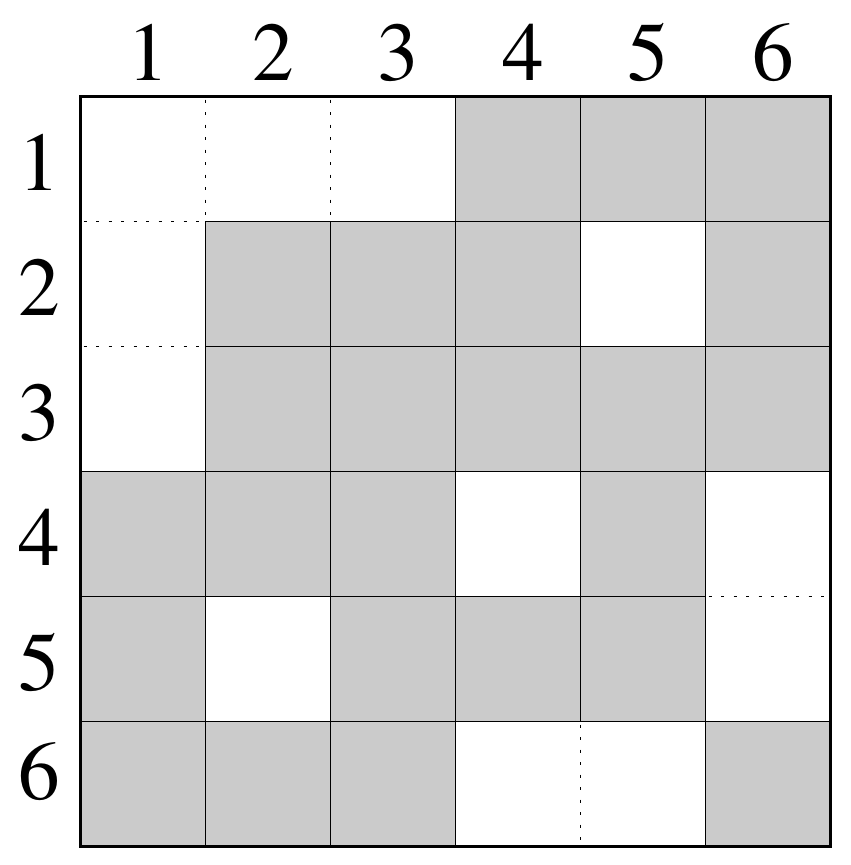}} &$\quad$&
\multicolumn{1}{r}{\includegraphics[height=90pt]{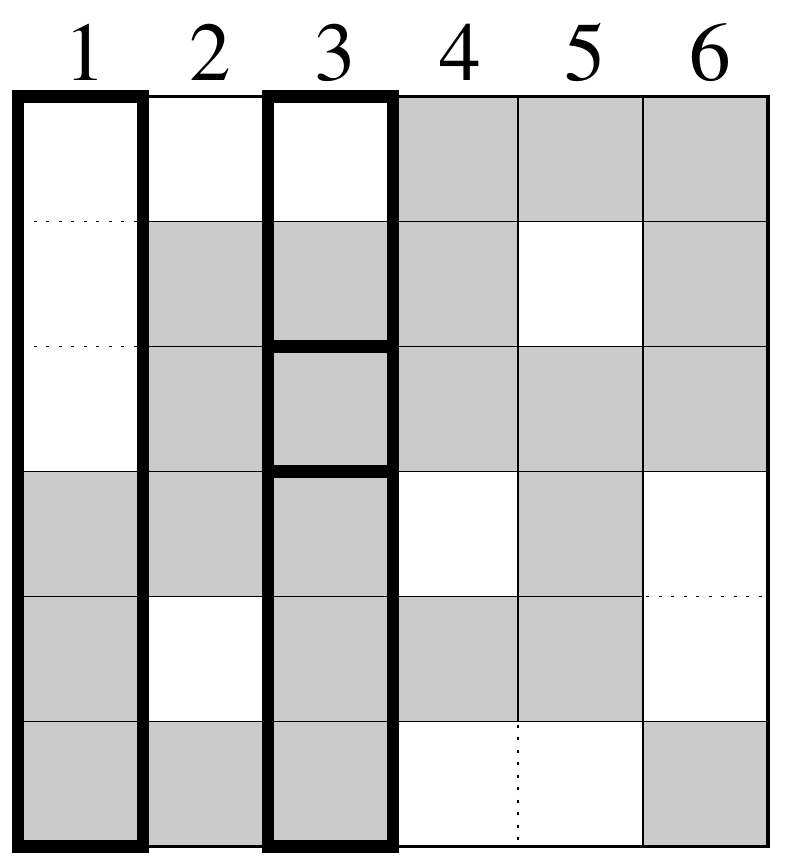}} &$\quad$&
\multicolumn{1}{r}{\includegraphics[height=90pt]{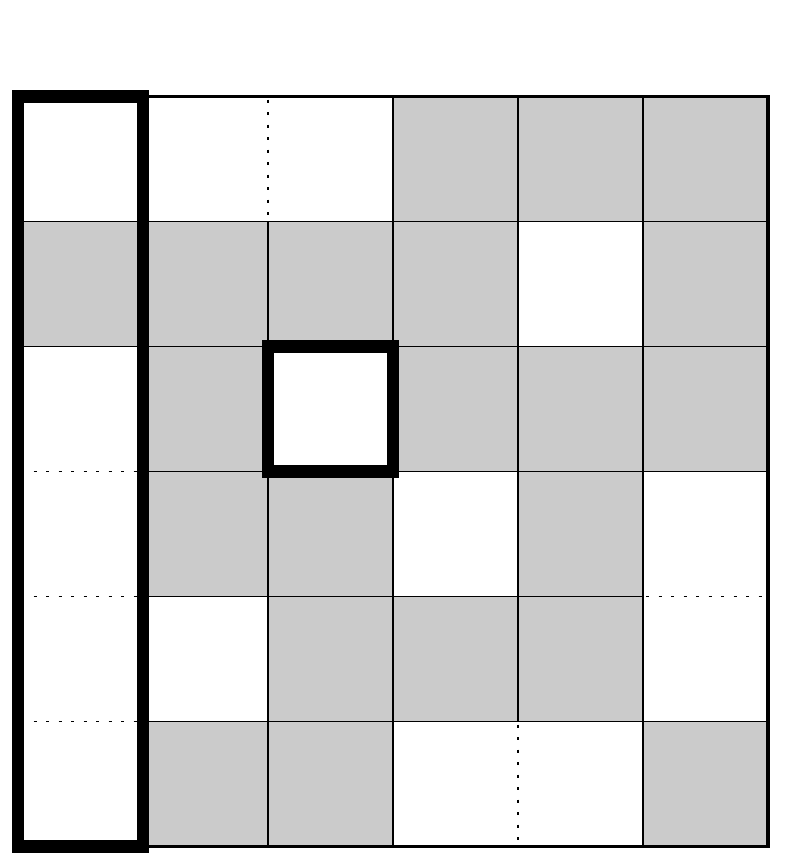}}\\[-1pt]
\hspace*{8pt}({\bf{a}}) && ({\bf{b}}) && ({\bf{c}}) \\[-3pt]
\multicolumn{1}{r}{\includegraphics[height=90pt]{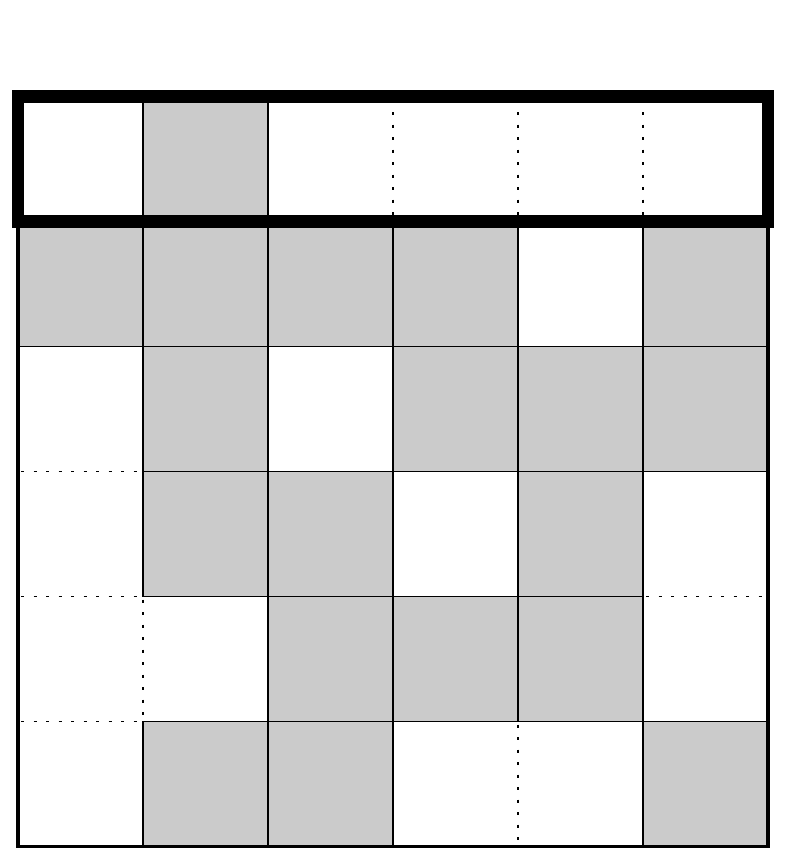}} &&
\multicolumn{1}{r}{\includegraphics[height=90pt]{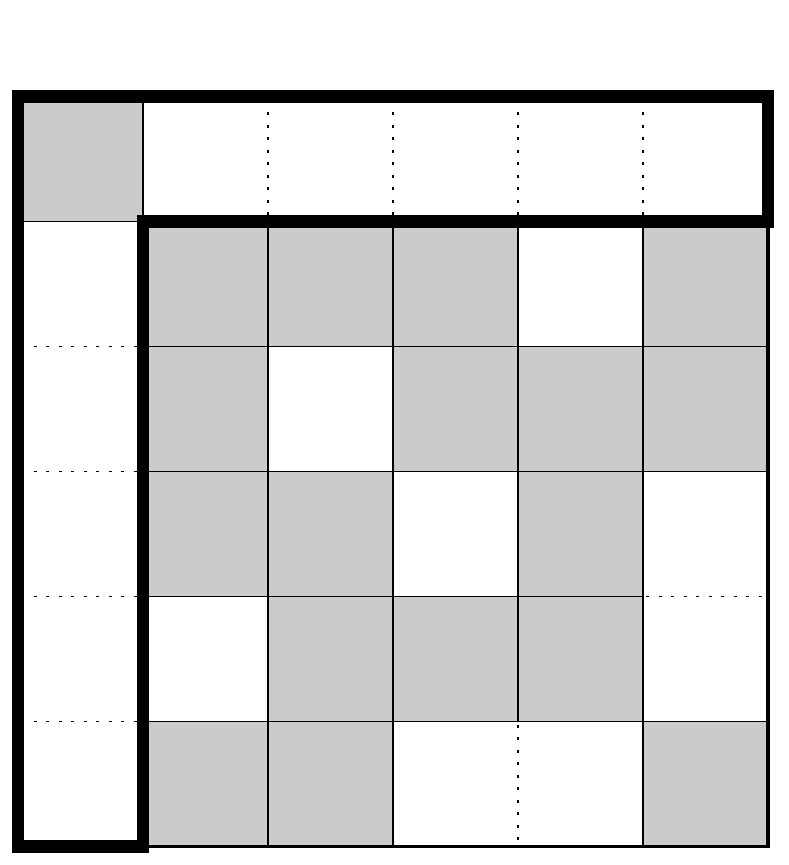}} &&
\multicolumn{1}{r}{\includegraphics[height=90pt]{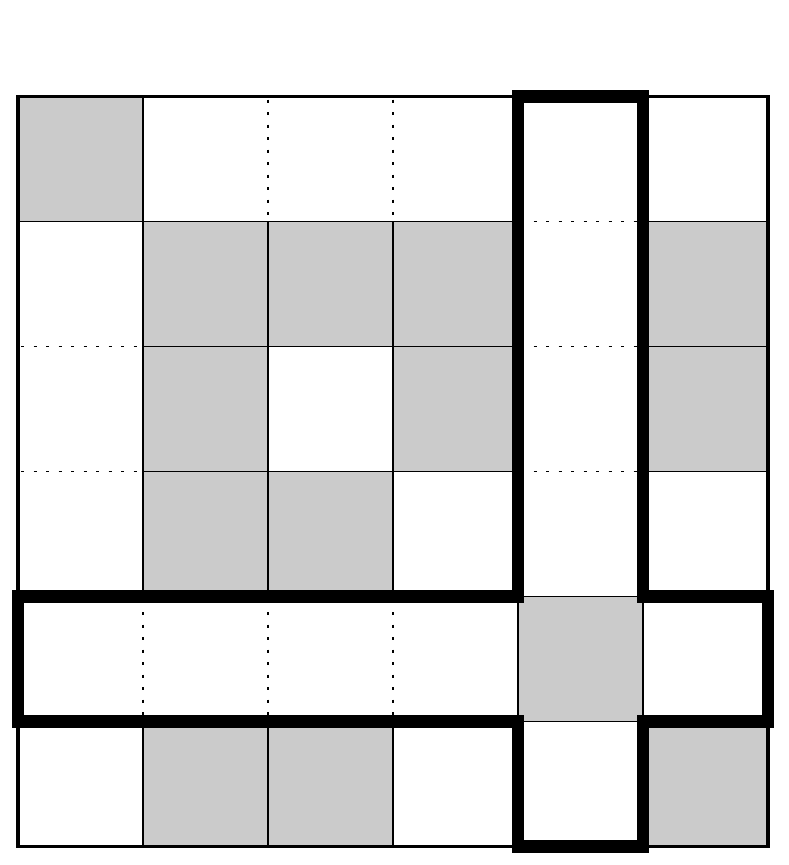}} \\[-1pt]
\hspace*{8pt}%({\bf{a}}) & ({\bf{b}}) & ({\bf{c}}) &
({\bf{d}}) && ({\bf{e}}) && ({\bf{f}})
\end{tabular}
\caption{Normalization of the Z block using column sweeps and
  elementwise elimination; (a) initial situation, (b)--(e) steps in
  the first iteration to normalize column 1 by sweeping with column 3,
and (f) second iteration of normalizing column 5 with elementwise
elimination only.}\label{Fig:ColumnZ}
\end{figure}

As an example we apply this method to the example in
Figure~\ref{Fig:ColumnZ}(a). Starting with
$\mathcal{I}=\{1,2,3,4,5,6\}$ we first determine the number of
off-diagonal elements to sweep in each single column, which turns out
to be three. For elimination using pairs of columns, we see that the
distance between columns 1 and 3 is one, provided we update the
diagonal entry in column 3. Columns 2 and 3 also have an off-diagonal
distance of two, as do columns 4 and 5. At each iteration we choose
the first minimum we encounter, in this case columns 1 and 3, as
highlighted in Figure~\ref{Fig:ColumnZ}(b). To clear column 1 we first
update the diagonal entry in 3 by applying a phase gate. Next, we
apply a \cnot\ operation that sweeps column 1 with the updated column
3, to arrive at the Z block shown in Figure~\ref{Fig:ColumnZ}(c).  As
seen in Figure~\ref{Fig:CnotUpdateZ}, the \cnot\ operation causes
fill-in of the X block, which we can eliminate by sweeping row 1 with
row 3. Doing so restores diagonality of the X block, and symmetry of
the Z block. The result of this operation can be seen in
Figure~\ref{Fig:ColumnZ}(d) . What remains is to pairwise eliminate
the remaining entries in column 1, and by symmetry of row 1, and clear
column 1 of the X block. This finalizes the clearance of column 1, so
we can remove it from the active set $\mathcal{I}$, and leaves us with
the tableau shown in Figure~\ref{Fig:ColumnZ}(e). Starting with a new
iteration, we again count the number of off-diagonal entries to sweep
per column. The minimum of two occurs in column 5. Pairwise sweeping
does not improve on this, and we therefore use the technique from
Section~\ref{Sec:PairwiseElimination} to clear these entries
directly. We then clear column 5 of the X block and remove the column
from $\mathcal{I}$. The algorithm continues in this fashion until
$\mathcal{I}$ is empty.

So far, we have only considered the number of \cnot\ operations. An
alternative approach, referred to in the experiments section as
greedy-2, takes into account the number of single-qubit gates when the
number of \cnot\ gates match. Recall that in the first iteration there
were several pairs of columns with a minimal off-diagonal distance of
one. The greedy-1 strategy chooses to clear column 1 with column 3,
which requires one phase gate to clear the diagonal entry of column 3,
a \cnot\ and \cz\ operation respectively for sweeping the column and
remaining off-diagonal entry, and finally a Hadamard operation to
clear column 1 of the X block. Alternatively choosing to clear column
2 with column 3 would require an initial \cnot\ for the column sweep,
a \cz\ for removing the remaining off-diagonal entry, and a Hadamard
operation to clear column 2 of X. The latter approach requires the
same number of \cnot\ operations, but requires one fewer single-qubit
gate. The greedy-2 method would therefore choose this option. For this
particular example, pairwise elimination requires ten \cnot\
operations, whereas the greedy approach require seven and six \cnot\
operations, respectively. For all three algorithms, the number of
single-qubit operations is six. The complexity of column-based
elimination of the Z block is $\mathcal{O}(k^4)$, where $k$ the rank
of the tableau. This assumes that at each stage of the algorithm we
recompute the distance between all pairs of remaining columns, and
more efficient implementations may be possible.

\subsection{Ordering of terms}\label{Sec:OrderingTerms}

Once the X block in the tableau has been cleared we can either undo
all row sweep and row swap operations, or reapply all Clifford
operators on the initial tableau, to obtain the diagonalized Pauli
terms corresponding to the given set of commuting Paulis.
Figure~\ref{Fig:ReorderZ}(a) shows the transpose of the resulting Z
block for a set of 20 Paulis over 7 qubits, represented as columns. In
the plot gray cells represents a Pauli \Pauli{Z} terms, while white
cells represent identity terms \Pauli{I}. For exponentiation we need
to add \cnot\ gates for each of the \Pauli{Z} terms. As illustrated in
Figure~\ref{Fig:Circuit3}, we can cancel \cnot\ operators between
successive \Pauli{Z} terms one the same qubit. The resulting number of
\cnot\ gates for each of the seven qubits is given on the right of
Figure~\ref{Fig:ReorderZ}(a), for a total of 72 \cnot\ gates. (For
ease of counting, imaging all-identity Paulis before the first and
after the last operator and count the number of transitions from white
to gray and vice versa.) In order to reduce the number of transitions
we can permute the order of the operators within the commuting
set. This is done in Figure~\ref{Fig:ReorderZ}(b), where we first sort
all operators in qubit one. We then recursively partition the
operators in the \Pauli{I} set, such that all \Pauli{I} operators
appear before \Pauli{Z} operators, and vice versa for the \Pauli{Z}
set. The resulting binary tree like structure in
Figure~\ref{Fig:ReorderZ}(b) reduces the total number of \cnot\ gates
needed to implement the circuit from the original 72 down to 58. The
order in which the qubits are traversed can make a big difference.
Figure~\ref{Fig:ReorderZ}(c) shows a histogram of the number of \cnot\
gates required for all possible permutations of traversal order,
ranging from 38 to 60 gates. The large range in gate count indicates
that there is still a lot of room for improvement for the ordering
strategy. As seen in Figure~\ref{Fig:ReorderZ}(b), qubits that appears
earlier in the ordering tend to require fewer \cnot\ gates. This can
be leveraged when optimizing the circuit for a particular quantum
processor where operators between non-neighboring qubits are
implemented using intermediate swap operations. In this case we can
reduce the number of \cnot\ operations between topologically distant
qubits by having them appear in the ordering earlier. Alternative
implementations where \cnot\ gates are connected to qubits of
successive \Pauli{Z} terms are possible, but will not be considered in
this paper. Ordering of operators in the Z block has a classical
complexity of $\mathcal{O}(mn)$.

\begin{figure}
\centering
\begin{tabular}{ccc}
\raisebox{11pt}{\includegraphics[width=0.34\textwidth]{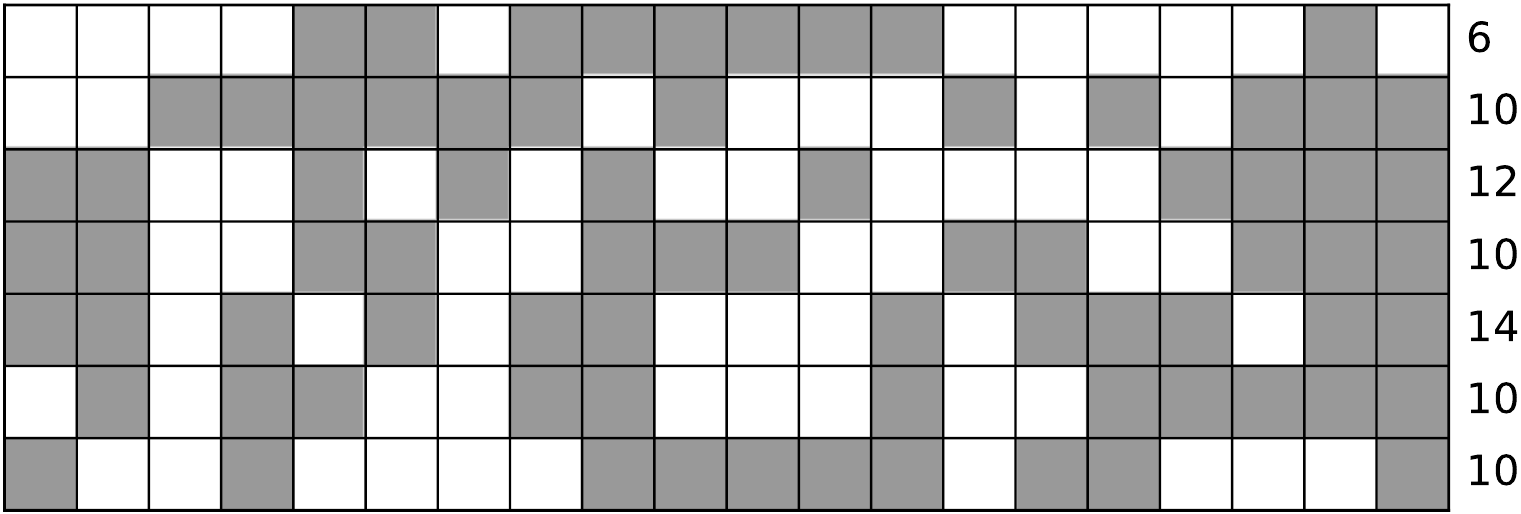}}&
\raisebox{11pt}{\includegraphics[width=0.34\textwidth]{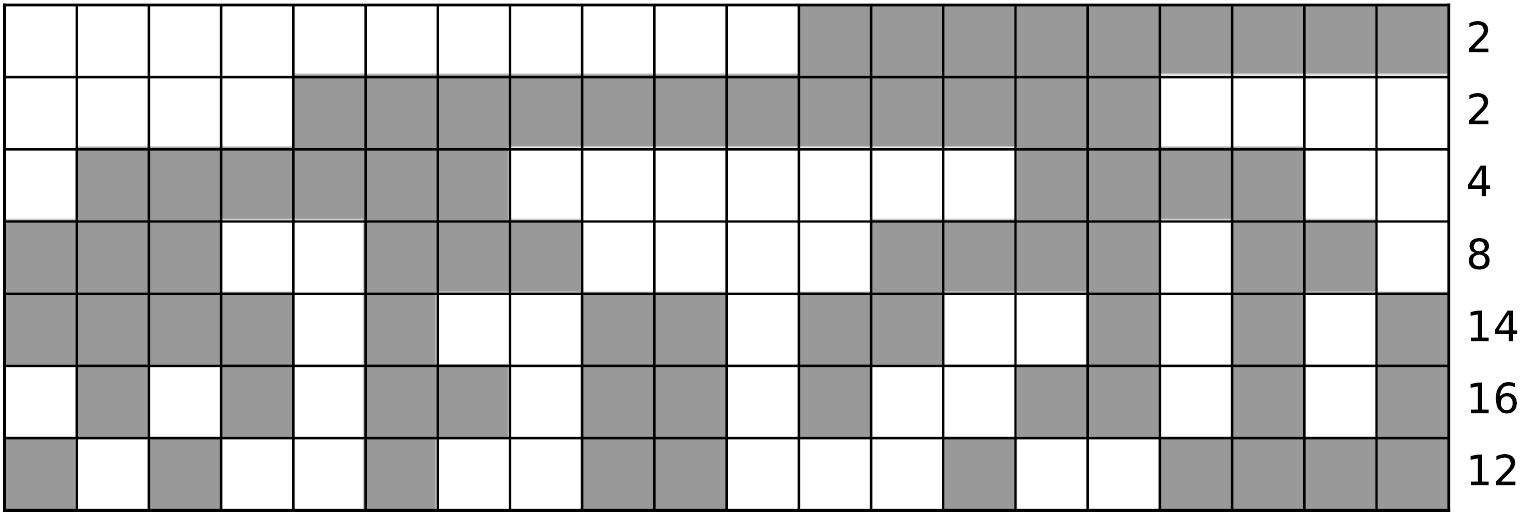}}&
\includegraphics[width=0.26\textwidth]{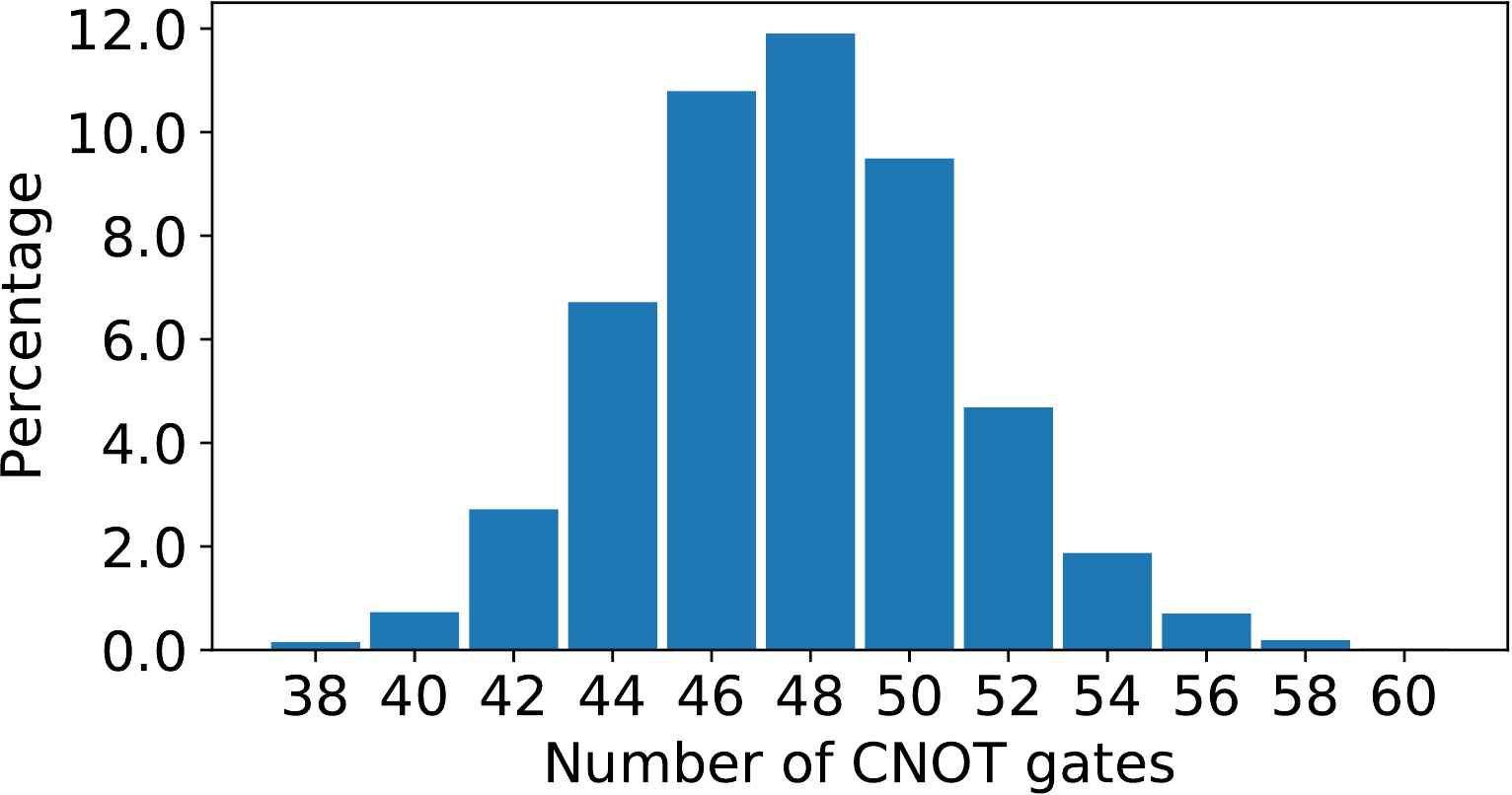}\\
({\bf{a}}) & ({\bf{b}}) & ({\bf{c}})
\end{tabular}
\caption{Transpose of the Z block with columns representing
  diagonalized Pauli operators (gray for \Pauli{Z} and white for
  \Pauli{I}), (a) directly after diagonalization, and (b) after
  reordering the columns. The required number of \cnot\ gates per
  qubit are given on the right of each rows. The total number of
  \cnot\ operations required for the circuit are 72 and 58,
  respectively. The histogram in plot (c) shows the percentage of all
  possible qubit orderings that require a certain number of \cnot\
  operations, ranging from 38 to 60.}\label{Fig:ReorderZ}
\end{figure}

\section{Experiments}\label{Sec:Experiments}

We now consider the practical application of the methods described in
earlier sections. In the experiments we consider the number of \cnot\
and single-qubit operations, as well as the circuit depth. The number
of \cnot\ gates that appear in the circuit are especially important
for processors with limited qubit connectivity. In particular, \cnot\
operations between qubits that are not physically connected may
require a substantial number of swap operations. We use
Qiskit~\cite{Qiskit} circuit optimization where indicated, and also
use the package to determine all circuit depths.

\subsection{Random Paulis}\label{Sec:ExpRandom}

\paragraph{Pauli bases.} As a first set of experiments we consider the
circuit complexity for diagonalizing random sets of commuting Paulis.
In order to run these experiments we need an algorithm for sampling
bases of commuting Paulis uniformly at random. For this we proceed in
two stages: first we uniformly sample a canonical generator set, and
second we sample a full-rank binary matrix. The resulting set of
Paulis is then obtained by multiplication of the generator set tableau
generator set with the binary matrix. Many of the random generators
can be sampled by setting the $X$ block in the tableau to the
identity, followed by randomly sampling a symmetric $Z$ block, as
required by Theorem~\ref{Thm:SymmetricZ}. Besides these there are
generators with one or more of the diagonal entries in $X$ set to
zero. Such entries are generated by clearing out the entries on,
below, and to the right of the given diagonal element in the $Z$ block
and exchanging the associated columns in the $X$ and $Z$ blocks.
Zeroing out the entries is needed to ensure that the diagonal element
in the $X$ block cannot be set to one using row exchanges. The
algorithm for stage one is summarized in
Algorithm~\ref{Alg:GeneratorCommuting}. For the first row of the
tableau we have $2^n$ possibilities for $Z$ if the diagonal of $X$ is
set, and a single possibility otherwise, for a total of $1+2^n$. For
the second row we can only set $n-1$ entries in $Z$ due to the
symmetry requirement, therefore giving a total of $1+2^{n-1}$. The
total number of possible generators thus obtained is indeed the
maximum~\cite{SAR2019Ba}:
\[
\prod_{i=0}^{n-1}(1+2^{n-i}).
\]
For the second stage we generate a binary $n\times n$ matrix with
entries selected uniformly at random. The probability that the given
matrix is full rank is given by~\cite{BER1980a}:
\[
\prod_{i=0}^{n-1} (1-2^{-(n-i)}) \leq \prod_{i=1}^{\infty}(1-2^{-i}) = 0.288789\ldots
\]
After sampling a matrix we therefore need to check whether the matrix
is full rank. If not we need to sample another matrix, until we find a
full-rank one. The expected number of matrix samples is no more than
five for any matrix size.

\begin{algorithm}[t!]
\caption{Random generator sets for commuting Paulis.}\label{Alg:GeneratorCommuting}
\hspace*{0pt}{\vspace*{-3pt}}\\
\hspace*{0pt}{{\bf Input}: Pauli size $n$.}\\
\hspace*{0pt}{{\bf Output}: Tableau with $n$ generators for
  random maximally commuting $n$-Pauli set.}\\
\hspace*{0pt}{{\bf Complexity}: $\mathcal{O}(n^2)$.}\\[-12pt]
\begin{algorithmic}[1]
\State{Initialize an empty tableau $T = [X,Z]$ with blocks of size
  $n\times n$}
\For{$i \in [n]$}
   \State{Draw integer $r$ uniformly at random from $[0,\ldots,2^{n+1-i}]$}
   \State{Set $X[i,i] = 1$}
   \If{$r = 2^{n+1-i}$}
      \State{Exchange $X[:,i]$ and $Z[:,i]$}
   \Else
       \For{$j \in [i,n]$}
          \State{Set $Z[i,j] = Z[j,i]= (r \mod 2)$}
          \State{$r \gets \lfloor r/2\rfloor$}
       \EndFor
   \EndIf
\EndFor
\State{Return the tableau with random signs}
\end{algorithmic}
\end{algorithm}

Using this procedure we generated twenty random sets of commuting
$n$-Pauli operators of size $n$ ranging from 3 to 25. The resulting
tableaus are guaranteed to have rank $n$ by construction. For each set
we apply the diagonalization procedure from
Section~\ref{Sec:PairwiseElimination} (cz), the \cnot{}-based approach
from Section~\ref{Sec:CNotDiagonalization}, either directly (cnot), or
using the \cnot{} reduction from~\cite{PAT2008MHa} with block size
equal to $\log_2(n)$ or the optimal block size in the range $1$
through $n$; labeled (cnot-log2) and (cnot-best), respectively. In
addition to the two greedy methods (greedy-1, greedy-2) described in
Section~\ref{Sec:ColumnwiseDiagonalization}, we also applied the
tableau normalization procedure described in~\cite{GAR2013MCa}, and
denoted (gmc).

The results averaged over the twenty problem instances of a given size
are summarized in Table~\ref{Table:Basic}. The first column of results
list the number of \cnot\ operations, the number of single-qubit
gates, and the depth of the generated circuit for diagonalizing the
set of Paulis. The second and third columns summarize the circuit
complexity when the methods are applied to simulate products of the
Pauli exponentials. We will first focus on the circuit complexity of
the diagonalization and consider the simulation results later. For the
diagonalization process we also provide an aggregated comparison of
the performance of the different methods in
Figure~\ref{Fig:Basic}. This figure gives the percentage of problem
instances, across all problem sizes, for which the method on the
vertical axis strictly outperforms the method on the horizontal axis.

From the results in Table~\ref{Table:Basic} we see that the
performance of the gmc method is closest to that of the cnot
method. Overall, though we still find that the cnot method requires
fewer \cnot\ gates for 62\% of the problems and fewer single-qubit
gates in 84\% of the cases. In terms of depth of the diagonalization
circuit, we see that gmc generally outperforms cnot-best and both
greedy methods. However, the latter three methods require far fewer
\cnot\ and single-qubit gates than gmc. The cz method generally
outperforms gmc and the three cnot methods in terms of both gate
counts and circuit depth. The greedy approaches excel at reducing the
number of \cnot\ gates, but generally have a larger circuit depth. The
greedy-2 approach additionally outperforms all methods in terms of the
number of single-qubit gates, although this difference is only
marginal for the cz method. The cnot-best method chooses a block size
that minimizes the \cnot\ count across all possible block sizes, and
by definition is therefore never outperformed by cnot-log2. The number
of single qubit gates is not affected by the optimization of the
\cnot\ operations and is therefore identical for all three cnot
methods. The optimal choice of blocksize was relatively small and
equal to two for 48\% of the test problems, three for 28\%, and 4 for
some 10\% of the problems. For problems with $n$ between 20 and 25,
the frequencies changed to 48\%, 40\%, and 10\%, respectively.
For the very small problem sizes it was often found that the
unoptimized \cnot\ circuit was at least as good as the optimized
one, and amounted to around 12\% over all test problems.

\begin{figure}[t!]
\centering
\begin{tabular}{rrr}
\includegraphics[width=165.6pt]{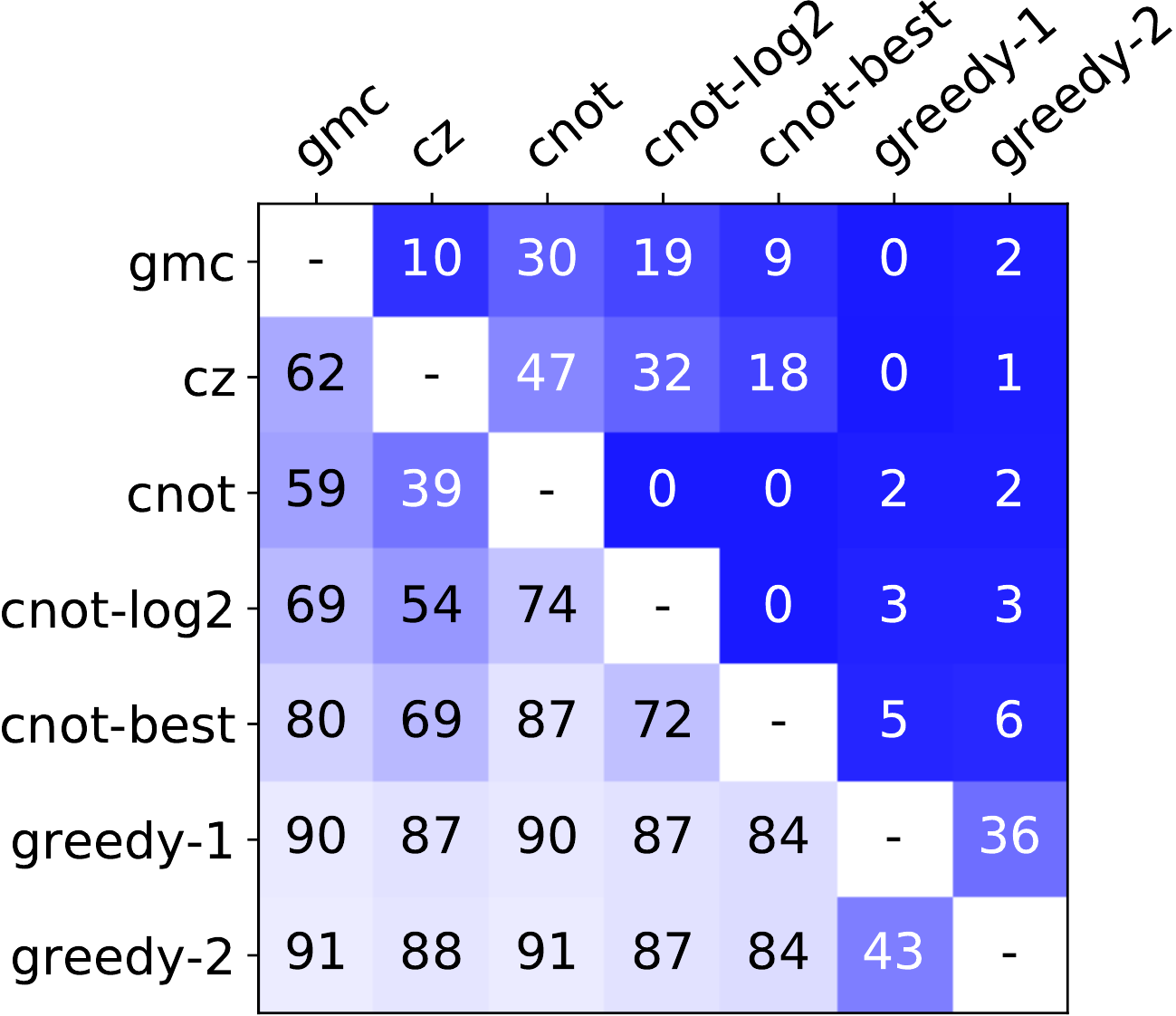}\hspace*{-5pt}&
\includegraphics[height=142pt]{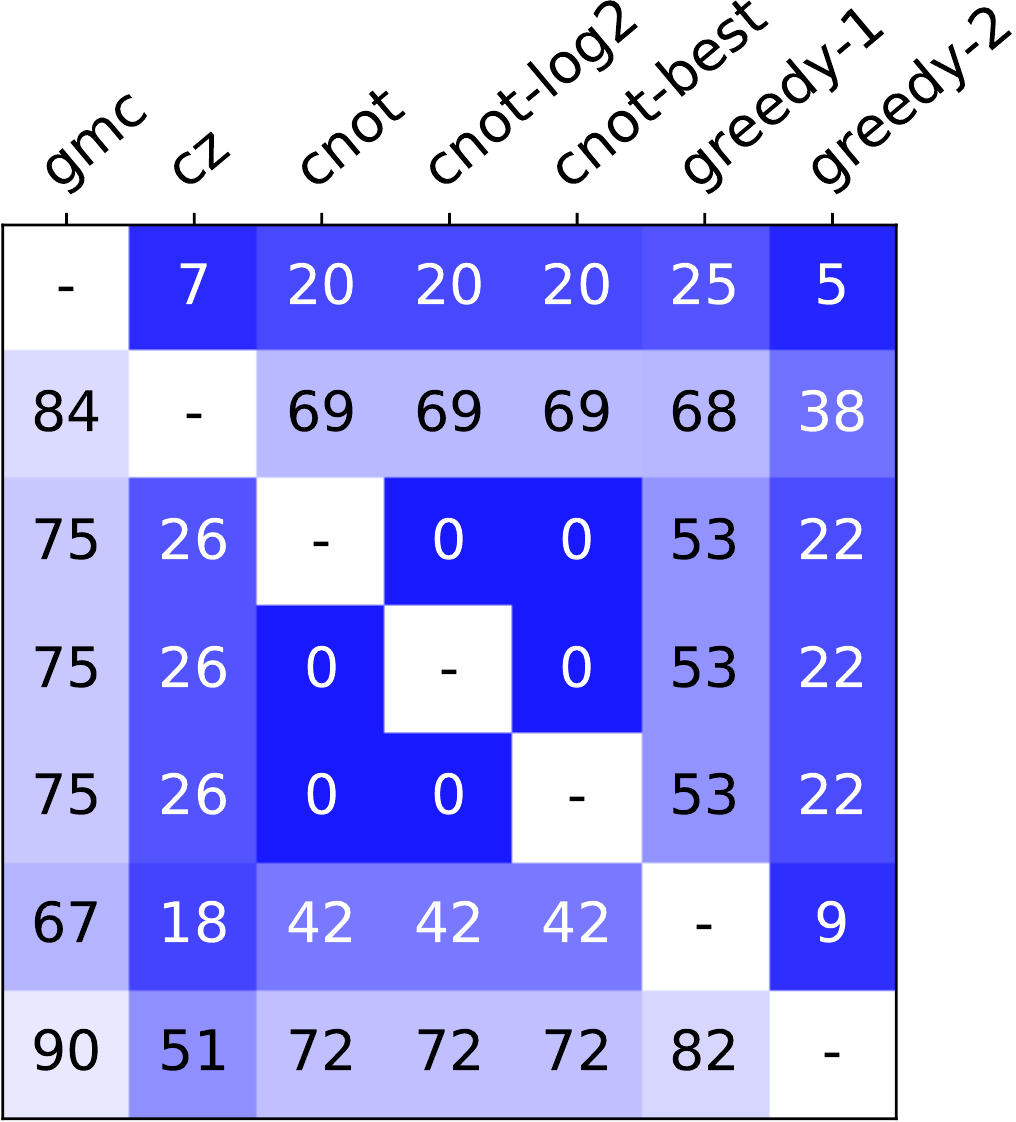}\hspace*{-5pt}&
\includegraphics[height=142pt]{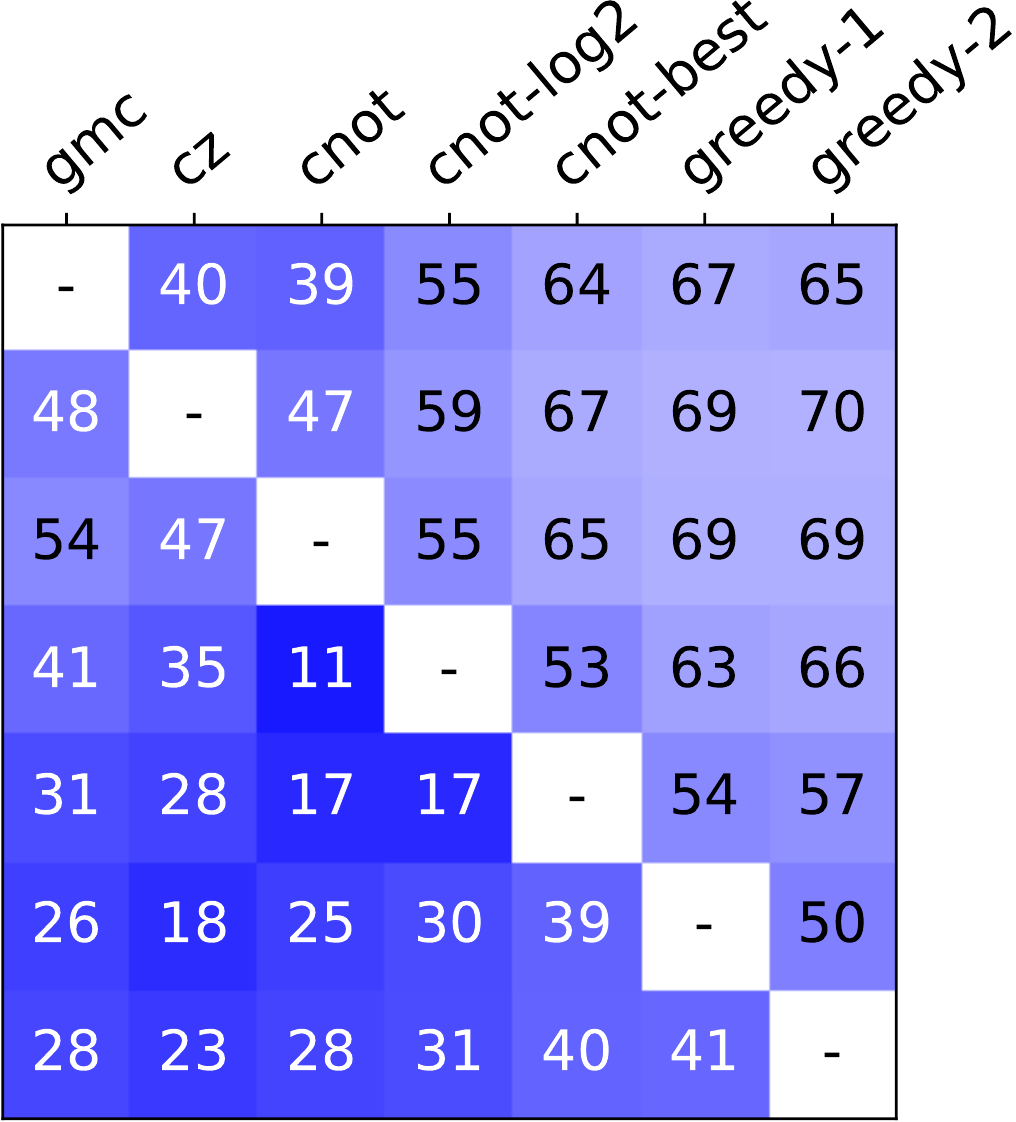}\\
\multicolumn{1}{c}{\hspace*{32pt}\cnot{}} &
\multicolumn{1}{c}{single qubit\hspace*{18pt}} &
\multicolumn{1}{c}{circuit depth\hspace*{18pt}}
\end{tabular}
\caption{Comparison of the different diagonalization methods on random
  Pauli basis. The percentage in each block (along with the associated
  color) indicates how often the method on the vertical axis is
  strictly better than the method along the horizontal axis in terms
  of the \cnot\ count, single-qubit gate count, and the circuit depth.}\label{Fig:Basic}
\end{figure}

We now consider the performance of the different methods in evaluating
the product of exponentials of the Paulis in each set. For this we
include the direct method, which was described in
Section~\ref{Sec:Existing}. The circuits generated are pre-optimized
by omitting gates that clearly cancel. For the direct method we
additionally apply level-two circuit optimization as proved by
Qiskit. The results of these experiments are summarized in the two
simulation columns of Table~\ref{Table:Basic}. The second of these
columns gives the result after optimizing the order of the Pauli
operators. For the diagonalization-based approaches we use the
procedure described in Section~\ref{Sec:OrderingTerms}, with sorting
applied according to the canonical qubit order. For the direct
approach we adopt a greedy approach in which we iteratively pick an
unused operator whose addition requires the smallest number of
additional \cnot\ gates, and in case of a tie, the smallest the number
of single-qubit gates.

\begin{table}
\centering
\begin{tabular}{|l|l|rrr|rrrr|rrrr|}
\hline
$n$ & Algorithm & \multicolumn{3}{c|}{Circuit $\mathcal{U}$} & \multicolumn{4}{c|}{Simulation} & \multicolumn{4}{c|}{Simulation (optimized)}\\
&&    &    &    & \multicolumn{2}{c}{\cnot{}} &     &    & \multicolumn{2}{c}{\cnot{}} &     &  \\
&& \cnot{} & single & depth & total  & exp.   & single & depth & total  & exp.   & single & depth\\
\hline
5& gmc &6&12&8&28&16&28&34&25&\cellcolor{highlight}{\color{darkblue}13}&28&32\\
& cz &5&11&8&25&15&27&33&23&\cellcolor{highlight}{\color{darkblue}13}&27&32\\
& cnot &5&14&7&25&15&33&33&23&\cellcolor{highlight}{\color{darkblue}13}&33&32\\
& cnot-log2 &\cellcolor{highlight}{\color{darkblue}4}&14&7&25&17&33&33&22&14&33&31\\
& cnot-best &\cellcolor{highlight}{\color{darkblue}4}&14&7&24&16&33&32&22&14&33&31\\
& greedy-1 &\cellcolor{highlight}{\color{darkblue}4}&11&8&\cellcolor{highlight}{\color{darkblue}22}&\cellcolor{highlight}{\color{darkblue}14}&28&32&\cellcolor{highlight}{\color{darkblue}21}&\cellcolor{highlight}{\color{darkblue}13}&28&31\\
& greedy-2 &\cellcolor{highlight}{\color{darkblue}4}&\cellcolor{highlight}{\color{darkblue}9}&\cellcolor{highlight}{\color{darkblue}6}&23&15&\cellcolor{highlight}{\color{darkblue}23}&\cellcolor{highlight}{\color{darkblue}31}&\cellcolor{highlight}{\color{darkblue}21}&\cellcolor{highlight}{\color{darkblue}13}&\cellcolor{highlight}{\color{darkblue}23}&\cellcolor{highlight}{\color{darkblue}29}\\
& direct &--&--&--&30&30&34&43&28&28&33&40\\
\hline
10& gmc &25&28&17&106&56&66&94&95&\cellcolor{highlight}{\color{darkblue}45}&66&84\\
& cz &22&23&16&100&56&57&92&89&\cellcolor{highlight}{\color{darkblue}45}&57&82\\
& cnot &22&26&\cellcolor{highlight}{\color{darkblue}14}&100&56&63&\cellcolor{highlight}{\color{darkblue}89}&90&46&63&\cellcolor{highlight}{\color{darkblue}80}\\
& cnot-log2 &21&26&15&96&\cellcolor{highlight}{\color{darkblue}54}&63&90&87&\cellcolor{highlight}{\color{darkblue}45}&63&81\\
& cnot-best &19&26&15&93&55&63&90&84&46&63&82\\
& greedy-1 &\cellcolor{highlight}{\color{darkblue}14}&25&17&86&58&60&95&\cellcolor{highlight}{\color{darkblue}74}&46&60&83\\
& greedy-2 &\cellcolor{highlight}{\color{darkblue}14}&\cellcolor{highlight}{\color{darkblue}21}&16&\cellcolor{highlight}{\color{darkblue}85}&57&\cellcolor{highlight}{\color{darkblue}53}&95&\cellcolor{highlight}{\color{darkblue}74}&46&\cellcolor{highlight}{\color{darkblue}53}&84\\
& direct &--&--&--&118&118&123&142&109&109&114&132\\
\hline
15& gmc &58&45&24&234&\cellcolor{highlight}{\color{darkblue}118}&104&174&216&100&104&157\\
& cz &51&\cellcolor{highlight}{\color{darkblue}35}&\cellcolor{highlight}{\color{darkblue}23}&221&119&\cellcolor{highlight}{\color{darkblue}85}&\cellcolor{highlight}{\color{darkblue}172}&202&100&\cellcolor{highlight}{\color{darkblue}85}&\cellcolor{highlight}{\color{darkblue}154}\\
& cnot &52&38&\cellcolor{highlight}{\color{darkblue}23}&225&121&92&174&205&101&92&157\\
& cnot-log2 &48&38&26&217&121&92&181&198&102&92&163\\
& cnot-best &45&38&27&210&120&92&183&191&101&92&165\\
& greedy-1 &\cellcolor{highlight}{\color{darkblue}32}&40&29&184&120&94&186&\cellcolor{highlight}{\color{darkblue}162}&\cellcolor{highlight}{\color{darkblue}98}&94&165\\
& greedy-2 &\cellcolor{highlight}{\color{darkblue}32}&\cellcolor{highlight}{\color{darkblue}35}&31&\cellcolor{highlight}{\color{darkblue}182}&\cellcolor{highlight}{\color{darkblue}118}&\cellcolor{highlight}{\color{darkblue}85}&189&163&99&\cellcolor{highlight}{\color{darkblue}85}&171\\
& direct &--&--&--&256&256&268&292&234&234&244&270\\
\hline
20& gmc &102&62&33&414&210&144&287&388&184&144&262\\
& cz &96&\cellcolor{highlight}{\color{darkblue}46}&32&402&210&\cellcolor{highlight}{\color{darkblue}112}&286&376&184&\cellcolor{highlight}{\color{darkblue}112}&260\\
& cnot &95&51&\cellcolor{highlight}{\color{darkblue}31}&398&\cellcolor{highlight}{\color{darkblue}208}&123&\cellcolor{highlight}{\color{darkblue}281}&373&183&123&\cellcolor{highlight}{\color{darkblue}258}\\
& cnot-log2 &90&51&35&388&\cellcolor{highlight}{\color{darkblue}208}&123&289&364&184&123&266\\
& cnot-best &82&51&43&372&\cellcolor{highlight}{\color{darkblue}208}&123&306&348&184&123&282\\
& greedy-1 &58&56&47&326&210&133&313&296&\cellcolor{highlight}{\color{darkblue}180}&133&284\\
& greedy-2 &\cellcolor{highlight}{\color{darkblue}56}&51&47&\cellcolor{highlight}{\color{darkblue}322}&210&123&314&\cellcolor{highlight}{\color{darkblue}292}&\cellcolor{highlight}{\color{darkblue}180}&123&284\\
& direct &--&--&--&458&458&459&505&424&424&423&469\\
\hline
25& gmc &151&81&41&626&324&186&417&586&\cellcolor{highlight}{\color{darkblue}284}&186&\cellcolor{highlight}{\color{darkblue}380}\\
& cz &147&\cellcolor{highlight}{\color{darkblue}60}&\cellcolor{highlight}{\color{darkblue}40}&617&323&\cellcolor{highlight}{\color{darkblue}144}&\cellcolor{highlight}{\color{darkblue}416}&578&\cellcolor{highlight}{\color{darkblue}284}&\cellcolor{highlight}{\color{darkblue}144}&\cellcolor{highlight}{\color{darkblue}380}\\
& cnot &150&63&\cellcolor{highlight}{\color{darkblue}40}&630&330&150&423&590&290&150&386\\
& cnot-log2 &142&63&51&614&330&150&444&573&289&150&408\\
& cnot-best &129&63&62&588&330&150&466&548&290&150&429\\
& greedy-1 &\cellcolor{highlight}{\color{darkblue}92}&74&69&\cellcolor{highlight}{\color{darkblue}506}&\cellcolor{highlight}{\color{darkblue}322}&173&472&470&286&173&438\\
& greedy-2 &\cellcolor{highlight}{\color{darkblue}92}&66&71&510&326&157&481&\cellcolor{highlight}{\color{darkblue}469}&285&157&440\\
& direct &--&--&--&707&707&714&764&651&651&660&708\\
\hline
\end{tabular}

\caption{The average circuit complexity over twenty random Pauli bases
  of size $n$ for the diagonalization circuit $\mathcal{U}$, and the
  entire simulation circuit, including exponentiation. The optimized
  simulation block gives the circuit complexity after appropriately
  reordering the Pauli operators. The best values are highlighted.}\label{Table:Basic}
\end{table}

Even with these relatively simple optimizations we can see that the
number of \cnot\ gates and circuit depth exhibit a noticeable
reduction. The same applies to the number of single-qubit gates in the
direct approach, where the gates for individual diagonalization of
neighboring operators can cancel. For the diagonalization approaches
the number of single-qubit gates is unaffected, since the optimization
only affects the central part of the circuit, which consists entirely
of \cnot\ and $R_z$ gates, and none of the $R_z$ gates can be
simplified, unless some of the Pauli operators are repeated.  Despite
the small number of Pauli terms in the exponentiation, the overhead of
applying simultaneous diagonalization and its adjoin is still small
enough for the overall number of \cnot\ gates, and certainly the
number of single-qubit gates to compare very favorably against the
direct method. The same applies for the circuit depth, where we
observe a puzzling phenomenon for the diagonalization methods, seen
across the different problem sizes: methods with a larger number of
\cnot\ gates tend to have a smaller circuit depth. The total depth of
the circuit is approximately twice the diagonalization circuit depth,
plus the number of \cnot\ gates in the central part responsible for
exponentiation, plus an additional single-qubit $R_z$ gate for each of
the $n$ operators. From the \cnot\ exp.~column in
Table~\ref{Table:Basic} we see that the number of \cnot\ gates in the
central part of the circuit is nearly identical for the different
methods, and the difference must therefore be due to the depth of the
diagonalization circuits. Having more \cnot\ gates in a shallower
circuit indicates a higher level of parallelism where two or more
gates can be applied simultaneously. This also suggests an improvement
to the cz approach: instead of simply sweeping the entries row by row,
we could process the entries in a way that promotes parallelism by
avoiding repeated dependence on a single qubit.  Another possible
modification, which applies to all methods, is to connect the \cnot\
gates between pairs of qubits where the Pauli term is \Pauli{Z}, and
only eventually connect the partial parity values to the ancilla. This
approach can help improve locality of the \cnot\ operators, and enable
a higher level of parallelism, as the cost of potentially more complex
optimization and circuit generation code. Finally, recent
  work~\cite{BRA2020Ma} has shown that any Clifford operator can be
  implemented with a circuit of depth at most $9n$ on linear
  nearest-neighbor architectures. Although this limit is not reached
  by the Clifford operators for diagonalization, $\mathcal{U}$, in
  Table~\ref{Table:Basic} it is an interesting direction for future
  study.

\paragraph{General sets of Paulis.} 

When ignoring the sign, the number of $n$-Pauli operators that can
mutually commute is $2^n$. We can therefore expect that the number of
commuting Paulis in a set exceeds $n$, which was used in the
experiments above. In our next set of experiments we consider sets of
size $m$. We generate these by multiplying the XZ blocks of the
initial tableaus used earlier by a full-rank $m\times n$ binary
matrix, thereby generating a new tableau with X and Z block sizes
equal to $m\times n$. The sign column of the tableau is initialized at
random.

\begin{table}[ht!]
\centering
\begin{tabular}{|l|l|rrr|rrr|rrr|}
\hline
$m$ & Algorithm & \multicolumn{3}{c|}{\cnot\ count} & \multicolumn{3}{c|}{Single qubit} & \multicolumn{3}{c|}{Depth}\\
&& base & opt & rnd & base & opt & rnd & base & opt & rnd\\
\hline
3& cz & \cellcolor{highlight}{\color{darkblue}74}& \cellcolor{highlight}{\color{darkblue}75}& \cellcolor{highlight}{\color{darkblue}72}& \cellcolor{highlight}{\color{darkblue}38}& \cellcolor{highlight}{\color{darkblue}38}& \cellcolor{highlight}{\color{darkblue}38}& 65& \cellcolor{highlight}{\color{darkblue}66}& 62\\
& csw-cz & 94& 95& 93& 52& 52& 52& 68& 70& 68\\
& csw-cnot & 173& 173& 170& 126& 126& 126& 86& 86& 84\\
& cnot & \cellcolor{highlight}{\color{darkblue}74}& \cellcolor{highlight}{\color{darkblue}75}& \cellcolor{highlight}{\color{darkblue}72}& 43& 43& 43& 65& \cellcolor{highlight}{\color{darkblue}66}& 62\\
& cnot-log2 & \cellcolor{highlight}{\color{darkblue}74}& \cellcolor{highlight}{\color{darkblue}75}& \cellcolor{highlight}{\color{darkblue}72}& 43& 43& 43& 65& \cellcolor{highlight}{\color{darkblue}66}& 63\\
& cnot-best & \cellcolor{highlight}{\color{darkblue}74}& \cellcolor{highlight}{\color{darkblue}75}& \cellcolor{highlight}{\color{darkblue}72}& 43& 43& 43& 65& \cellcolor{highlight}{\color{darkblue}66}& 62\\
& greedy-1 & \cellcolor{highlight}{\color{darkblue}74}& \cellcolor{highlight}{\color{darkblue}75}& \cellcolor{highlight}{\color{darkblue}72}& 39& 39& 39& 65& \cellcolor{highlight}{\color{darkblue}66}& 62\\
& greedy-2 & \cellcolor{highlight}{\color{darkblue}74}& 76& \cellcolor{highlight}{\color{darkblue}72}& \cellcolor{highlight}{\color{darkblue}38}& \cellcolor{highlight}{\color{darkblue}38}& \cellcolor{highlight}{\color{darkblue}38}& \cellcolor{highlight}{\color{darkblue}64}& \cellcolor{highlight}{\color{darkblue}66}& \cellcolor{highlight}{\color{darkblue}61}\\
& direct & 77& 77& 76& 77& 78& 76& 85& 84& 84\\
\hline
10& cz & 227& 214& 201& \cellcolor{highlight}{\color{darkblue}72}& \cellcolor{highlight}{\color{darkblue}72}& \cellcolor{highlight}{\color{darkblue}72}& 159& 147& 135\\
& csw-cz & 258& 250& 235& 93& 93& 93& 164& 157& 141\\
& csw-cnot & 299& 293& 277& 132& 132& 132& 173& 168& 152\\
& cnot & 227& 212& 200& 76& 76& 76& \cellcolor{highlight}{\color{darkblue}158}& \cellcolor{highlight}{\color{darkblue}146}& \cellcolor{highlight}{\color{darkblue}132}\\
& cnot-log2 & 225& 210& 198& 76& 76& 76& 159& 147& 134\\
& cnot-best & 220& 206& 194& 76& 76& 76& 163& 151& 137\\
& greedy-1 & 212& 201& 186& 79& 79& 79& 162& 152& 137\\
& greedy-2 & \cellcolor{highlight}{\color{darkblue}211}& \cellcolor{highlight}{\color{darkblue}199}& \cellcolor{highlight}{\color{darkblue}185}& \cellcolor{highlight}{\color{darkblue}72}& \cellcolor{highlight}{\color{darkblue}72}& \cellcolor{highlight}{\color{darkblue}72}& 165& 153& 138\\
& direct & 236& 220& 203& 242& 221& 204& 259& 244& 226\\
\hline
50& cz & 702& 602& 569& \cellcolor{highlight}{\color{darkblue}142}& \cellcolor{highlight}{\color{darkblue}142}& \cellcolor{highlight}{\color{darkblue}142}& 616& 518& 483\\
& csw-cz & 702& 598& 564& 173& 173& 173& 615& 514& \cellcolor{highlight}{\color{darkblue}479}\\
& csw-cnot & 704& 597& 568& 153& 153& 153& 617& 512& 483\\
& cnot & 701& 589& 566& 153& 153& 153& \cellcolor{highlight}{\color{darkblue}614}& \cellcolor{highlight}{\color{darkblue}505}& 480\\
& cnot-log2 & 691& 580& 556& 153& 153& 153& 621& 512& 489\\
& cnot-best & 675& 564& 544& 153& 153& 153& 638& 529& 508\\
& greedy-1 & \cellcolor{highlight}{\color{darkblue}626}& 526& 491& 163& 163& 163& 643& 546& 508\\
& greedy-2 & 628& \cellcolor{highlight}{\color{darkblue}525}& \cellcolor{highlight}{\color{darkblue}488}& 153& 153& 153& 650& 548& 510\\
& direct & 1134& 1005& 959& 1152& 1018& 977& 1251& 1117& 1068\\
\hline
200& cz & 2209& 1601& 1544& \cellcolor{highlight}{\color{darkblue}292}& \cellcolor{highlight}{\color{darkblue}292}& \cellcolor{highlight}{\color{darkblue}292}& 2273& 1668& 1609\\
& csw-cz & 2204& 1600& 1532& 322& 322& 322& 2269& 1669& 1599\\
& csw-cnot & 2203& 1596& 1531& 302& 302& 302& 2266& \cellcolor{highlight}{\color{darkblue}1662}& \cellcolor{highlight}{\color{darkblue}1597}\\
& cnot & 2186& 1598& 1536& 303& 303& 303& \cellcolor{highlight}{\color{darkblue}2249}& 1664& 1601\\
& cnot-log2 & 2177& 1588& 1522& 303& 303& 303& 2257& 1671& 1605\\
& cnot-best & 2161& 1572& 1506& 303& 303& 303& 2273& 1688& 1622\\
& greedy-1 & \cellcolor{highlight}{\color{darkblue}2123}& \cellcolor{highlight}{\color{darkblue}1518}& \cellcolor{highlight}{\color{darkblue}1457}& 313& 313& 313& 2290& 1690& 1625\\
& greedy-2 & 2128& \cellcolor{highlight}{\color{darkblue}1518}& 1459& 303& 303& 303& 2300& 1692& 1631\\
& direct & 4526& 3798& 3714& 4574& 3823& 3752& 4986& 4238& 4151\\
\hline
\end{tabular}

\caption{Average complexity of the complete circuit, including
  diagonalization and exponentiation, over twenty
  problem instances of $m$ Pauli operators on $20$
  qubits using no optimization (base), single-pass optimization (opt),
  or the best of 100 randomized optimizations (rnd). The best values
  are highlighted.}\label{Table:Nonsquare}
\end{table}

We perform three types of optimization regarding the operator
order. The base option uses the operators in the order they are
provided. The opt strategy applies the ordering described above for
our experiments with sets of size $n$. The final optimization strategy
(rnd) aims to minimize the number of \cnot\ gates based on random
permutations. In particular, for the diagonalization methods, we use
permutations of $[n]$ to determine the qubit sorting order, as
described in Section~\ref{Sec:OrderingTerms}. For the direct approach
we use permutations of $[m]$ to shuffle the operator order before
applying the greedy optimization procedure described above; the first
permutation is the canonical ordering to ensure the results are at
least as good as those of the opt strategy. For our experiments we use
100 random permutations per setting and then select the result that
has the lowest number of \cnot\ gates. The gmc method as given
in~\cite{GAR2013MCa} does not apply to non-square tableaus and we
therefore do not use it in subsequent experiments. Instead, we
consider the \cz\ (csw-cz) and \cnot\ (csw-cnot) diagonalization
algorithms described in~\cite{CRA2019SWPa}. The average circuit
complexities for simulation, obtained for the three optimization
procedures for $n = 20$ and varying values of $m$, are shown in
Table~\ref{Table:Nonsquare}. Results in the table are grouped by the
resource type: \cnot\ and single-qubit counts and depth. Note that
this differs from Table~\ref{Table:Basic}, where the results were
grouped by optimization type (base or opt).  Looking at
Table~\ref{Table:Nonsquare} we see that, aside from the csw methods,
diagonalization-based simulation is uniformly better than the direct
method on our test problems, even for $m$ much less than $n$. If
needed, the csw methods augment tableaus with additional rows to make
them full rank. For $m < n$ rows are always added, which leads to an
increased circuit complexity due to additional entries that need to be
cleared. For $m\geq n$ this overhead generally disappears and reduces
the csw-cz complexity to that of method cz. The csw-cnot algorithm has
an additional \cnot\ stage that slightly increases the circuit
complexity compared to method cnot (the \cnot\ reduction technique
from~\cite{PAT2008MHa} also applies to the csw-cnot algorithm but was
not implemented here).  The diagonalization part of the circuit has a
complexity that is essentially constant for $m \geq n$, and the
overhead therefore diminishes as $m$ grows, thereby leading to a
further improvement over the direct method. Aside for $m=3$ we see
that the single optimization step used in opt can significantly reduce
the \cnot\ gate count and circuit depth. As before, the number of
single-qubit gates is unaffected by optimization for the
diagonalization-based methods, but reduced substantially for the
direct method. Randomized optimization helps further lower the circuit
complexity, although the improvement is much less pronounced.

In Table~\ref{Table:Nonsquare} we purposely omit results on the
complexity of the diagonalization circuit, as they were found to be
similar for $m < n$ and identical for $m\geq n$ to the ones shown in
Table~\ref{Table:Basic}.  The fact that we obtain identical circuits
for $m\geq n$ may seem surprising at first, but becomes apparent when
noting that a circuit that diagonalizes a generating set for Paulis
automatically diagonalizes all Paulis in the group it generates. We
here show the result for a slightly different procedure of
diagonalizing the X block, as summarized in
Algorithm~\ref{Alg:NormalizeFullRank}.

\begin{theorem}\label{Thm:NormalizeFullRank}
  Given a full-rank tableau $T=[X,Z]$ in
  $\mathbb{F}_2^{n\times 2n}$ Then the output of
  Algorithm~\ref{Alg:NormalizeFullRank} applied to tableau $B\cdot T$
  gives the same tableau and index set $\mathcal{I}$ for any
  full-rank $B\in \mathbb{F}_2^{m\times n}$ with $m\geq n$.
\end{theorem}
\begin{proof}
  For analysis it will be easier to update the algorithm to omit column
  exchanges between the $X$ and $Z$ blocks, and instead sweep directly
  based on the entries in the column of $X$ if the index was found,
  there or based on the entries in the column of $Z$, otherwise. Note
  that full-rankness of the tableau guarantees that at least one of
  the indices exists. Although we do not apply the column exchanges,
  we do maintain index set $\mathcal{I}$. Applying the Hadamard
  operator to the columns (qubits) in $\mathcal{I}$ after
  normalization, then gives the original algorithm since row-based
  operations commute with Hadamard.

  All tableaus are generated as linear combinations of rows in $T$. It
  then follows from full-rankedness of $B$ that all Paulis
  corresponding to the tableaus can be instantiated using the same
  generator set.  The updated normalization algorithm produces
  generator sets of the same form used in
  Algorithm~\ref{Alg:GeneratorCommuting}. From the analysis of the
  latter algorithm we know that representation in this form is unique;
  no generator set has more than one tableau
  representation. Algorithm~\ref{Alg:NormalizeFullRank} must therefore
  return the same tableau and index set $\mathcal{I}$.
\end{proof}

Given that the tableaus after diagonalization of the $X$ block the
number of Hadamard gates used in the process are identical, it follows
that the circuit complexity for simultaneous diagonalization is the
same for $m \geq n$. For \cz-based diagonalization, the expected
\cnot\ count then follows directly from the construction of random
Pauli bases in Algorithm~\ref{Alg:GeneratorCommuting}. For each of the
rows that are set in the $Z$ block, on average half of the entries will be
one. In case of the column swap, no additional entries are set to one,
and the expected number of elements to sweep is therefore
\[
\sum_{i=0}^{n-1} \frac{n-i}{2}\cdot\frac{2^{n-i}}{2^{n-i}+1} \leq n(n-1)/4.
\]
A consequence of Theorem~\ref{Thm:NormalizeFullRank} is that
Algorithm~\ref{Alg:NormalizeFullRank} can be used to generate a unique
representation of a stabilizer state, irrespective of its original
representation. Moreover, the $Z$ block and index set $\mathcal{I}$
can be concisely represented as a $n\times n+1$ binary matrix.
Similarly, the technique can be used to check if two sets of commuting
Paulis have a common generator set up to signs. Note that our
condition of full-rankedness of the tableau $T$ can be relaxed; if
needed the tableau can be augmented by adding rows with the missing
diagonal elements. These basis vectors are never used in linear
combinations of the original rows in $T$ and can be discarded after
normalization.

\begin{algorithm}[t!]
\caption{Normalization of full-rank tableau.}\label{Alg:NormalizeFullRank}
\hspace*{0pt}{\vspace*{-3pt}}\\
\hspace*{0pt}{{\bf Input}: Full-rank tableau $T=[X,Z]$ with block size
  $m\times n$ such that $m\geq n$.}\\[-12pt]
\begin{algorithmic}[1]
\State{Initialize $\mathcal{I} = \emptyset$}
\For{$k \in [n]$}
   \State{Search for index $i$ with $i\geq k$ such that $X[i,k]=1$}
   \If{index not found}
      \State{$\mathcal{I} \gets \mathcal{I} \cup \{k\}$}
      \State{Exchange $X[:,k]$ and $Z[:,k]$}
      \State{Search for index $i$ with $i\geq k$ such that $X[i,k]=1$}
   \EndIf
   \State{Swap rows $i$ and $k$}
   \For{$i \in [m]$ such that $i\neq k$ and $X[i,k]=1$}
      \State{Sweep row $i$ with row $k$}
   \EndFor
\EndFor
\State{Return the updated tableau along with $\mathcal{I}$.}
\end{algorithmic}
\end{algorithm}

\subsection{Quantum chemistry}\label{Sec:QuantumChemistry}

\begin{table}[t!]
\centering
\begin{tabular}{|llrrl|rrr|rrr|rrr|}
\multicolumn{1}{l}{Mol.} & basis & \# & paulis & \multicolumn{1}{r}{rep.} & \multicolumn{3}{c}{largest first}&\multicolumn{3}{c}{independent set}&\multicolumn{3}{c}{sequential}\\
\hline
BeH$_{2}$ & STO3g & 14 & 666 & BK& 33 & 19 & 54& 23 & 22 & 106& 32 & 14 & 106\\
&&&& JW& 33 & 19 & 52& 25 & 16 & 106& 31 & 14 & 106\\
&&&& P& 37 & 16 & 47& 21 & 27 & 106& 31 & 13 & 106\\
\hline
C$_{2}$ & STO3g & 20 & 3079 & BK& 68 & 28 & 204& 68 & 35 & 211& 75 & 34 & 211\\
&&&& JW& 63 & 36 & 204& 70 & 34 & 211& 76 & 33 & 177\\
&&&& P& 63 & 32 & 204& 74 & 30 & 211& 79 & 26 & 211\\
\hline
H$_{2}$O & 6-31G* & 36 & 41915 & BK& 483 & 68 & 667& -- & -- & --& 414 & 76 & 667\\
&&&& JW& 469 & 68 & 667& -- & -- & --& 426 & 75 & 667\\
&&&& P& 469 & 68 & 667& -- & -- & --& 422 & 74 & 667\\
\hline
H$_{2}$O & 6-31G & 26 & 12732 & BK& 204 & 50 & 352& 210 & 46 & 352& 202 & 48 & 352\\
&&&& JW& 204 & 50 & 352& 199 & 49 & 352& 193 & 50 & 352\\
&&&& P& 193 & 56 & 342& 204 & 46 & 352& 202 & 46 & 352\\
\hline
H$_{2}$O & STO3g & 14 & 1086 & BK& 48 & 23 & 72& 43 & 20 & 106& 47 & 20 & 106\\
&&&& JW& 48 & 23 & 72& 45 & 18 & 106& 49 & 16 & 106\\
&&&& P& 48 & 23 & 75& 44 & 18 & 106& 45 & 20 & 106\\
\hline
H$_{2}$O & ccpvdz & 48 & 128793 & BK& -- & -- & --& -- & -- & --& 796 & 116 & 1177\\
&&&& JW& -- & -- & --& -- & -- & --& 802 & 114 & 1177\\
&&&& P& -- & -- & --& -- & -- & --& 821 & 112 & 1177\\
\hline
H$_{2}$ & 6-31G & 8 & 185 & BK& 9 & 20 & 32& 8 & 20 & 37& 9 & 16 & 37\\
&&&& JW& 9 & 20 & 29& 8 & 20 & 37& 11 & 16 & 37\\
&&&& P& 10 & 16 & 29& 8 & 20 & 37& 9 & 16 & 37\\
\hline
H$_{2}$ & STO3g & 4 & 15 & BK& 2 & 7 & 11& 2 & 7 & 11& 2 & 7 & 11\\
&&&& JW& 2 & 7 & 11& 2 & 7 & 11& 2 & 7 & 11\\
&&&& P& 2 & 7 & 11& 2 & 7 & 11& 2 & 7 & 11\\
\hline
HCl & STO3g & 20 & 5851 & BK& 117 & 42 & 199& 149 & 33 & 211& 126 & 33 & 211\\
&&&& JW& 113 & 38 & 184& 144 & 34 & 211& 125 & 36 & 211\\
&&&& P& 115 & 40 & 192& 147 & 32 & 211& 123 & 36 & 211\\
\hline
LiH & STO3g & 12 & 631 & BK& 38 & 10 & 68& 25 & 25 & 79& 38 & 12 & 79\\
&&&& JW& 38 & 11 & 62& 24 & 24 & 79& 35 & 12 & 79\\
&&&& P& 38 & 10 & 68& 24 & 25 & 79& 37 & 10 & 79\\
\hline
NH$_{3}$ & STO3g & 16 & 3057 & BK& 92 & 24 & 137& 86 & 28 & 137& 96 & 26 & 137\\
&&&& JW& 93 & 22 & 137& 87 & 28 & 137& 94 & 24 & 137\\
&&&& P& 96 & 25 & 137& 85 & 28 & 137& 93 & 26 & 137\\
\hline
\end{tabular}

\caption{Problem instances of different molecules when discretized in
  the given bases, along with the number of qubits (\#) and the
  resulting number of Pauli terms in the Hamiltonian. The entries in
  columns for largest-first, independent-set, and sequential greedy
  partitioning methods give the number of sets in the partition, as
  well as the median and maximum size of the sets, respectively, for
  each of the three encodings: Bravyi-Kitaev (BK), Jordan-Wigner (JW), and parity
  (P).}\label{Table:Cliques}
\end{table}

The Hamiltonians we have considered so far were randomly generated,
and may therefore be structurally different from those found in
practical applications. In this section we look at the time evolution
of Hamiltonians arising from fermionic many-body quantum systems. We
use the spin Hamiltonians obtained in~\cite{CHO2020MCa} by using the
second quantization formalism of the fermionic system followed by a
conversion to interacting spin models by applying the Jordan-Wigner,
Bravyi-Kitaev, or parity encodings~\cite{JOR1928Wa,BRA2002Ka}. The
resulting Hamiltonians are expressed as a weighted set of Paulis, as
desired. Table~\ref{Table:Cliques} summarizes the molecular
Hamiltonians, along with the basis sets~\cite{SZA1989Oa} used in the
discretization. In order to apply simultaneous diagonalization we
first need to partition the Hamiltonian terms into sets of commuting
Paulis. For this we use two different greedy coloring strategies
(largest first, and independent set) implemented in
NetworkX~\cite{HAG2008SSa}, along with a custom implementation of a
greedy algorithm in which each of the Paulis is sequentially added to
the first set it commutes with, creating a new set if needed (this
approach, which was also given in~\cite{CRA2019SWPa}, has the
additional advantage that no graph needs to be constructed).  Overall,
as seen in Table~\ref{Table:Cliques}, the three different partitioning
strategies give similar results in terms of number of partitions, as
well as median and maximum partition size. The same applies across the
different encoding schemes, but we assume that these are given; the
partitioning scheme can be freely chosen. Note that the maximum
partition size can be much larger than the number of qubits (terms in
each of the Paulis). In some cases the NetworkX graph coloring
algorithms either ran out of memory or did not return a result in a
reasonable amount of time. Throughout the results we indicate those
cases are by dashes.

Once the terms in the Hamiltonian are partitioned into commuting sets
we can apply the different simulation algorithms to each of the
individual partitions. We compare the diagonalization-based approaches
with direct exponentiation. As before, we apply level-two circuit
optimization as provided in Qiskit to the direct exponentiation
approach only as it was found not to give any improvements in circuit
complexity for the diagonalization-based circuits. We additionally use
the opt strategy described in Section~\ref{Sec:ExpRandom} to determine
the order of the Paulis within each partition. For the \cz\ and direct
methods we additional allow the use of the rnd optimization
strategy. In the determination of the circuit complexity we assume
that the $R_z$ operators has a single-gate implementation. We
determine the total circuit depth by simple adding the depths of the
circuits for each of the partitions. It might be possible to reduce
the depth and single-qubit counts due to potential simplifications at
the circuit boundaries; we expect this reduction to be negligible.

\begin{table}
%\centering
\vspace*{-20pt}\hspace*{-25pt}
\begin{tabular}{|l|rrrrrrrrrrr|}
\hline
{\bf{Method}}& BeH$_{2}$& C$_{2}$& H$_{2}$O& H$_{2}$O& H$_{2}$O& H$_{2}$O& H$_{2}$& H$_{2}$& HCl& LiH& NH$_{3}$\\[-5pt]
& {\tiny{STO3g}}& {\tiny{STO3g}}& {\tiny{6-31G*}}& {\tiny{6-31G}}& {\tiny{STO3g}}& {\tiny{ccpvdz}}& {\tiny{6-31G}}& {\tiny{STO3g}}& {\tiny{STO3g}}& {\tiny{STO3g}}& {\tiny{STO3g}}\\
\hline
cz&2,162&10,232&177,800&45,956&3,370&621,416&438&\cellcolor{highlight}{\color{darkblue}28}&19,236&1,980&9,272\\
cz-rnd&\cellcolor{highlight}{\color{darkblue}1,870}&\cellcolor{highlight}{\color{darkblue}8,930}&\cellcolor{highlight}{\color{darkblue}150,568}&\cellcolor{highlight}{\color{darkblue}39,086}&\cellcolor{highlight}{\color{darkblue}2,884}&\cellcolor{highlight}{\color{darkblue}507,468}&\cellcolor{highlight}{\color{darkblue}392}&\cellcolor{highlight}{\color{darkblue}28}&\cellcolor{highlight}{\color{darkblue}16,536}&\cellcolor{highlight}{\color{darkblue}1,678}&\cellcolor{highlight}{\color{darkblue}7,980}\\
cnot&2,968&14,856&311,792&71,412&4,402&1,177,624&530&34&28,156&2,616&12,672\\
cnot-log2&2,888&14,672&308,400&70,482&4,314&1,168,166&512&32&27,864&2,552&12,516\\
cnot-best&2,832&14,534&305,228&69,792&4,258&1,159,544&498&32&27,472&2,508&12,366\\
greedy-1&2,152&10,190&176,174&45,830&3,372&614,688&438&\cellcolor{highlight}{\color{darkblue}28}&19,084&1,980&9,216\\
greedy-2&2,226&10,358&180,292&46,496&3,412&628,566&434&30&19,606&2,036&9,354\\
direct&3,662&19,732&366,152&94,394&6,210&1,284,042&896&40&39,274&3,276&14,462\\
direct-rnd&3,352&19,390&365,486&93,864&5,750&1,283,034&744&36&38,894&2,882&14,052\\
\hline
\multicolumn{12}{c}{\cnot\ count}\\[8pt]
\hline
cz&2,632&12,252&197,968&53,748&4,144&675,152&597&\cellcolor{highlight}{\color{darkblue}45}&23,071&2,476&11,463\\
cz-rnd&\cellcolor{highlight}{\color{darkblue}2,328}&\cellcolor{highlight}{\color{darkblue}10,860}&\cellcolor{highlight}{\color{darkblue}169,937}&\cellcolor{highlight}{\color{darkblue}46,763}&\cellcolor{highlight}{\color{darkblue}3,645}&\cellcolor{highlight}{\color{darkblue}558,418}&\cellcolor{highlight}{\color{darkblue}547}&\cellcolor{highlight}{\color{darkblue}45}&\cellcolor{highlight}{\color{darkblue}20,282}&\cellcolor{highlight}{\color{darkblue}2,164}&\cellcolor{highlight}{\color{darkblue}10,112}\\
cnot&3,374&16,289&305,025&74,961&5,145&1,124,391&695&61&30,689&3,087&14,531\\
cnot-log2&3,346&16,262&306,391&75,059&5,118&1,130,922&701&60&30,746&3,058&14,469\\
cnot-best&3,317&16,196&306,892&74,851&5,097&1,135,219&697&60&30,593&3,047&14,414\\
greedy-1&2,639&12,343&200,344&54,211&4,170&689,060&598&\cellcolor{highlight}{\color{darkblue}45}&23,135&2,484&11,511\\
greedy-2&2,706&12,418&204,638&54,827&4,165&704,256&585&\cellcolor{highlight}{\color{darkblue}45}&23,574&2,503&11,549\\
direct&5,148&26,591&468,176&124,024&8,640&1,604,480&1,297&67&52,855&4,689&19,896\\
direct-rnd&4,818&26,216&467,617&123,631&8,243&1,603,611&1,140&63&52,427&4,248&19,365\\
\hline
\multicolumn{12}{c}{Circuit depth}\\[8pt]
\hline
cz&1,442&6,197&77,687&23,704&2,376&223,043&411&25&11,527&1,345&6,179\\
cz-rnd&1,442&6,197&77,687&23,704&2,376&223,043&411&25&11,527&1,345&6,179\\
cnot&2,582&10,105&117,795&37,168&4,056&323,863&733&67&18,555&2,425&10,259\\
cnot-log2&2,582&10,105&117,795&37,168&4,056&323,863&733&67&18,555&2,425&10,259\\
cnot-best&2,582&10,105&117,795&37,168&4,056&323,863&733&67&18,555&2,425&10,259\\
greedy-1&\cellcolor{highlight}{\color{darkblue}1,192}&\cellcolor{highlight}{\color{darkblue}5,401}&\cellcolor{highlight}{\color{darkblue}65,891}&\cellcolor{highlight}{\color{darkblue}20,440}&\cellcolor{highlight}{\color{darkblue}2,050}&\cellcolor{highlight}{\color{darkblue}193,249}&\cellcolor{highlight}{\color{darkblue}339}&\cellcolor{highlight}{\color{darkblue}21}&\cellcolor{highlight}{\color{darkblue}10,045}&\cellcolor{highlight}{\color{darkblue}1,117}&\cellcolor{highlight}{\color{darkblue}5,551}\\
greedy-2&1,398&5,623&67,945&21,112&2,236&194,863&403&25&10,411&1,321&5,667\\
direct&3,920&21,393&328,223&93,220&6,880&1,036,434&1,061&43&41,249&3,577&17,097\\
direct-rnd&3,696&21,014&327,977&93,138&6,436&1,036,446&935&43&40,833&3,165&16,621\\
\hline
cz&1,472&6,367&80,767&24,726&2,352&\cellcolor{highlight}{\color{darkblue}228,193}&\cellcolor{highlight}{\color{darkblue}281}&\cellcolor{highlight}{\color{darkblue}19}&11,787&1,411&\cellcolor{highlight}{\color{darkblue}5,747}\\
cz-rnd&1,472&6,367&80,767&24,726&2,352&\cellcolor{highlight}{\color{darkblue}228,193}&\cellcolor{highlight}{\color{darkblue}281}&\cellcolor{highlight}{\color{darkblue}19}&11,787&1,411&\cellcolor{highlight}{\color{darkblue}5,747}\\
cnot&2,642&10,589&118,335&38,334&4,088&326,919&657&71&18,995&2,577&10,811\\
cnot-log2&2,642&10,589&118,335&38,334&4,088&326,919&657&71&18,995&2,577&10,811\\
cnot-best&2,642&10,589&118,335&38,334&4,088&326,919&657&71&18,995&2,577&10,811\\
greedy-1&\cellcolor{highlight}{\color{darkblue}1,328}&\cellcolor{highlight}{\color{darkblue}5,901}&\cellcolor{highlight}{\color{darkblue}80,063}&\cellcolor{highlight}{\color{darkblue}23,966}&\cellcolor{highlight}{\color{darkblue}2,194}&232,743&\cellcolor{highlight}{\color{darkblue}281}&\cellcolor{highlight}{\color{darkblue}19}&\cellcolor{highlight}{\color{darkblue}11,105}&\cellcolor{highlight}{\color{darkblue}1,281}&5,759\\
greedy-2&1,402&6,339&80,657&24,824&2,302&232,527&\cellcolor{highlight}{\color{darkblue}281}&\cellcolor{highlight}{\color{darkblue}19}&11,721&1,425&5,959\\
direct&4,258&22,530&369,157&105,192&7,262&1,206,990&619&\cellcolor{highlight}{\color{darkblue}19}&44,943&3,667&15,715\\
direct-rnd&3,732&22,071&368,375&104,296&6,610&1,206,104&531&\cellcolor{highlight}{\color{darkblue}19}&44,080&3,349&15,150\\
\hline
cz&1,366&\cellcolor{highlight}{\color{darkblue}6,195}&\cellcolor{highlight}{\color{darkblue}79,305}&\cellcolor{highlight}{\color{darkblue}24,224}&\cellcolor{highlight}{\color{darkblue}2,034}&\cellcolor{highlight}{\color{darkblue}231,859}&\cellcolor{highlight}{\color{darkblue}297}&\cellcolor{highlight}{\color{darkblue}19}&\cellcolor{highlight}{\color{darkblue}10,737}&1,243&\cellcolor{highlight}{\color{darkblue}5,757}\\
cz-rnd&1,366&\cellcolor{highlight}{\color{darkblue}6,195}&\cellcolor{highlight}{\color{darkblue}79,305}&\cellcolor{highlight}{\color{darkblue}24,224}&\cellcolor{highlight}{\color{darkblue}2,034}&\cellcolor{highlight}{\color{darkblue}231,859}&\cellcolor{highlight}{\color{darkblue}297}&\cellcolor{highlight}{\color{darkblue}19}&\cellcolor{highlight}{\color{darkblue}10,737}&1,243&\cellcolor{highlight}{\color{darkblue}5,757}\\
cnot&2,646&10,719&118,173&38,272&4,020&329,561&633&71&18,259&2,689&10,371\\
cnot-log2&2,646&10,719&118,173&38,272&4,020&329,561&633&71&18,259&2,689&10,371\\
cnot-best&2,646&10,719&118,173&38,272&4,020&329,561&633&71&18,259&2,689&10,371\\
greedy-1&\cellcolor{highlight}{\color{darkblue}1,334}&6,591&85,931&25,296&2,048&256,419&313&\cellcolor{highlight}{\color{darkblue}19}&11,353&\cellcolor{highlight}{\color{darkblue}1,199}&6,099\\
greedy-2&1,414&6,693&85,821&25,728&2,158&253,069&309&\cellcolor{highlight}{\color{darkblue}19}&11,595&1,295&6,145\\
direct&3,710&20,131&340,798&95,026&6,020&1,115,800&625&\cellcolor{highlight}{\color{darkblue}19}&39,977&3,277&15,926\\
direct-rnd&3,392&19,599&340,765&94,680&5,678&1,115,394&559&\cellcolor{highlight}{\color{darkblue}19}&39,058&2,841&14,978\\
\hline
\multicolumn{12}{c}{Single-qubit count}
\end{tabular}
\caption{Results based on the greedy sequential partitioning method,
  with the \cnot\ count and circuit depth for the Jordan-Wigner
  encoding, as well as the single-qubit count for the Jordan-Wigner,
  Bravyi-Kitaev, and parity encodings, respectively. The best values
  are highlighted.}\label{Table:SequentialJW}
\end{table}

The circuit complexity when partitioning the Hamiltonians with the
greedy sequential approach is shown in
Table~\ref{Table:SequentialJW}. The first thing we note is that the
\cnot{}-based diagonalization performs substantially worse than both
the \cz{} and greedy-based approaches, in stark contrast to the
results on random Paulis in Section~\ref{Sec:ExpRandom}, where the
performance of the \cnot\ approach closely matched that of \cz{}. This
could be caused by the fill-in during normalization the Z block, which
is present in the \cnot\ approach, but absent in the other two
diagonalization approaches. Despite its relatively poor performance,
the \cnot\ method still consistently outperforms direct exponentiation
in terms of circuit depth and \cnot\ count. This difference is much
more substantial for the \cz{} and greedy approaches, where we see
reductions of up to $50\%$. Application of the rnd optimization
strategy shows good improvements for \cz\ diagonalization. For the
direct method, however, the improvements are only marginal.

When applied to randomized Hamiltonians, we saw that the greedy
approach generally required fewer \cnot\ gates than the \cz\
approach. We concluded that this was mostly due to the \cnot\ count in
the diagonalization part of the circuit, rather than exponentiation
part. For the experiments here we see that the difference between the
different methods is minimal at best.  A similar pattern emerges for
the remaining experiments using different partitioning algorithms. The
\cnot\ counts and circuit depths for these experiments are summarized
in Tables~\ref{Table:SimulateOtherCNot}
and~\ref{Table:SimulateOtherDepth}, respectively. Due to their poor
relative performance we omit the results for \cnot{}-based
diagonalization, and also leave out the single-qubit counts, as these
are very similar to the ones given in Table~\ref{Table:SequentialJW}.

Somewhat surprisingly, we see that across the different simulation
methods the results for independent-set greedy partitioning are
substantially worse than those of the other two partitioning methods,
despite the similarity of the partition metrics shown in
Table~\ref{Table:Cliques}.  To get a better understanding of what
causes this difference we plot the number of \cnot\ gates for the
diagonalization circuit for each of the partitions against the size of
the partition. The resulting plots, shown in
Figure~\ref{Fig:PartitioningCNot}, indicate that the largest first and
sequential strategies behave very similar. The independent set
coloring strategy on the other hand often requires a substantial
larger number of \cnot\ gates for small partitions. This difference is
seen across all molecules and encoding schemes, but is especially
apparent for the Jordan-Wigner encoding. Given that, among the three
coloring strategies considered, the independent-set strategy is the
most time-consuming anyway, we would not recommend its use in this
setting.

Overall, we see from the results in
Tables~\ref{Table:SimulateOtherCNot}
and~\ref{Table:SimulateOtherDepth} that the circuits for the
Hamiltonians based on the Jordan-Wigner encoding tend to be simpler
than those for the Bravyi-Kitaev and parity encodings. Finally, we
note that each subset of commuting Paulis can be simulated
independently. It is therefore possible to choose a different method
per partition. For instance, we could select the direct method for
some partitions and the diagonalization-based method for others. To
implement this idea we modified the experiments based on
diagonalization such that the direct method was used if it was found
to have a circuit with fewer \cnot\ gates. The improvements obtained
with this approach were very minor and in fact showed that in most
cases the diagonalized-based approach was not outperformed by the
direct method on any of partitions.

\begin{table}
%\centering
\hspace*{-21pt}
\begin{tabular}{|l|rrrrrrrrrrr|}
\hline
{\bf{Method}}& BeH$_{2}$& C$_{2}$& H$_{2}$O& H$_{2}$O& H$_{2}$O& H$_{2}$O& H$_{2}$& H$_{2}$& HCl& LiH& NH$_{3}$\\[-5pt]
& {\tiny{STO3g}}& {\tiny{STO3g}}& {\tiny{6-31G*}}& {\tiny{6-31G}}& {\tiny{STO3g}}& {\tiny{ccpvdz}}& {\tiny{6-31G}}& {\tiny{STO3g}}& {\tiny{STO3g}}& {\tiny{STO3g}}& {\tiny{STO3g}}\\
\hline
cz&2,162&10,232&177,800&45,956&\cellcolor{highlight}{\color{darkblue}3,370}&621,416&438&\cellcolor{highlight}{\color{darkblue}28}&19,236&\cellcolor{highlight}{\color{darkblue}1,980}&9,272\\
greedy-1&\cellcolor{highlight}{\color{darkblue}2,152}&\cellcolor{highlight}{\color{darkblue}10,190}&\cellcolor{highlight}{\color{darkblue}176,174}&\cellcolor{highlight}{\color{darkblue}45,830}&3,372&\cellcolor{highlight}{\color{darkblue}614,688}&438&\cellcolor{highlight}{\color{darkblue}28}&\cellcolor{highlight}{\color{darkblue}19,084}&\cellcolor{highlight}{\color{darkblue}1,980}&\cellcolor{highlight}{\color{darkblue}9,216}\\
greedy-2&2,226&10,358&180,292&46,496&3,412&628,566&\cellcolor{highlight}{\color{darkblue}434}&30&19,606&2,036&9,354\\
direct&3,662&19,732&366,152&94,394&6,210&1,284,042&896&40&39,274&3,276&14,462\\
\hline
cz&2,682&14,492&237,834&63,612&4,316&728,514&478&\cellcolor{highlight}{\color{darkblue}24}&27,238&2,602&12,524\\
greedy-1&\cellcolor{highlight}{\color{darkblue}2,618}&\cellcolor{highlight}{\color{darkblue}13,932}&\cellcolor{highlight}{\color{darkblue}223,976}&\cellcolor{highlight}{\color{darkblue}61,042}&\cellcolor{highlight}{\color{darkblue}4,264}&\cellcolor{highlight}{\color{darkblue}692,734}&478&\cellcolor{highlight}{\color{darkblue}24}&\cellcolor{highlight}{\color{darkblue}26,114}&\cellcolor{highlight}{\color{darkblue}2,552}&\cellcolor{highlight}{\color{darkblue}12,144}\\
greedy-2&2,678&14,056&228,688&62,740&4,380&704,886&\cellcolor{highlight}{\color{darkblue}450}&\cellcolor{highlight}{\color{darkblue}24}&26,672&2,566&12,238\\
direct&4,162&23,032&391,524&108,832&7,256&1,301,306&788&30&46,888&3,862&15,928\\
\hline
cz&2,778&13,350&230,076&61,308&4,302&777,212&488&\cellcolor{highlight}{\color{darkblue}26}&24,096&2,492&12,122\\
greedy-1&\cellcolor{highlight}{\color{darkblue}2,716}&\cellcolor{highlight}{\color{darkblue}12,628}&\cellcolor{highlight}{\color{darkblue}211,994}&\cellcolor{highlight}{\color{darkblue}57,628}&\cellcolor{highlight}{\color{darkblue}4,200}&\cellcolor{highlight}{\color{darkblue}714,516}&\cellcolor{highlight}{\color{darkblue}484}&\cellcolor{highlight}{\color{darkblue}26}&\cellcolor{highlight}{\color{darkblue}23,182}&\cellcolor{highlight}{\color{darkblue}2,442}&\cellcolor{highlight}{\color{darkblue}11,658}\\
greedy-2&2,792&12,996&217,898&59,424&4,226&734,266&486&\cellcolor{highlight}{\color{darkblue}26}&23,792&2,450&11,730\\
direct&4,250&22,062&423,198&111,836&6,964&1,509,184&844&32&45,018&3,658&16,434\\
\hline
\multicolumn{12}{c}{sequential}\\[6pt]
\hline
cz&2,140&10,040&179,120&\cellcolor{highlight}{\color{darkblue}45,912}&3,296&--&\cellcolor{highlight}{\color{darkblue}418}&\cellcolor{highlight}{\color{darkblue}28}&17,944&2,034&9,116\\
greedy-1&\cellcolor{highlight}{\color{darkblue}2,136}&9,966&\cellcolor{highlight}{\color{darkblue}177,984}&46,060&\cellcolor{highlight}{\color{darkblue}3,286}&--&\cellcolor{highlight}{\color{darkblue}418}&\cellcolor{highlight}{\color{darkblue}28}&\cellcolor{highlight}{\color{darkblue}17,738}&\cellcolor{highlight}{\color{darkblue}2,032}&\cellcolor{highlight}{\color{darkblue}9,082}\\
greedy-2&2,204&\cellcolor{highlight}{\color{darkblue}9,908}&182,060&47,012&3,358&--&420&30&18,204&2,100&9,246\\
direct&3,542&19,512&378,236&95,912&5,920&--&864&40&38,596&3,296&14,730\\
\hline
cz&2,806&13,988&231,924&60,952&4,340&--&\cellcolor{highlight}{\color{darkblue}468}&\cellcolor{highlight}{\color{darkblue}24}&26,330&2,556&12,026\\
greedy-1&\cellcolor{highlight}{\color{darkblue}2,752}&\cellcolor{highlight}{\color{darkblue}13,500}&\cellcolor{highlight}{\color{darkblue}222,066}&\cellcolor{highlight}{\color{darkblue}58,570}&\cellcolor{highlight}{\color{darkblue}4,276}&--&\cellcolor{highlight}{\color{darkblue}468}&\cellcolor{highlight}{\color{darkblue}24}&\cellcolor{highlight}{\color{darkblue}25,474}&\cellcolor{highlight}{\color{darkblue}2,536}&\cellcolor{highlight}{\color{darkblue}11,728}\\
greedy-2&2,842&13,706&227,460&60,128&4,294&--&472&\cellcolor{highlight}{\color{darkblue}24}&25,802&2,540&11,858\\
direct&4,340&23,018&395,596&108,644&7,064&--&788&30&46,152&3,904&16,154\\
\hline
cz&2,838&12,630&224,882&58,914&4,310&--&534&\cellcolor{highlight}{\color{darkblue}26}&23,638&2,462&12,218\\
greedy-1&\cellcolor{highlight}{\color{darkblue}2,724}&\cellcolor{highlight}{\color{darkblue}12,208}&\cellcolor{highlight}{\color{darkblue}210,390}&\cellcolor{highlight}{\color{darkblue}55,478}&\cellcolor{highlight}{\color{darkblue}4,200}&--&\cellcolor{highlight}{\color{darkblue}520}&\cellcolor{highlight}{\color{darkblue}26}&\cellcolor{highlight}{\color{darkblue}22,656}&\cellcolor{highlight}{\color{darkblue}2,404}&\cellcolor{highlight}{\color{darkblue}11,684}\\
greedy-2&2,804&12,494&218,376&57,012&4,252&--&522&\cellcolor{highlight}{\color{darkblue}26}&23,338&2,420&11,914\\
direct&4,168&21,594&438,016&111,836&7,066&--&802&32&44,208&3,570&16,320\\
\hline
\multicolumn{12}{c}{largest first}\\[6pt]
\hline
cz&3,124&18,490&--&94,094&5,178&--&522&\cellcolor{highlight}{\color{darkblue}28}&38,610&2,764&16,154\\
greedy-1&2,862&\cellcolor{highlight}{\color{darkblue}16,528}&--&\cellcolor{highlight}{\color{darkblue}84,402}&\cellcolor{highlight}{\color{darkblue}4,754}&--&\cellcolor{highlight}{\color{darkblue}500}&\cellcolor{highlight}{\color{darkblue}28}&\cellcolor{highlight}{\color{darkblue}33,648}&2,476&\cellcolor{highlight}{\color{darkblue}14,572}\\
greedy-2&\cellcolor{highlight}{\color{darkblue}2,810}&16,828&--&84,956&4,760&--&502&30&33,790&\cellcolor{highlight}{\color{darkblue}2,460}&14,578\\
direct&4,596&26,876&--&144,680&7,708&--&940&48&55,002&4,142&19,430\\
\hline
cz&3,192&20,096&--&100,522&5,764&--&510&\cellcolor{highlight}{\color{darkblue}24}&42,154&2,808&16,374\\
greedy-1&\cellcolor{highlight}{\color{darkblue}2,988}&\cellcolor{highlight}{\color{darkblue}18,214}&--&\cellcolor{highlight}{\color{darkblue}90,070}&5,434&--&504&\cellcolor{highlight}{\color{darkblue}24}&\cellcolor{highlight}{\color{darkblue}37,584}&2,706&15,110\\
greedy-2&3,056&18,388&--&91,384&\cellcolor{highlight}{\color{darkblue}5,430}&--&\cellcolor{highlight}{\color{darkblue}490}&\cellcolor{highlight}{\color{darkblue}24}&37,874&\cellcolor{highlight}{\color{darkblue}2,654}&\cellcolor{highlight}{\color{darkblue}15,044}\\
direct&5,028&29,124&--&145,340&8,500&--&876&30&59,796&4,582&19,402\\
\hline
cz&2,930&20,270&--&100,352&5,478&--&\cellcolor{highlight}{\color{darkblue}512}&\cellcolor{highlight}{\color{darkblue}26}&39,898&2,538&16,252\\
greedy-1&\cellcolor{highlight}{\color{darkblue}2,774}&\cellcolor{highlight}{\color{darkblue}17,724}&--&\cellcolor{highlight}{\color{darkblue}88,224}&\cellcolor{highlight}{\color{darkblue}5,136}&--&520&\cellcolor{highlight}{\color{darkblue}26}&\cellcolor{highlight}{\color{darkblue}34,830}&\cellcolor{highlight}{\color{darkblue}2,434}&14,686\\
greedy-2&2,778&17,908&--&89,642&5,180&--&524&\cellcolor{highlight}{\color{darkblue}26}&35,304&2,488&\cellcolor{highlight}{\color{darkblue}14,676}\\
direct&4,938&28,008&--&149,600&8,164&--&820&32&57,350&4,274&20,216\\
\hline
\multicolumn{12}{c}{independent set}
\end{tabular}
\caption{The \cnot\ counts for different exponentiation methods for
  the sequential, largest-first, and independent-set partitioning
  methods. The results per method correspond to different encodings
  given by, from top to bottom, Jordan-Wigner, Bravyi-Kitaev, and
  parity. The best values are
  highlighted.}\label{Table:SimulateOtherCNot}
\end{table}

\begin{table}
%\centering
\hspace*{-26pt}
\begin{tabular}{|l|rrrrrrrrrrr|}
\hline
{\bf{Method}}& BeH$_{2}$& C$_{2}$& H$_{2}$O& H$_{2}$O& H$_{2}$O& H$_{2}$O& H$_{2}$& H$_{2}$& HCl& LiH& NH$_{3}$\\[-5pt]
& {\tiny{STO3g}}& {\tiny{STO3g}}& {\tiny{6-31G*}}& {\tiny{6-31G}}& {\tiny{STO3g}}& {\tiny{ccpvdz}}& {\tiny{6-31G}}& {\tiny{STO3g}}& {\tiny{STO3g}}& {\tiny{STO3g}}& {\tiny{STO3g}}\\
\hline
cz&\cellcolor{highlight}{\color{darkblue}2,632}&\cellcolor{highlight}{\color{darkblue}12,252}&\cellcolor{highlight}{\color{darkblue}197,968}&\cellcolor{highlight}{\color{darkblue}53,748}&\cellcolor{highlight}{\color{darkblue}4,144}&\cellcolor{highlight}{\color{darkblue}675,152}&597&\cellcolor{highlight}{\color{darkblue}45}&\cellcolor{highlight}{\color{darkblue}23,071}&\cellcolor{highlight}{\color{darkblue}2,476}&\cellcolor{highlight}{\color{darkblue}11,463}\\
greedy-1&2,639&12,343&200,344&54,211&4,170&689,060&598&\cellcolor{highlight}{\color{darkblue}45}&23,135&2,484&11,511\\
greedy-2&2,706&12,418&204,638&54,827&4,165&704,256&\cellcolor{highlight}{\color{darkblue}585}&\cellcolor{highlight}{\color{darkblue}45}&23,574&2,503&11,549\\
direct&5,148&26,591&468,176&124,024&8,640&1,604,480&1,297&67&52,855&4,689&19,896\\
\hline
cz&\cellcolor{highlight}{\color{darkblue}3,106}&\cellcolor{highlight}{\color{darkblue}15,385}&\cellcolor{highlight}{\color{darkblue}232,718}&\cellcolor{highlight}{\color{darkblue}66,870}&\cellcolor{highlight}{\color{darkblue}5,031}&\cellcolor{highlight}{\color{darkblue}713,895}&665&\cellcolor{highlight}{\color{darkblue}40}&\cellcolor{highlight}{\color{darkblue}29,451}&\cellcolor{highlight}{\color{darkblue}3,012}&\cellcolor{highlight}{\color{darkblue}14,476}\\
greedy-1&3,126&15,760&241,321&68,108&5,128&747,577&669&\cellcolor{highlight}{\color{darkblue}40}&29,949&3,045&14,510\\
greedy-2&3,127&15,778&246,652&69,765&5,170&761,381&\cellcolor{highlight}{\color{darkblue}632}&\cellcolor{highlight}{\color{darkblue}40}&30,479&3,035&14,575\\
direct&5,462&28,990&480,574&135,023&9,339&1,579,383&1,082&47&58,898&5,055&20,482\\
\hline
cz&\cellcolor{highlight}{\color{darkblue}3,211}&\cellcolor{highlight}{\color{darkblue}14,547}&\cellcolor{highlight}{\color{darkblue}222,324}&\cellcolor{highlight}{\color{darkblue}64,066}&\cellcolor{highlight}{\color{darkblue}5,067}&\cellcolor{highlight}{\color{darkblue}723,684}&\cellcolor{highlight}{\color{darkblue}675}&\cellcolor{highlight}{\color{darkblue}42}&\cellcolor{highlight}{\color{darkblue}27,072}&2,972&\cellcolor{highlight}{\color{darkblue}14,006}\\
greedy-1&3,226&14,794&230,178&65,109&5,077&761,581&691&\cellcolor{highlight}{\color{darkblue}42}&27,548&2,979&14,199\\
greedy-2&3,263&15,023&236,216&66,650&5,078&783,561&685&\cellcolor{highlight}{\color{darkblue}42}&28,017&\cellcolor{highlight}{\color{darkblue}2,942}&14,081\\
direct&5,482&27,492&505,346&135,925&8,905&1,768,728&1,109&49&55,852&4,816&20,747\\
\hline
\multicolumn{12}{c}{sequential}\\[6pt]
\hline
cz&2,637&12,047&\cellcolor{highlight}{\color{darkblue}202,774}&\cellcolor{highlight}{\color{darkblue}54,353}&\cellcolor{highlight}{\color{darkblue}4,114}&--&\cellcolor{highlight}{\color{darkblue}578}&\cellcolor{highlight}{\color{darkblue}45}&22,043&\cellcolor{highlight}{\color{darkblue}2,526}&\cellcolor{highlight}{\color{darkblue}11,410}\\
greedy-1&\cellcolor{highlight}{\color{darkblue}2,625}&12,073&205,027&54,882&4,143&--&582&\cellcolor{highlight}{\color{darkblue}45}&\cellcolor{highlight}{\color{darkblue}21,995}&2,534&11,443\\
greedy-2&2,669&\cellcolor{highlight}{\color{darkblue}11,976}&208,817&55,735&4,169&--&587&\cellcolor{highlight}{\color{darkblue}45}&22,348&2,572&11,541\\
direct&5,055&25,724&472,407&124,417&8,341&--&1,201&67&51,747&4,720&19,977\\
\hline
cz&3,249&\cellcolor{highlight}{\color{darkblue}15,164}&\cellcolor{highlight}{\color{darkblue}235,537}&\cellcolor{highlight}{\color{darkblue}65,397}&\cellcolor{highlight}{\color{darkblue}5,087}&--&658&\cellcolor{highlight}{\color{darkblue}40}&\cellcolor{highlight}{\color{darkblue}29,047}&\cellcolor{highlight}{\color{darkblue}2,976}&\cellcolor{highlight}{\color{darkblue}14,158}\\
greedy-1&\cellcolor{highlight}{\color{darkblue}3,220}&15,452&242,472&66,241&5,113&--&\cellcolor{highlight}{\color{darkblue}657}&\cellcolor{highlight}{\color{darkblue}40}&29,471&3,020&14,186\\
greedy-2&3,267&15,697&248,194&67,843&5,102&--&659&\cellcolor{highlight}{\color{darkblue}40}&29,630&2,988&14,238\\
direct&5,623&28,708&481,724&134,483&9,149&--&1,084&47&58,032&5,176&20,661\\
\hline
cz&\cellcolor{highlight}{\color{darkblue}3,243}&\cellcolor{highlight}{\color{darkblue}14,178}&\cellcolor{highlight}{\color{darkblue}226,834}&\cellcolor{highlight}{\color{darkblue}63,108}&\cellcolor{highlight}{\color{darkblue}5,056}&--&723&\cellcolor{highlight}{\color{darkblue}42}&\cellcolor{highlight}{\color{darkblue}26,766}&2,962&\cellcolor{highlight}{\color{darkblue}14,132}\\
greedy-1&3,245&14,492&231,382&63,627&5,102&--&720&\cellcolor{highlight}{\color{darkblue}42}&26,975&\cellcolor{highlight}{\color{darkblue}2,950}&14,185\\
greedy-2&3,319&14,700&239,259&65,045&5,151&--&\cellcolor{highlight}{\color{darkblue}713}&\cellcolor{highlight}{\color{darkblue}42}&27,482&2,953&14,233\\
direct&5,421&26,792&514,274&134,831&9,000&--&1,082&49&54,649&4,717&20,585\\
\hline
\multicolumn{12}{c}{largest first}\\[6pt]
\hline
cz&3,445&\cellcolor{highlight}{\color{darkblue}17,245}&--&\cellcolor{highlight}{\color{darkblue}81,963}&5,575&--&666&\cellcolor{highlight}{\color{darkblue}45}&\cellcolor{highlight}{\color{darkblue}34,713}&3,150&\cellcolor{highlight}{\color{darkblue}16,186}\\
greedy-1&3,397&18,131&--&89,008&5,544&--&667&\cellcolor{highlight}{\color{darkblue}45}&36,317&2,997&16,567\\
greedy-2&\cellcolor{highlight}{\color{darkblue}3,316}&18,485&--&90,229&\cellcolor{highlight}{\color{darkblue}5,521}&--&\cellcolor{highlight}{\color{darkblue}645}&\cellcolor{highlight}{\color{darkblue}45}&36,636&\cellcolor{highlight}{\color{darkblue}2,996}&16,682\\
direct&5,688&31,457&--&162,892&9,457&--&1,241&75&63,620&5,135&23,390\\
\hline
cz&\cellcolor{highlight}{\color{darkblue}3,369}&\cellcolor{highlight}{\color{darkblue}18,867}&--&\cellcolor{highlight}{\color{darkblue}87,745}&\cellcolor{highlight}{\color{darkblue}6,058}&--&686&\cellcolor{highlight}{\color{darkblue}40}&\cellcolor{highlight}{\color{darkblue}38,437}&3,167&\cellcolor{highlight}{\color{darkblue}16,917}\\
greedy-1&3,513&19,784&--&94,168&6,239&--&700&\cellcolor{highlight}{\color{darkblue}40}&40,071&3,218&17,381\\
greedy-2&3,559&20,031&--&96,185&6,153&--&\cellcolor{highlight}{\color{darkblue}673}&\cellcolor{highlight}{\color{darkblue}40}&40,559&\cellcolor{highlight}{\color{darkblue}3,151}&17,357\\
direct&6,172&33,873&--&164,848&10,275&--&1,156&47&68,945&5,638&23,333\\
\hline
cz&3,338&\cellcolor{highlight}{\color{darkblue}18,886}&--&\cellcolor{highlight}{\color{darkblue}90,278}&\cellcolor{highlight}{\color{darkblue}5,920}&--&\cellcolor{highlight}{\color{darkblue}681}&\cellcolor{highlight}{\color{darkblue}42}&\cellcolor{highlight}{\color{darkblue}37,505}&3,002&16,932\\
greedy-1&\cellcolor{highlight}{\color{darkblue}3,300}&19,683&--&93,367&5,936&--&721&\cellcolor{highlight}{\color{darkblue}42}&38,179&\cellcolor{highlight}{\color{darkblue}2,999}&16,961\\
greedy-2&3,304&19,882&--&95,243&5,956&--&720&\cellcolor{highlight}{\color{darkblue}42}&38,644&3,070&\cellcolor{highlight}{\color{darkblue}16,893}\\
direct&6,068&32,496&--&168,310&9,856&--&1,082&49&66,235&5,335&24,195\\
\hline
\multicolumn{12}{c}{independent set}
\end{tabular}
\caption{The circuit depth for different exponentiation methods for
  the sequential, largest-first, and independent-set partitioning
  methods. The results per method correspond to different encodings
  given by, from top to bottom, Jordan-Wigner, Bravyi-Kitaev, and
  parity. The best values are
  highlighted.}\label{Table:SimulateOtherDepth}
\end{table}

\begin{figure}
\centering
\begin{tabular}{ccc}
\hspace*{23pt}{\small{Bravyi-Kitaev}} & \hspace*{23pt}{\small{Jordan-Wigner}} & \hspace*{23pt}{\small{parity}}\\[1pt]
\includegraphics[width=0.3\textwidth]{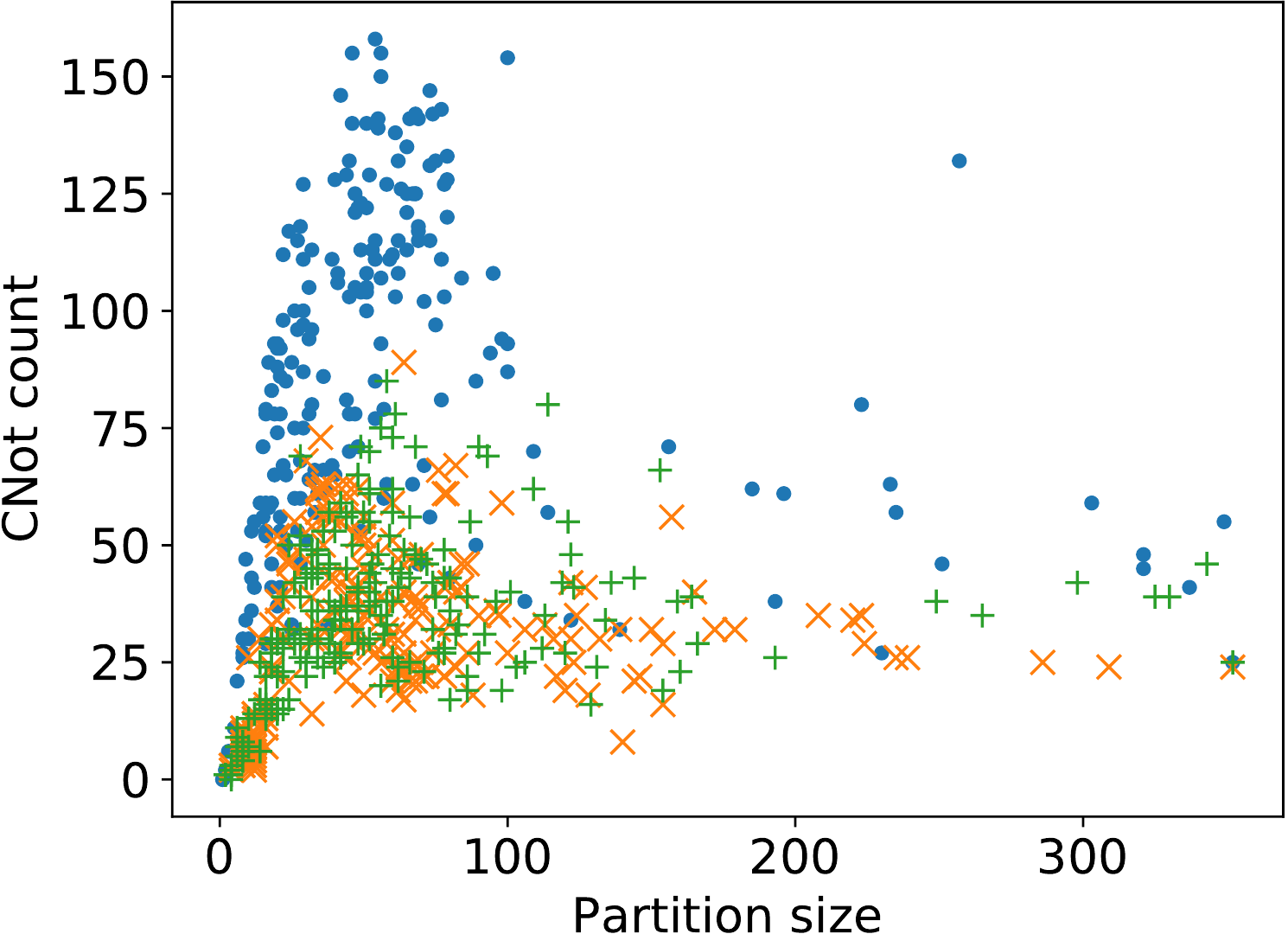}&
\includegraphics[width=0.3\textwidth]{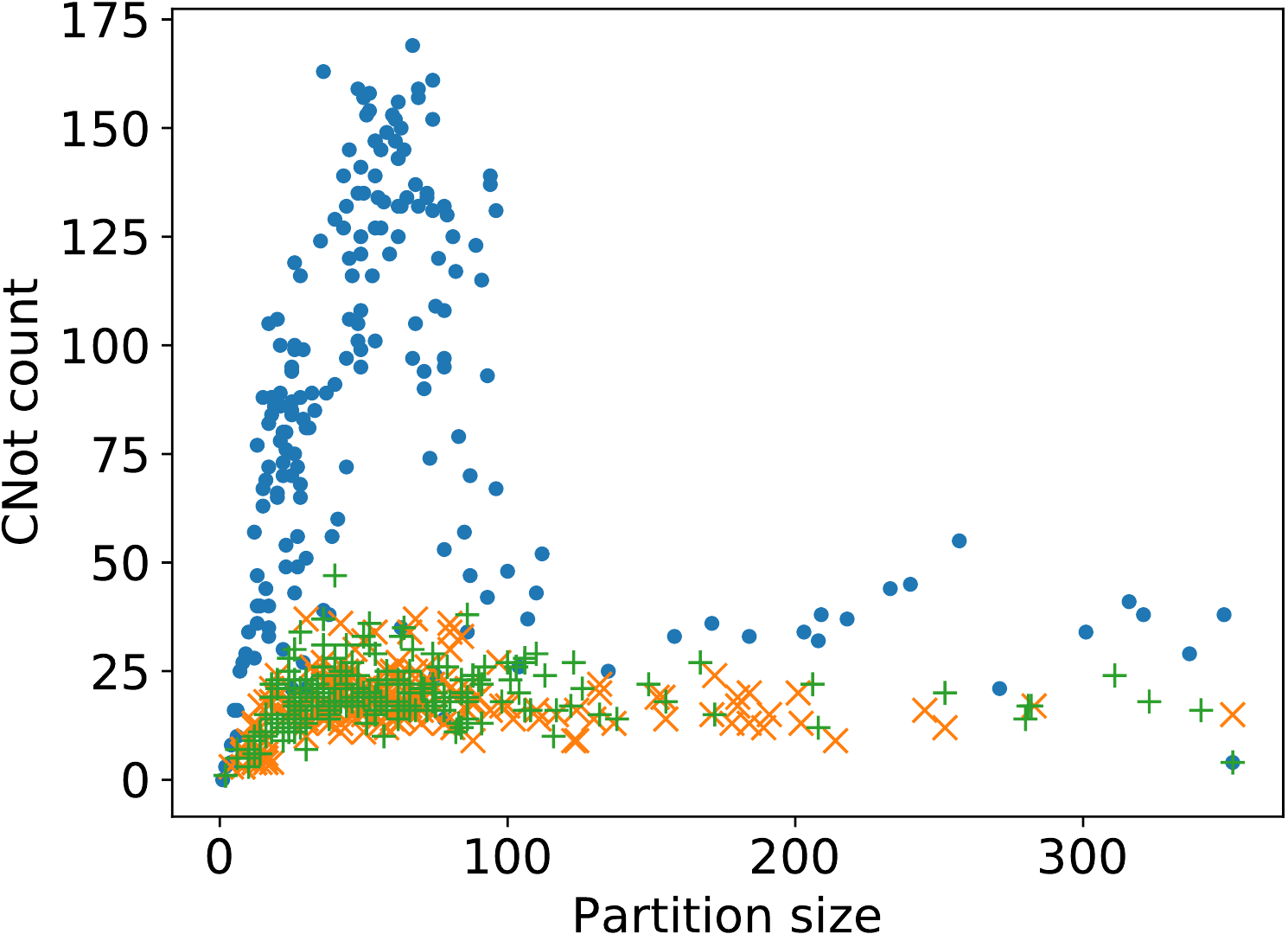}&
\includegraphics[width=0.3\textwidth]{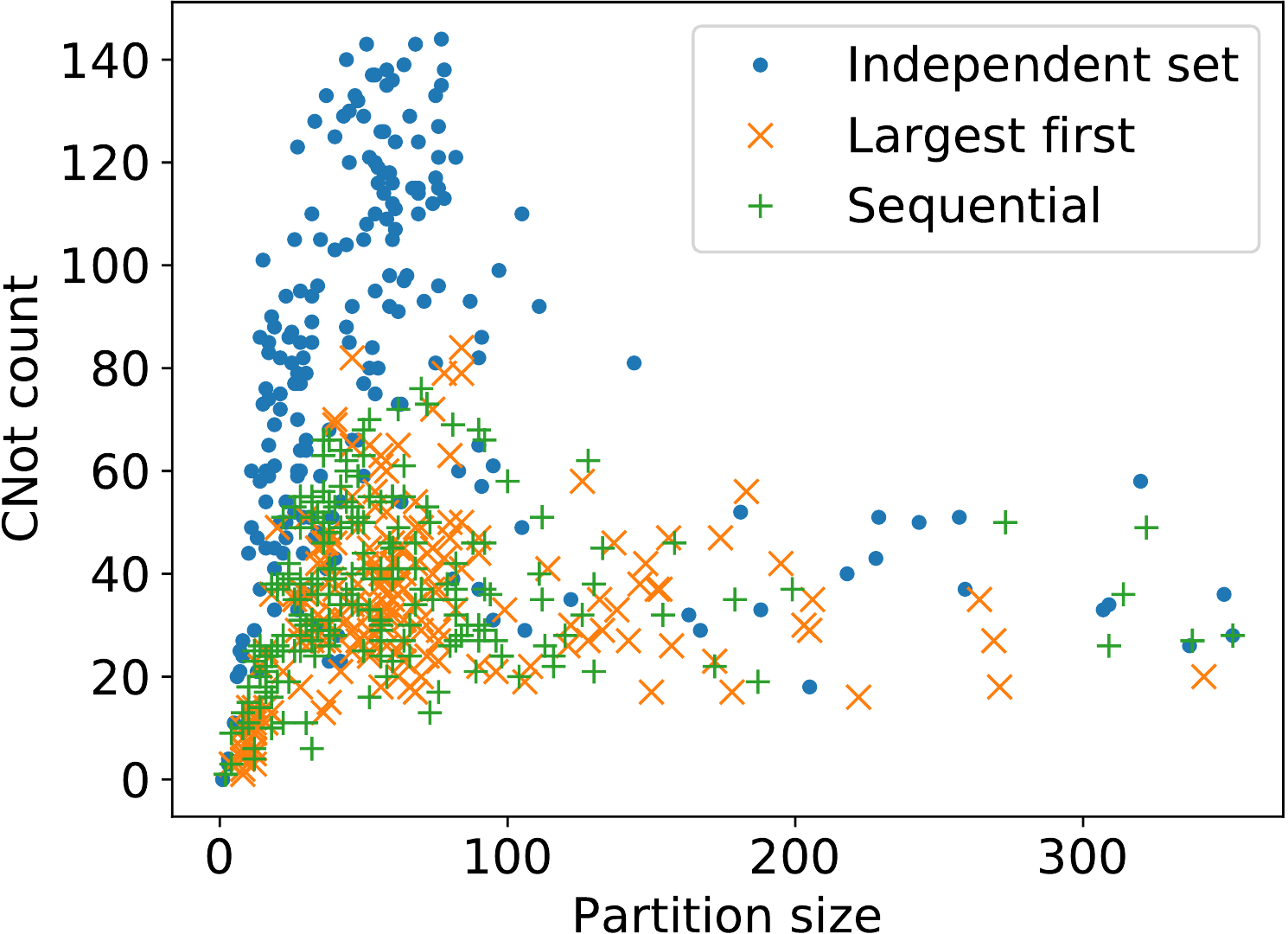}\\[3pt]
\includegraphics[width=0.3\textwidth]{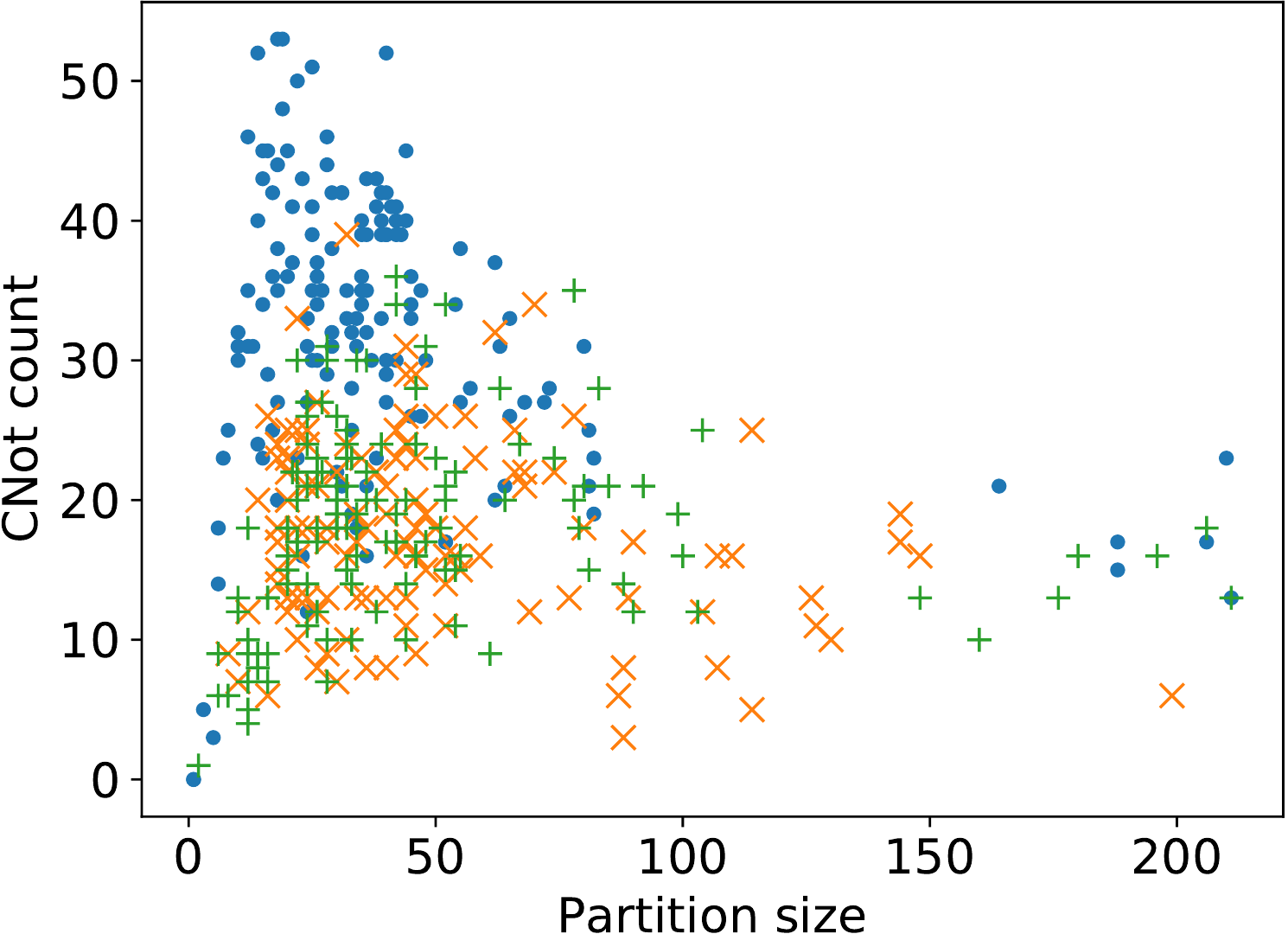}&
\includegraphics[width=0.3\textwidth]{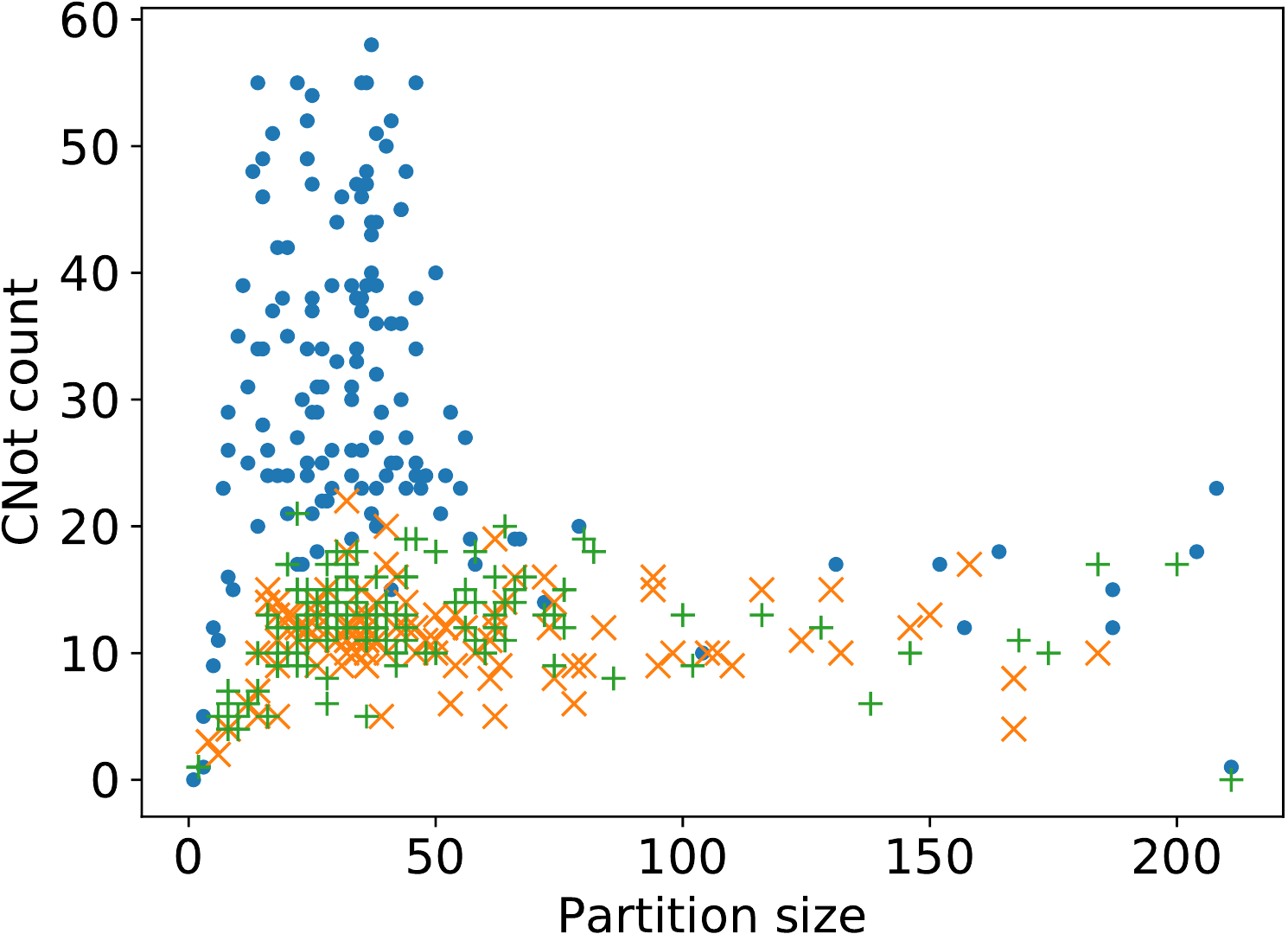}&
\includegraphics[width=0.3\textwidth]{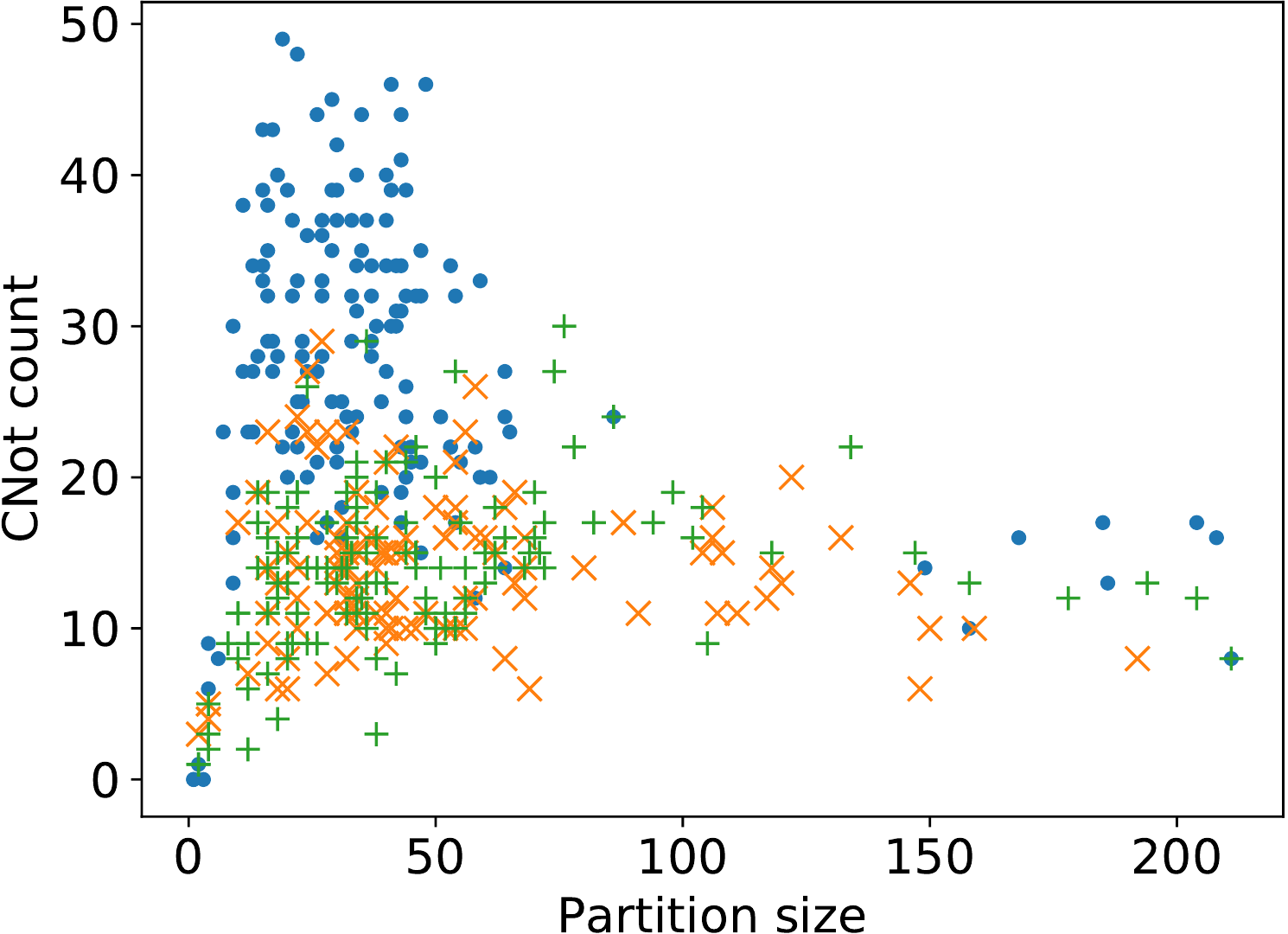}\\[3pt]
\end{tabular}
\caption{The \cnot\ count for the diagonalization circuit for each
  partition, plotted against the partition size for different graph
  coloring strategies. The top and bottom plots show the results for
  H$_2$O using the 6-31G basis and \cz\ diagonalization, and HCl using
  the STO3g basis and greedy-2 diagonalization respectively, for
  different encodings.}\label{Fig:PartitioningCNot}
\end{figure}

\section{Conclusions}\label{Sec:Conclusions}

In this paper we have shown that the use of simultaneous
diagonalization for Hamiltonian simulation can yield substantial
reduction of circuit complexity in terms of \cnot\ count and circuit
depth, compared to direct exponentiation of the individual Pauli
operators; to the best of our knowledge, this is the first time
simultaneous diagonalization has been used to reduce the circuit
complexity in Hamiltonian simulation. The proposed approach first
partitions the Pauli operators into sets of mutually commuting
operators. We used two different strategies provided by the NetworkX
package (independent set and largest first) and compared them against
a pure greedy scheme in which Paulis are added sequentially to the
first partition whose elements it commutes with. Given the need to
instantiate the entire commutativity graph in NetworkX prior to
coloring, the latter strategy is clearly favorable in terms of
computational complexity. For synthetic test problems we found the
three strategies to have similar performance, but a clear difference
was found in application to problems in quantum chemistry, where the
independent-set strategy was found to perform substantially worse
compared to the other two.

The next step is to generate circuits that simultaneously diagonalize
the operators in each set of commuting operators in the
partition. This can be done by representing the Pauli operators in a
tableau form consisting of X and Z blocks, along with a sign vector.
The operators are diagonalized when all entries in the X block are
eliminated using appropriate Clifford operators along with row and
column manipulations.  We introduce novel elimination schemes that
first diagonalize the X block using row operations and Hadamard gates
only. When applied to tableaus with full column rank, the resulting
schemes give circuits consisting of sequence of H-S-CZ-H and
H-S-CX-S-H gates respectively. The introduction of column-based
elimination of entries in the Z block can help reduce the \cnot\
count, and may have separate application in representing stabilizer
states.

We apply the proposed techniques to random sets of commuting Pauli
operators as well as practical problems arising in quantum
chemistry. To facilitate the generation of random test problems we
introduce an efficient new algorithm for uniformly sampling generator
sets of commuting Paulis. The resulting insights also lead to a
compact and unique representation in the form of a binary
$n\times n+2$ matrix for sets of commuting Paulis that can be
generated using the same generator set. This construction can also be
used in the representation of stabilizer states. For the chemistry
problems we show that the \cnot\ count can be reduced by a factor of
two to three compared to the direct approach. The circuit depth is
generally halved, but this may be further improved when considering
the mapping of the circuits to systems with limited qubit
connectivity.

\section*{Acknowledgments}

EvdB would like to thank Sergey Bravyi and Andrew Cross for useful
discussions.  The Hamiltonians used for the chemistry experiments were
kindly provided by Antonio Mezzacapo~\cite{CHO2020MCa}.

\bibliography{bibliography}

\begin{thebibliography}{35}
\providecommand{\natexlab}[1]{#1}
\providecommand{\url}[1]{\texttt{#1}}
\expandafter\ifx\csname urlstyle\endcsname\relax
  \providecommand{\doi}[1]{doi: #1}\else
  \providecommand{\doi}{doi: \begingroup \urlstyle{rm}\Url}\fi

\bibitem[Aaronson and Gottesman(2004)]{AAR2004Ga}
Scott Aaronson and Daniel Gottesman.
\newblock Improved simulation of stabilizer circuits.
\newblock \emph{Physical Review A}, \penalty0 (70):\penalty0 052328, 2004.
\newblock \doi{10.1103/PhysRevA.70.052328}.

\bibitem[Abraham et~al.(2019)Abraham, Akhalwaya, Aleksandrowicz, Alexander,
  Alexandrowics, Arbel, Asfaw, Azaustre, Barkoutsos, Barron, Bello, Ben-Haim,
  Bevenius, Bishop, Bosch, Bucher, CZ, Cabrera, Calpin, Capelluto, Carballo,
  Carrascal, Chen, Chen, Chen, Chow, Claus, Clauss, Cross, Cross, Cruz-Benito,
  Cryoris, Culver, C{\'o}rcoles-Gonzales, Dague, Dartiailh, Davila, Ding,
  Dumitrescu, Dumon, Duran, Eendebak, Egger, Everitt, Fern{\'a}ndez, Frisch,
  Fuhrer, Gould, Gacon, Gadi, Gago, Gambetta, Garcia, Garion, Gawel-Kus,
  Gomez-Mosquera, de~la Puente~Gonz{\'a}lez, Greenberg, Gunnels, Haide,
  Hamamura, Havlicek, Hellmers, Herok, Horii, Howington, Hu, Hu, Imai,
  Imamichi, Iten, Itoko, Javadi-Abhari, Jessica, Johns, Kanazawa, Karazeev,
  Kassebaum, Kovyrshin, Krishnan, Krsulich, Kus, LaRose, Lambert, Latone,
  Lawrence, Liu, Liu, Mac, Maeng, Malyshev, Marecek, Marques, Mathews, Matsuo,
  McClure, McGarry, McKay, Meesala, Mezzacapo, Midha, Minev, Mooring, Morales,
  Moran, Murali, M{\"u}ggenburg, Nadlinger, Nannicini, Nation, Naveh,
  Nick-Singstock, Niroula, Norlen, O'Riordan, Ollitrault, Oud, Padilha, Paik,
  Perriello, Phan, Pistoia, Pozas-iKerstjens, Prutyanov, P{\'e}rez, Quintiii,
  Raymond, Redondo, Reuter, Rodr{\'\i}guez, Ryu, Sandberg, Sathaye, Schmitt,
  Schnabel, Scholten, Schoute, Sertage, Shammah, Shi, Silva, Siraichi,
  Sivarajah, Smolin, Soeken, Steenken, Stypulkoski, Takahashi, Taylor, Taylour,
  Thomas, Tillet, Tod, de~la Torre, Trabing, Treinish, TrishaPe, Turner,
  Vaknin, Valcarce, Varchon, Vogt-Lee, Vuillot, Weaver, Wieczorek, Wildstrom,
  Wille, Winston, Woehr, Woerner, Woo, Wood, Wood, Wood, Wootton, Yeralin, Yu,
  and Zdanski]{Qiskit}
H{\'e}ctor Abraham, Ismail~Yunus Akhalwaya, Gadi Aleksandrowicz, et~al.
\newblock Qiskit: An open-source framework for quantum computing, 2019.
\newblock URL \url{https://qiskit.org}.

\bibitem[Amy and Mosca(2019)]{AMY2019Ma}
Matthew Amy and Michele Mosca.
\newblock {T-Count} optimization and {R}eed-{M}uller codes.
\newblock \emph{IEEE Transactions on Information Theory}, 65\penalty0
  (8):\penalty0 4771--4784, 2019.
\newblock \doi{10.1109/TIT.2019.2906374}.

\bibitem[Amy et~al.(2014)Amy, Maslov, and Mosca]{AMY2014MMa}
Matthew Amy, Dmitri Maslov, and Michele Mosca.
\newblock Polynomial-time {T}-depth optimization of {C}lifford+{T} circuits via
  matroid partitioning.
\newblock \emph{IEEE Transactions on Computer-Aided Design of Integrated
  Circuits and Systems}, 33\penalty0 (10):\penalty0 1476--1489, 2014.
\newblock \doi{10.1109/TCAD.2014.2341953}.

\bibitem[Berlekamp(1980)]{BER1980a}
Elwyn~R. Berlekamp.
\newblock The technology of error-correcting codes.
\newblock \emph{Proceedings of the IEEE}, 68\penalty0 (5), 1980.
\newblock \doi{10.1109/PROC.1980.11696}.

\bibitem[Berry and Childs(2012)]{BER2012Ca}
Dominic.~W. Berry and Andrew~M. Childs.
\newblock Black-box {H}amiltonian simulation and unitary implementation.
\newblock \emph{Quantum Information \& Computation}, 12\penalty0
  (1\&{}2):\penalty0 29--62, 2012.

\bibitem[Berry et~al.(2015)Berry, Childs, Cleve, Kothari, and
  Somma]{BER2015CCKa}
Dominic~W. Berry, Andrew~M. Childs, Richard Cleve, Robin Kothari, and
  Rolando~D. Somma.
\newblock Simulating {H}amiltonian dynamics with a truncated {T}aylor series.
\newblock \emph{Physical Review Letters}, 114\penalty0 (9):\penalty0 090502,
  2015.
\newblock \doi{10.1103/PhysRevLett.114.090502}.

\bibitem[Bollob{\'a}s(2013)]{BOL2013a}
B{\'e}la Bollob{\'a}s.
\newblock \emph{Modern graph theory}, volume 184 of \emph{Graduate texts in
  mathematics}.
\newblock Springer Science~\& Business Media, New York, USA, 2013.
\newblock \doi{10.1007/978-1-4612-0619-4}.

\bibitem[Bonet-Monroig et~al.(2019)Bonet-Monroig, Babbush, and
  O'Brien]{BON2019BBa}
Xavier Bonet-Monroig, Ryan Babbush, and Thomas~E O'Brien.
\newblock Nearly optimal measurement scheduling for partial tomography of
  quantum states.
\newblock arXiv:1908.05628, 2019.
\newblock URL \url{https://arxiv.org/abs/1908.05628}.

\bibitem[Bravyi and Maslov(2020)]{BRA2020Ma}
Sergey Bravyi and Dmitri Maslov.
\newblock Hadamard-free circuits expose the structure of the {C}lifford group.
\newblock arXiv:2003.09412, 2020.
\newblock URL \url{https://arxiv.org/abs/2003.09412}.

\bibitem[Bravyi and Kitaev(2002)]{BRA2002Ka}
Sergey~B. Bravyi and Alexei~{Yu}. Kitaev.
\newblock Fermionic quantum computation.
\newblock \emph{Annals of Physics}, 298\penalty0 (1):\penalty0 210--226, 2002.
\newblock \doi{10.1006/aphy.2002.6254}.

\bibitem[Childs and Wiebe(2012)]{CHI2012Wa}
Andrew~M. Childs and Nathan Wiebe.
\newblock {H}amiltonian simulation using linear combinations of unary
  operators.
\newblock \emph{Quantum Information \& Computation}, 12\penalty0 (11\&
  12):\penalty0 901--924, 2012.

\bibitem[Childs et~al.(2018)Childs, Maslov, Nam, Ross, and Su]{CHI2018MNRa}
Andrew~M. Childs, Dmitri Maslov, Yonseong Nam, Neil~J. Ross, and Yuan Su.
\newblock Toward the first quantum simulation with quantum speedup.
\newblock \emph{PNAS}, 115\penalty0 (38):\penalty0 9456--9461, 2018.
\newblock \doi{10.1073/pnas.1801723115}.

\bibitem[Choo et~al.(2020)Choo, Mezzacapo, and Carleo]{CHO2020MCa}
Kenny Choo, Antonio Mezzacapo, and Giuseppe Carleo.
\newblock Fermionic neural-network states for ab-initio electronic structure.
\newblock \emph{Nature Communications}, 11\penalty0 (1):\penalty0 2368, 2020.
\newblock \doi{10.1038/s41467-020-15724-9}.

\bibitem[Crawford et~al.(2019)Crawford, van Straaten, Wang, Parks, Campbell,
  and Brierley]{CRA2019SWPa}
Ophelia Crawford, Barnaby van Straaten, Daochen Wang, Thomas Parks, Earl
  Campbell, and Stephen Brierley.
\newblock Efficient quantum measurement of {P}auli operators in the presence of
  finite sampling error.
\newblock arXiv:1908.06942, 2019.
\newblock URL \url{https://arxiv.org/abs/1908.06942}.

\bibitem[Feynman(1982)]{FEY1982a}
Richard~P. Feynman.
\newblock Simulating physics with computers.
\newblock \emph{International Journal of Theoretical Physics}, 21\penalty0
  ({6--7}):\penalty0 467--488, 1982.
\newblock \doi{10.1007/BF02650179}.

\bibitem[Garc\'{i}a et~al.(2013)Garc\'{i}a, Markov, and Cross]{GAR2013MCa}
H\'{e}ctor~J. Garc\'{i}a, Igor~L. Markov, and Andrew~W. Cross.
\newblock Efficient inner-product algorithm for stabilizer states.
\newblock arXiv:1210.6646, 2013.
\newblock URL \url{https://arxiv.org/abs/1210.6646}.

\bibitem[Gokhale et~al.(2019)Gokhale, Angiuli, Ding, Gui, Tomesh, Suchara,
  Martonosi, and Chong]{GOK2019ADGa}
Pranav Gokhale, Olivia Angiuli, Yongshan Ding, Kaiwen Gui, Teague Tomesh,
  Martin Suchara, Margaret Martonosi, and Frederic~T. Chong.
\newblock Minimizing state preparations in variational quantum eigensolver by
  partitioning into commuting families.
\newblock arXiv:1907.13623, 2019.
\newblock URL \url{https://arxiv.org/abs/1907.13623}.

\bibitem[Gottesman(1998)]{GOT1998a}
Daniel Gottesman.
\newblock The {H}eisenberg representation of quantum computers.
\newblock In \emph{Proceedings of the 22nd International Colloquium on Group
  Theoretical Methods in Physics -- GROUP22 ICGTMP98}, pages 32--43, Cambridge,
  MA, 1998. International Press.
\newblock URL \url{https://arxiv.org/abs/quant-ph/9807006}.

\bibitem[Hagberg et~al.(2008)Hagberg, Schult, and Swart]{HAG2008SSa}
Aric~A. Hagberg, Daniel~A. Schult, and Pieter~J. Swart.
\newblock Exploring network structure, dynamics, and function using {NetworkX}.
\newblock In \emph{Proceedings of the 7th Python in Science Conference
  (SciPy2008)}, pages 11–--15, August 2008.

\bibitem[Horn and Johnson(1985)]{HOR1985Ja}
Roger~A. Horn and Charles~R. Johnson.
\newblock \emph{Matrix Analysis}.
\newblock Cambridge university press, 1985.
\newblock \doi{10.1017/CBO9780511810817}.

\bibitem[Jena et~al.(2019)Jena, Genin, and Mosca]{JEN2019GMa}
Andrew Jena, Scott Genin, and Michele Mosca.
\newblock Pauli partitioning with respect to gate sets.
\newblock arXiv:1907.07859, 2019.
\newblock URL \url{https://arxiv.org/abs/1907.07859}.

\bibitem[Jordan and Wigner(1928)]{JOR1928Wa}
Pascual Jordan and Eugene Wigner.
\newblock {\"{U}}ber das {P}aulische {{\"{A}}}quivalenzverbot.
\newblock \emph{Zeitschrift f{\"{u}}r Physik}, 47\penalty0 (9--10):\penalty0
  631--651, 1928.
\newblock \doi{10.1007/978-3-662-02781-3_9}.

\bibitem[Kandala et~al.(2017)Kandala, Mezzacapo, Temme, Takita, Brink, Chow,
  and Gambetta]{KAN2017MTTa}
Abhinav Kandala, Antonio Mezzacapo, Kristan Temme, Maika Takita, Markus Brink,
  Jerry~M. Chow, and Jay~M. Gambetta.
\newblock Hardware-efficient variational quantum eigensolver for small
  molecules and quantum magnets.
\newblock \emph{Nature}, 549\penalty0 (7671):\penalty0 242--246, 2017.
\newblock \doi{10.1038/nature23879}.

\bibitem[Lloyd(1996)]{LLO1996a}
Seth Lloyd.
\newblock Universal quantum simulators.
\newblock \emph{Science}, 273\penalty0 (5278):\penalty0 1073--1078, 1996.

\bibitem[Low and Chuang(2017)]{LOW2017Ca}
Guang~Hao Low and Isaac~L. Chuang.
\newblock Optimal {H}amiltonian simulation by quantum signal processing.
\newblock \emph{Physical Review Letters}, 118\penalty0 (1):\penalty0 010501,
  2017.
\newblock \doi{10.1103/PhysRevLett.118.010501}.

\bibitem[Nielsen and Chuang(2010)]{NIE2010Ca}
Michael~A. Nielsen and Isaac~L. Chuang.
\newblock \emph{Quantum Computation and Quantum Information}.
\newblock Cambridge University Press, 2010.
\newblock \doi{10.1017/CBO9780511976667}.

\bibitem[Patel et~al.(2008)Patel, Markov, and Hayes]{PAT2008MHa}
Ketan~N. Patel, Igor~L. Markov, and John~P. Hayes.
\newblock Optimal synthesis of linear reversible circuits.
\newblock \emph{Quantum Information \& Computation}, 8\penalty0
  (3\&4):\penalty0 282--294, 2008.

\bibitem[Sarkar and van~den Berg(2019)]{SAR2019Ba}
Rahul Sarkar and Ewout van~den Berg.
\newblock On sets of commuting and anticommuting {P}aulis.
\newblock arXiv:1909.08123, 2019.
\newblock URL \url{https://arxiv.org/abs/1909.08123}.

\bibitem[Suzuki(1991)]{SUZ1991a}
Masuo Suzuki.
\newblock General theory of fractal path integrals with applications to
  many-body theories and statistical physics.
\newblock \emph{Journal of Mathematical Physics}, 32\penalty0 (2):\penalty0
  400--407, February 1991.
\newblock \doi{10.1063/1.529425}.

\bibitem[Szabo and Ostlund(1989)]{SZA1989Oa}
Attila Szabo and Neil~S. Ostlund.
\newblock \emph{Modern Quantum Chemistry: Introduction to Advanced Electronic
  Structure Theory}.
\newblock McGraw-Hill, 1989.

\bibitem[Tranter et~al.(2018)Tranter, Love, Mintert, and Coveney]{TRA2018LMVa}
Andrew Tranter, Peter~J. Love, Florian Mintert, and Peter~V. Coveney.
\newblock A comparison of the {B}ravyi-{K}itaev and {J}ordan-{W}igner
  transformations for the quantum simulation of quantum chemistry.
\newblock \emph{Journal of Chemical Theory and Computation}, 14\penalty0
  (11):\penalty0 5617--5630, 2018.
\newblock \doi{10.1021/acs.jctc.8b00450}.

\bibitem[Trotter(1959)]{TRO1959a}
Hale~F. Trotter.
\newblock On the product of semi-groups of operators.
\newblock \emph{Proc. Am. Math. Soc.}, 10\penalty0 (4):\penalty0 545--551,
  1959.
\newblock \doi{10.2307/2033649}.

\bibitem[Verteletskyi et~al.(2020)Verteletskyi, Yen, and Izmaylov]{VER2020YIa}
Vladyslav Verteletskyi, Tzu-Ching Yen, and Artur~F. Izmaylov.
\newblock Measurement optimization in the variational quantum eigensolver using
  a minimum clique cover.
\newblock \emph{The Journal of Chemical Physics}, 152\penalty0 (12):\penalty0
  124114, 2020.
\newblock \doi{10.1063/1.5141458}.

\bibitem[Yen et~al.(2020)Yen, Verteletskyi, and Izmaylov]{YEN2020VIa}
Tzu-Ching Yen, Vladyslav Verteletskyi, and Artur~F. Izmaylov.
\newblock Measuring all compatible operators in one series of a single-qubit
  measurements using unitary transformations.
\newblock \emph{Journal of Chemical Theory and Computation}, 16\penalty0
  (4):\penalty0 2400--2409, 2020.
\newblock \doi{10.1021/acs.jctc.0c00008}.

\end{thebibliography}

\end{document}